%% file: full.tex
\documentclass[a4paper,USenglish]{llncs}

\usepackage{fullpage}
\usepackage{microtype}
\usepackage{thm-restate}
\usepackage{ntheorem}
\theoremstyle{plain}
\theoremsymbol{}
\theoremprework{}
\theorempostwork{}
\theoremseparator{.}

\newtheorem{observation}{Observation}
\renewtheorem{theorem}{Theorem}
\renewtheorem{lemma}{Lemma}
\renewtheorem{claim}{Claim}

\let\doendproof\endproof
\renewcommand\endproof{~\hfill\qed\doendproof}

\renewcommand{\subparagraph}[1]{\paragraph{#1}}

\usepackage{xcolor}
\usepackage{caption}
\usepackage{subcaption}

\usepackage{amsmath}
\usepackage{amssymb}

\usepackage[normalem]{ulem}
\usepackage{framed}

\def\lowhook#1{\lower.47em\hbox{$#1$}}
\def\LEorL{\mathrel{\lowhook\lhook\kern-.1em{\le}\kern-.08em\lowhook\rhook}}
\def\GEorG{\mathrel{\lowhook\lhook\kern-.08em{\ge}\kern-.1em\lowhook\rhook}}

\usepackage{dsfont}
\usepackage{xspace}
\usepackage{hyperref}
\usepackage{rotating}
\usepackage[noend]{algpseudocode}
\usepackage{stmaryrd}
\usepackage{stackengine}
\usepackage{scalerel}

\usepackage[ruled,vlined]{algorithm2e}
\SetKwInput{Input}{Input\rule{1.99ex}{0cm}}
\SetKwInput{Output}{Output}
\SetKwInput{Data}{Data}
\SetKwFor{Los}{}{}{}

\renewcommand{\paragraph}[1]{\smallskip\noindent\textbf{#1}\xspace}
\graphicspath{{../figures/}}

\usepackage[nameinlink,capitalise]{cleveref}
\Crefname{observation}{Observation}{Observations}
\Crefname{property}{Property}{Properties}
\Crefname{algorithm}{Algorithm}{Algorithms}
\Crefname{section}{Section}{Sections}
\Crefname{observation}{Observation}{Observations}
\Crefname{lemma}{Lemma}{Lemmas}
\Crefname{claim}{Claim}{Claims}
\Crefname{figure}{Fig.}{Figs.}
\Crefname{figure}{Fig.}{Figs.}
\Crefname{enumi}{Condition}{Conditions}
\Crefname{statement}{Statement}{Statements}

\newcommand{\UBE}{$2$UBE\xspace}
\newcommand{\UBEs}{$2$UBEs\xspace}
\newcommand{\HPC}{HP-completion\xspace}
\newcommand{\KUBE}{{\scshape $k$UBE Testing}\xspace}
\newcommand{\TUBE}{{\scshape 3UBE Testing}\xspace}
\newcommand{\TWOUBE}{{\scshape $2$UBE Testing}\xspace}
\newcommand{\BETW}{{\scshape Betweenness}\xspace}
\newcommand{\shell}{shell digraph}
\newcommand{\filled}{filled \shell}
\newcommand{\gadgeted}{$\Lambda$-\filled}

\usepackage{lipsum}
\usepackage{dsfont}
\usepackage{verbatim}

\usepackage{xspace}
\usepackage{paralist}
\usepackage{dsfont}
\usepackage{amssymb}
\usepackage{gensymb}
\usepackage{wrapfig}

\usepackage{color}
\definecolor{realblue}{rgb}{0,0,1}
\definecolor{blue}{rgb}{0.274,0.392,0.666}
\definecolor{darkerblue}{rgb}{0.094,0.455,0.804}
\definecolor{darkblue}{rgb}{0.063,0.306,0.545}
\definecolor{red}{rgb}{0.627,0.117,0.156}
\definecolor{green}{rgb}{0,0.588,0.509}
\definecolor{orange}{rgb}{0.903,0.739,0.382}
\definecolor{realred}{rgb}{1,0,0}
\definecolor{ourgreen}{rgb}{0,0.588,0.509}

\definecolor{lipicsblue}{rgb}{0.08235294118,0.3098039216,0.537254902}
\definecolor{ourred}{rgb}{1,0.3,0.3}
\definecolor{darkgreen}{rgb}{0, 0.7, 0}

  \DeclareMathOperator{\midset}{mid}
  \DeclareMathOperator{\skel}{skel}

  \DeclareMathOperator{\mypert}{pert}

  \newcommand{\pert}[1]{\mypert(#1)}

  \newcommand{\bwsymb}{\ensuremath{\beta}}
  \DeclareMathOperator{\bw}{\bwsymb}

\newcommand{\darkblue}[1]{{{\textcolor{darkblue}{#1}\xspace}}}

\renewcommand{\emph}[1]{\darkblue{\em #1}}

\hypersetup{colorlinks,breaklinks,colorlinks=true,urlcolor=lipicsblue,citecolor=lipicsblue,linkcolor=lipicsblue}

\newcommand{\LBR}{\mbox{\raisebox{-2pt}{\includegraphics[page=1,height=10pt,width=7pt]{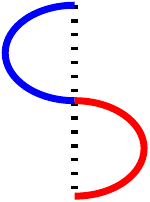}}\hspace{1pt}-$\langle L, B, R \rangle$}\xspace}
\newcommand{\RBL}{\mbox{\raisebox{-2pt}{\includegraphics[page=2,height=10pt,width=7pt]{icons.pdf}}\hspace{1pt}-$\langle R, B, L \rangle$}\xspace}
\newcommand{\RBR}{\mbox{\raisebox{-2pt}{\includegraphics[page=3,height=10pt,width=7pt]{icons.pdf}}\hspace{1pt}-$\langle R, B, R \rangle$}\xspace}
\newcommand{\LBL}{\mbox{\raisebox{-2pt}{\includegraphics[page=4,height=10pt,width=7pt]{icons.pdf}}\hspace{1pt}-$\langle L, B, L \rangle$}\xspace}
\newcommand{\RRR}{\mbox{\raisebox{-2pt}{\includegraphics[page=5,height=10pt,width=5pt]{icons.pdf}}\hspace{1pt}-$\langle R, R, R \rangle$}\xspace}
\newcommand{\QRRR}{\mbox{\raisebox{-2pt}{\includegraphics[page=19,height=10pt,width=5pt]{icons.pdf}}\hspace{1pt}-$\langle R, R, R \rangle$}\xspace}
\newcommand{\LLL}{\mbox{\raisebox{-2pt}{\includegraphics[page=6,height=10pt,width=5pt]{icons.pdf}}\hspace{1pt}-$\langle L, L, L \rangle$}\xspace}
\newcommand{\QLLL}{\mbox{\raisebox{-2pt}{\includegraphics[page=20,height=10pt,width=5pt]{icons.pdf}}\hspace{1pt}-$\langle L, L, L \rangle$}\xspace}
\newcommand{\NBR}{\mbox{\raisebox{-2pt}{\includegraphics[page=7,height=10pt,width=7pt]{icons.pdf}}\hspace{1pt}-$\langle N, B, R \rangle$}\xspace}
\newcommand{\NBL}{\mbox{\raisebox{-2pt}{\includegraphics[page=8,height=10pt,width=7pt]{icons.pdf}}\hspace{1pt}-$\langle N, B, L \rangle$}\xspace}
\newcommand{\RBN}{\mbox{\raisebox{-2pt}{\includegraphics[page=9,height=10pt,width=7pt]{icons.pdf}}\hspace{1pt}-$\langle R, B, N \rangle$}\xspace}
\newcommand{\LBN}{\mbox{\raisebox{-2pt}{\includegraphics[page=10,height=10pt,width=7pt]{icons.pdf}}\hspace{1pt}-$\langle L, B, N \rangle$}\xspace}
\newcommand{\NRR}{\mbox{\raisebox{-2pt}{\includegraphics[page=11,height=10pt,width=7pt]{icons.pdf}}\hspace{1pt}-$\langle N, R, R \rangle$}\xspace}
\newcommand{\NLL}{\mbox{\raisebox{-2pt}{\includegraphics[page=12,height=10pt,width=7pt]{icons.pdf}}\hspace{1pt}-$\langle N, L, L \rangle$}\xspace}
\newcommand{\RRN}{\mbox{\raisebox{-2pt}{\includegraphics[page=13,height=10pt,width=7pt]{icons.pdf}}\hspace{1pt}-$\langle R, R, N \rangle$}\xspace}
\newcommand{\LLN}{\mbox{\raisebox{-2pt}{\includegraphics[page=14,height=10pt,width=7pt]{icons.pdf}}\hspace{1pt}-$\langle L, L, N \rangle$}\xspace}
\newcommand{\NBN}{\mbox{\raisebox{-2pt}{\includegraphics[page=15,height=10pt,width=7pt]{icons.pdf}}\hspace{1pt}-$\langle N, B, N \rangle$}\xspace}
\newcommand{\NRN}{\mbox{\raisebox{-2pt}{\includegraphics[page=16,height=10pt,width=7pt]{icons.pdf}}\hspace{1pt}-$\langle N, R, N \rangle$}\xspace}

\newcommand{\NNN}{\mbox{\raisebox{-2pt}{\includegraphics[page=18,height=10pt,width=7pt]{icons.pdf}}\hspace{1pt}-$\langle N, N, N \rangle$}\xspace}

\begin{document}

\newcommand{\perugia}{$^1$}
\newcommand{\rome}{$^2$}
\newcommand{\kit}{$^3$}

\newcommand{\acks}{
  This work was supported in part 
  by project ``Algoritmi e sistemi di analisi visuale di reti complesse e di grandi dimensioni'' -- Ricerca di Base 2018, Dipartimento di Ingegneria dell'Università degli Studi di Perugia (Binucci, Di Giacomo, and Didimo),
  and in part
  by MIUR Project ``MODE'' under PRIN 20157EFM5C, 
  by MIUR Project ``AHeAD'' under PRIN 20174LF3T8, 
  by MIUR-DAAD JMP N$^\circ$ 34120, 
  by H2020-MSCA-RISE project 734922 -- ``CONNECT'',
  and by Roma Tre University Azione 4 Project ``GeoView'' ({Da Lozzo} and Patrignani).\xspace}

\title{Upward Book Embeddings of st-Graphs\thanks{\acks}}

\author{
Carla Binucci\perugia,
{Giordano {Da Lozzo}}\rome, 
Emilio {Di Giacomo}\perugia,\\
Walter Didimo\perugia,
Tamara Mchedlidze\kit,
Maurizio Patrignani\rome
}
\institute{
Universit\`a degli Studi di Perugia, Perugia, Italy\\
\href{mailto:carla.binucci@unipg.it,emilio.digiacomo@unipg.it,walter.didimo@unipg.it}{\{carla.binucci,emilio.digiacomo,walter.didimo\}@unipg.it}\\
\smallskip
Karlsruhe Institute of Technology, Karlsruhe, Germany\\
\href{mailto:mched@iti.uka.de}{mched@iti.uka.de}\\
\smallskip
Roma Tre University, Rome, Italy\\
\href{mailto:giordano.dalozzo@uniroma3.it,maurizio.patrignani@uniroma3.it}{\{giordano.dalozzo,maurizio.patrignani\}@uniroma3.it}
}

\maketitle

\begin{abstract} 
We study \emph{$k$-page upward book embeddings} ($k$UBEs) of $st$-graphs, that is, book embeddings of single-source single-sink directed acyclic graphs on $k$ pages with the additional requirement that the vertices of the graph appear in a topological ordering along the spine of the book.
We show that testing whether a graph admits a $k$UBE is NP-complete for $k\geq 3$. A hardness result for this problem was previously known only for $k = 6$ \href{https://epubs.siam.org/doi/10.1137/S0097539795291550}{[Heath and Pemmaraju, 1999]}.
Motivated by this negative result, we focus our attention on $k=2$. On the algorithmic side, we present polynomial-time algorithms for testing the existence of \UBEs of planar \mbox{$st$-graphs} with branchwidth~$\bwsymb$ and of plane $st$-graphs whose faces have a special structure. These algorithms run in $O(f(\bwsymb)\cdot n+n^3)$ time and $O(n)$ time, respectively, where $f$ is a singly-exponential function on $\bwsymb$. Moreover, on the combinatorial side, we present two notable families of plane \mbox{$st$-graphs} that always admit an embedding-preserving~\UBE.
\end{abstract}

\section{Introduction}\label{se:introduction}
A \emph{$k$-page book embedding} $\langle \pi, \sigma \rangle$ of an undirected graph $G=(V,E)$ consists of a vertex ordering \mbox{$\pi: V \leftrightarrow \{1,2,\dots,|V|\}$} and of an assignment $\sigma: E \rightarrow \{1,\dots,k\}$ of the edges of $G$ to one of $k$ sets, called \emph{pages}, so that for any two edges $(a,b)$ and $(c,d)$ in the same page, with $\pi(a) < \pi(b)$ and $\pi(c) < \pi(d)$, we have neither $\pi(a) < \pi(c)$ $< \pi(b) < \pi(d)$ nor $\pi(c) < \pi(a) $ $< \pi(d) < \pi(b)$. 
From a geometric perspective, a $k$-page book embedding can be associated with a \emph{canonical drawing} $\Gamma(\pi,\sigma)$ of $G$ 
where the $k$ pages correspond to $k$ half-planes sharing a vertical line, called the \emph{spine}. Each vertex $v$ is a point on the spine with $y$-coordinate $\pi(v)$; each edge $e$ is a circular arc on the $\sigma(e)$-th page, and the edges in the same page do not cross.

For $k$-page book embeddings of directed graphs (digraphs), a typical requirement is that all the edges are oriented in the upward direction. This implies that $G$ is acyclic and that all the vertices appear along the spine in a topological ordering.
This type of book embedding for digraphs is called an \emph{upward $k$-page book embedding} of $G$ (for short, \emph{$k$UBE}).
Note that, when $k=2$ and the two pages are coplanar, drawing $\Gamma(\pi,\sigma)$ is an \emph{upward planar drawing} of $G$, i.e., a planar drawing where all the edges monotonically increase in the upward direction. The study of upward planar drawings is a most prolific topic in the theory of graph visualization~\cite{DBLP:journals/talg/AngeliniLBDKRR18,DBLP:journals/algorithmica/AngeliniLBF17,DBLP:journals/algorithmica/BertolazziBD02,DBLP:journals/siamcomp/BertolazziBMT98,DBLP:journals/cj/BinucciD16,DBLP:journals/comgeo/Brandenburg14,DBLP:conf/gd/ChaplickCCLNPTW17,DBLP:conf/gd/LozzoBFPR18,DBLP:journals/tcs/BattistaT88,DBLP:journals/dcg/BattistaTT92,DBLP:journals/siamcomp/GargT01,DBLP:journals/cj/RextinH17}.

The \emph{page number} of a (di)graph $G$ (also called \emph{book thickness}) is the minimum number~$k$ such that $G$ admits a (upward) $k$-page book embedding. Computing the page number of directed and undirected graphs is a widely studied problem, which finds applications in a variety of domains, including VLSI design, fault-tolerant processing, parallel process scheduling, sorting networks, parallel matrix computations~\cite{Chung87,HeathLR92,pemmaraju92}, computational origami~\cite{AkitayaDHL17}, and graph drawing~\cite{DBLP:journals/jgaa/BiedlSWW99,GiacomoDLW06,GiordanoLMSW15,DBLP:conf/gd/Wood01}. See~\cite{DBLP:journals/dmtcs/DujmovicW04} for additional references.

\subparagraph{Book embeddings of undirected graphs.} Seminal results on book embeddings of undirected graphs are described in the paper of    
Bernhart and Kainen~\cite{BERNHART1979320}. They prove that the graphs with page number one are exactly the outerplanar graphs, while graphs with page number two are 
the sub-Hamiltonian graphs.
This second result implies that it is NP-complete to decide whether a graph admits a 2-page book embedding~\cite{Wigderson82}. Yannakakis~\cite{Yannakakis89} proved that every planar graph has a 4-page book embedding, while the fascinating question whether the page number of planar graphs can be reduced to three is still open. The aforementioned works have inspired several papers about the page number of specific families of undirected graphs (e.g.,~\cite{DBLP:journals/algorithmica/BekosBKR17,BekosGR16,Chung87,EnomotoNO97}) and about the relationship between the page number and other graph parameters (e.g.,~\cite{DujmovicW05,DBLP:journals/dam/GanleyH01,DBLP:journals/jal/Malitz94a,DBLP:journals/jal/Malitz94}).  
Different authors studied constrained versions of $k$-page book embeddings where either the vertex ordering $\pi$ is (partially) fixed~\cite{DBLP:journals/jgaa/AngeliniLBFPR17,DBLP:journals/heuristics/Cimikowski06,DBLP:journals/tc/MasudaNKF90,DBLP:conf/stacs/Unger88,DBLP:conf/stacs/Unger92} or the page assignment $\sigma$ for the edges is given~\cite{DBLP:conf/gd/AngeliniBB12,DBLP:journals/tcs/AngeliniLN15,DBLP:journals/jda/AngeliniBFPR12,DBLP:journals/tcs/HongN18}. 
Relaxed versions of book embeddings where edge crossings are allowed (called \emph{$k$-page drawings}) or where edges can cross the spine (called \emph{topological book embeddings}) have also been considered (e.g.,~\cite{DBLP:conf/compgeom/AbregoAFRS12,DBLP:conf/gd/BannisterE14,DBLP:journals/ejc/BinucciGHL18,DBLP:journals/comgeo/CardinalHKTW18,DBLP:journals/comgeo/GiacomoDLW05,em-egtpb-99,emo-lbneo-99}). Finally, 2-page (topological) book embeddings find applications to point-set embedding and universal point set (e.g.,~\cite{DBLP:journals/jgaa/AngeliniEFKLMTW14,DBLP:journals/tcs/BadentGL08,DBLP:journals/ijfcs/GiacomoLT06,DBLP:journals/algorithmica/GiacomoLT10,DBLP:journals/dcg/EverettLLW10,DBLP:conf/gd/LofflerT15}).

\subparagraph{Book embeddings of directed graphs.}
As for undirected graphs, there are many papers devoted to the study of upper and lower bounds on the page number of directed graphs.   
Heath et al.~\cite{HeathPT99} show that directed trees and unicyclic digraphs have page number one and two, respectively.  Alzohairi and Rival~\cite{Alzohairi97}, and later Di Giacomo et al.~\cite{GiacomoDLW06} with an improved linear-time construction, show that series-parallel digraphs have page number two. Mchedlidze and Symvonis~\cite{crossingFreeHpCompl09} generalize this result and prove that $N$-free upward planar digraphs, which contain series-parallel digraphs, also have page number two  (a digraph is upward planar if it admits an upward planar drawing). 
Frati et al.~\cite{FratiFR13} give several conditions under which upward planar triangulations have bounded page number. Overall, the question asked by Nowakowski and Parker~\cite{NowakowskiP89} almost 30 years ago, of whether the page number of upward planar digraphs is bounded, remains open. Several works study the page number of acyclic digraphs in terms of posets, i.e., the page number of their Hasse diagram (e.g.,~\cite{AlhashemJZ15,NowakowskiP89}). 

About the lower bounds, Nowakowski and Parker~\cite{NowakowskiP89} give an example of a \emph{planar $st$-graph} that requires three pages for an upward book embedding (see \cref{fi:no2ube}). A planar $st$-graph is an upward planar digraph with a single source $s$ and a single sink $t$. Hung~\cite{Hung89} shows an upward planar digraph with page number four, while Heath and Pemmaraju~\cite{HeathP97} describe an acyclic planar digraph (which is not upward planar) requiring $\lfloor n/2\rfloor$ pages. Syslo~\cite{Syslo89} provides a lower bound on the page number of a poset in terms of its bump number. 

Besides the study of upper and lower bounds on the page number of digraphs, several papers concentrate on the design of testing algorithms for the existence of $k$UBEs. The problem is NP-complete for $k=6$~\cite{HeathP99}. For $k=2$, Mchedlidze and Symvonis~\cite{MchedlidzeS11} give linear-time testing algorithms for outerplanar and planar triangulated $st$-graphs.
An $O(w^2n^w)$-time testing algorithm for $2$UBEs of planar $st$-graphs whose width is $w$ is given in~\cite{crossingFreeHpCompl09}, where the \emph{width} is the minimum number of directed paths that cover all the vertices. Heath and Pemmaraju~\cite{HeathP99} describe a linear-time algorithm to recognize digraphs that admit $1$UBEs. 

Finally, as for the undirected case, constrained or relaxed variants of $k$UBEs for digraphs are studied \cite{AkitayaDHL17,GiacomoGL11,GiordanoLMSW15}, as well as applications to the point-set embedding problem \cite{GiacomoDLW06,GiordanoLMSW15}.

\subparagraph{Contribution.} Our paper is motivated by the gap present in the literature about the computation of upward book embeddings of digraphs: Polynomial-time algorithms are known only for one page or for two pages and subclasses of planar digraphs, while NP-completeness is known only for exactly $6$ pages. 
We shrink this gap and address the research direction proposed by Heath and Pemmaraju~\cite{HeathP99}: Identification of graph classes for which the existence of $k$UBEs can be solved efficiently. Our results are as follows:

\begin{itemize}
	\item We prove that testing whether a digraph $G$ admits a $k$UBE is NP-complete for every $k \geq 3$, even if $G$ is an $st$-graph (\cref{se:complexity}).	An analogous result was previously known only for the constrained version in which the page assignment is given~\cite{AkitayaDHL17}.  
	
	\item We describe another meaningful subclass of upward planar digraphs that admit a \UBE (\cref{se:existential}). This class is structurally different from the $N$-free upward planar digraphs, the largest class of upward $2$-page book embeddable digraphs previously known.
    
    \item We give algorithms to test the existence of a \UBE for notable families of planar $st$-graphs. First, we give a linear-time algorithm for plane $st$-graphs whose faces have a special structure (\cref{se:testing-special}). Then, we describe an $O(f(\bw)\cdot n + n^3)$-time algorithm for $n$-vertex planar $st$-graphs of branchwidth $\bw$, where $f$ is a singly-exponential function (\cref{se:testing}). 
    The algorithm works for both variable and fixed embedding. This result also implies a sub-exponential-time algorithm for general planar $st$-graphs. 
\end{itemize}    

Full details for omitted or sketched proofs can be found in the Appendix.

\section{Preliminaries}\label{se:preliminaries}

We assume familiarity with basic definitions on graph connectivity and planarity (see, e.g.,~\cite{dett-gd-99} and \cref{se:app-prel}). We only consider  (di)graphs without loops and multiple edges, and we denote by $V(G)$ and $E(G)$ the sets of vertices and edges of a (di)graph $G$.  

A digraph $G$ is a \emph{planar $st$-graph} if and only if: (i) it is acyclic; (ii) it has a single source $s$ and a single sink $t$; and (iii) it admits a planar embedding $\mathcal E$ with $s$ and $t$ on the outer face. A graph $G$ together with $\mathcal E$ is a \emph{planar embedded $st$-graph}, also called a \emph{plane $st$-graph}.

Let $G$ be a plane $st$-graph and let $e=(u,v)$ be an edge of $G$. The \emph{left face} (resp. \emph{right face}) of $e$ is the face to the left (resp. right) of $e$ while moving from $u$ to $v$. The boundary of every face $f$ of $G$ consists of two directed paths $p_l$ and $p_r$ from a common source $s_f$ to a common sink $t_f$. The paths $p_l$ and $p_r$ are the \emph{left path} and the \emph{right path} of $f$, respectively. The vertices $s_f$ and $t_f$ are the \emph{source} and the \emph{sink} of $f$, respectively. If $f$ is the outer face, $p_l$ (resp. $p_r$) consists of the edges for which $f$ is the left face (resp. right face); in this case $p_l$ and $p_r$ are also called the \emph{left boundary} and the \emph{right boundary} of $G$, respectively. 
If $f$ is an internal face, $p_l$ (resp. $p_r$) consists of the edges for which $f$ is the right face (resp. left face).      

The \emph{dual graph} $G^*$ of a plane $st$-graph $G$ is a plane $st$-graph (possibly with multiple edges) such that: 
\begin{inparaenum}[(i)]
\item $G^*$ has a vertex associated with each internal face of $G$ and two vertices $s^*$~and~$t^*$ associated with the outer face of $G$, that are the source and the sink of $G^*$, respectively; 
\item for each internal edge $e$ of $G$, $G^*$ has a dual edge from the left to the right face of $e$; 
\item for each edge $e$ in the left boundary of $G$, there is an edge from $s^*$ to the right face of $e$; (v) for each edge $e$ in the right boundary of $G$, there is an edge from the left face of $e$ to $t^*$.     
\end{inparaenum}

Consider a planar $st$-graph $G$ and let $\overline{G}$ be a planar $st$-graph obtained by augmenting $G$ with directed edges in such a way that it contains a directed Hamiltonian $st$-path $P_{\overline{G}}$. The graph $\overline{G}$ is an \emph{\HPC} of $G$.
Consider now a plane $st$-graph $G$ and let $\mathcal E$ be a planar embedding of $G$. 
Let $\overline{G}$ be an embedded {\HPC} of $G$ whose embedding $\overline{\mathcal E}$ is such that its restriction to $G$ is $\mathcal{E}$. We say that $\overline{G}$ is an \emph{embedding-preserving} \emph{\HPC} of $G$.

Bernhart and Kainen~\cite{BERNHART1979320} prove that an undirected planar graph admits a $2$-page book embedding if and only if it is \emph{sub-Hamiltonian}, i.e., it can be made Hamiltonian by adding edges while preserving its planarity. \cref{th:2ube-hpc-preserving} is an immediate consequence of the result in~\cite{BERNHART1979320} for planar digraphs (see also~\cref{fi:2ube-hpc}); when we say that a \UBE $\langle \pi,\sigma \rangle$ is \emph{embedding-preserving} we mean that the drawing $\Gamma(\pi,\sigma)$ preserves the planar embedding of $G$. 

\begin{theorem}\label{th:2ube-hpc-preserving}
	A planar (plane) $st$-graph $G$ admits a (embedding-preserving) \UBE $\langle \pi,\sigma \rangle$ if and only if $G$ admits a (embedding-preserving) \HPC $\overline{G}$. Also, the order $\pi$ coincides with the order of the vertices along $P_{\overline{G}}$.
\end{theorem}

\begin{figure}[t]
	\centering
	\begin{subfigure}{.18\textwidth}
		\centering
		\includegraphics[width=\columnwidth, page=1]{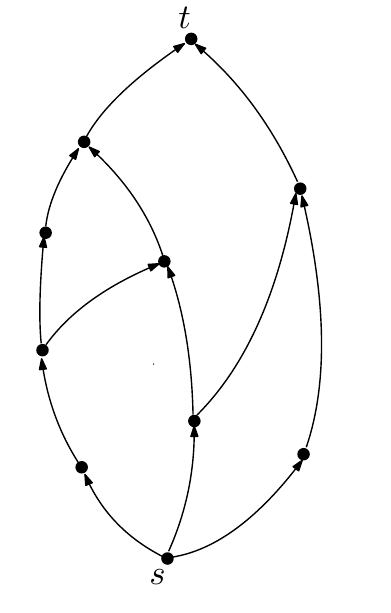}
		\subcaption{}
		\label{fi:2ube-hpc.1}
	\end{subfigure}
	\begin{subfigure}{.36\textwidth}
		\centering
		\includegraphics[width=\columnwidth, page=4]{2ube-hpc}
		\subcaption{}
		\label{fi:2ube-hpc.4}
	\end{subfigure}
	\begin{subfigure}{.18\textwidth}
		\centering
		\includegraphics[width=\columnwidth, page=2]{2ube-hpc}
		\subcaption{}
		\label{fi:2ube-hpc.2}
	\end{subfigure}
	\begin{subfigure}{.18\textwidth}
		\centering
		\includegraphics[width=\columnwidth, page=3]{2ube-hpc}
		\subcaption{}
		\label{fi:2ube-hpc.3}
	\end{subfigure}  
	\caption{\small{(a) A plane $st$-graph $G$. (b) The dual of $G$ is shown in gray. (c) An embedding-preserving \HPC of~$G$. (d)  An embedding-preserving \UBE $\Gamma$ of $G$ corresponding to (c). }}\label{fi:2ube-hpc}
\end{figure}

\section{NP-Completeness for $\mathbf{k}$UBE ($\mathbf{k \geq 3}$)}\label{se:complexity}
We prove that the \KUBE problem of deciding whether a digraph $G$ admits an upward $k$-page book embedding is NP-complete for each fixed $k \geq 3$.  
The proof uses a reduction from the \BETW problem~\cite{Opatrny79}.

\begin{quote}
	\textsc{Betweenness}
	
	\textit{Instance:}  
	\begin{minipage}[t]{0.85\linewidth}
		A finite set $S$ of elements and a set $R \subseteq S \times S \times S$ of triplets.  
	\end{minipage}
	
	\textit{Question:} 
	\begin{minipage}[t]{0.85\linewidth}
		Does there exist an ordering $\tau: S \rightarrow \mathbb{N}$ of the elements of $S$ such that for any element $(a,b,c) \in R$ either $\tau(a) < \tau(b) <\tau(c)$ or $\tau(c) < \tau(b) < \tau(a)$?	
	\end{minipage} 
\end{quote}

We incrementally define a set of families of digraphs and prove some properties of these digraphs. Then, we use the digraphs of these families to reduce a generic instance of \BETW to an instance of \TUBE, thus proving the hardness result for $k=3$. We then explain how the proof can be easily adapted to work for $k>3$.  

For a digraph $G$, we denote by $u \leadsto v$ a directed path from a vertex $u$ to a vertex $v$ in $G$. Let $\gamma=\langle \pi, \sigma \rangle$ be a 3UBE of $G$. Two edges $(u,v)$ and $(w,z)$ of $G$ \emph{conflict} if either $\pi(u) < \pi(w) < \pi(v) < \pi(z)$ or $\pi(w) < \pi(u) < \pi(z) < \pi(v)$. Two conflicting edges cannot be assigned to the same page. The next property will be used in the following; it is immediate from the definition of book embedding and from the pigeonhole principle.

\begin{property}\label{pr:conflicting}
	In a 3UBE there cannot exist $4$ edges that mutually conflict.
\end{property} 

\subparagraph{Shell digraphs.}  
The first family that we define are the \emph{\shell{}s}, recursively defined as follows. Digraph $G_0$, depicted in \cref{fi:hardness-a}, consists of a directed path $P$ with $8$ vertices denoted as $s_0$, $q_0$, $p_{-1}$,  $t_{-1}$, $s'_0$, $q'_0$, $t'_0$, and $p_0$ in the order they appear along $P$. Besides the edges of $P$, the following directed edges exists in $G_0$: $(s_0,s'_0)$, $(q_0,q'_0)$, $(t_{-1},p_0)$. Finally, there is a vertex $t_0$ connected to $P$ by means of the two directed edges $(p_{-1},t_0)$ and $(t'_0,t_0)$. Graph $G_h$ is obtained from $G_{h-1}$ with additional vertices and edges as shown in  \cref{fi:hardness-b}. A new directed path of two vertices $s_h$ and $q_h$ is connected to $G_{h-1}$ with the edge $(q_h,s_{h-1})$; a second path of four vertices $s'_{h}$, $q'_{h}$, $t'_{h}$, and $p_h$ is connected to $G_h$ with the edge $(t_{h-1},s'_h)$. The following edges exist between these new vertices: $(s_h,s'_h)$, $(q_h,q'_h)$, $(t_{h-1},p_h)$. Finally, there is a vertex $t_h$ connected to the other vertices by means of the two directed edges $(p_{h-1},t_h)$ and $(t'_h,t_h)$. For any $h \geq 0$, the edges $(s_h,s'_h)$ and $(q_h,q'_h)$ are called the \emph{forcing edges} of $G_h$; the edges  $(p_{h-1},t_h)$ and $(t_{h-1},p_h)$ are the \emph{channel edges} of $G_h$; the edge $(t'_h,t_h)$ is the \emph{closing edge} of $G_h$. The vertices and edges of $G_h \setminus G_{h-1}$ are the \emph{exclusive vertices and edges} of $G_h$. The following lemma establishes some basic properties of the \shell{}s. 

\begin{figure}[!t]
	\centering
	\begin{subfigure}{.4\textwidth}
		\centering
		\includegraphics[width=\columnwidth, page=1]{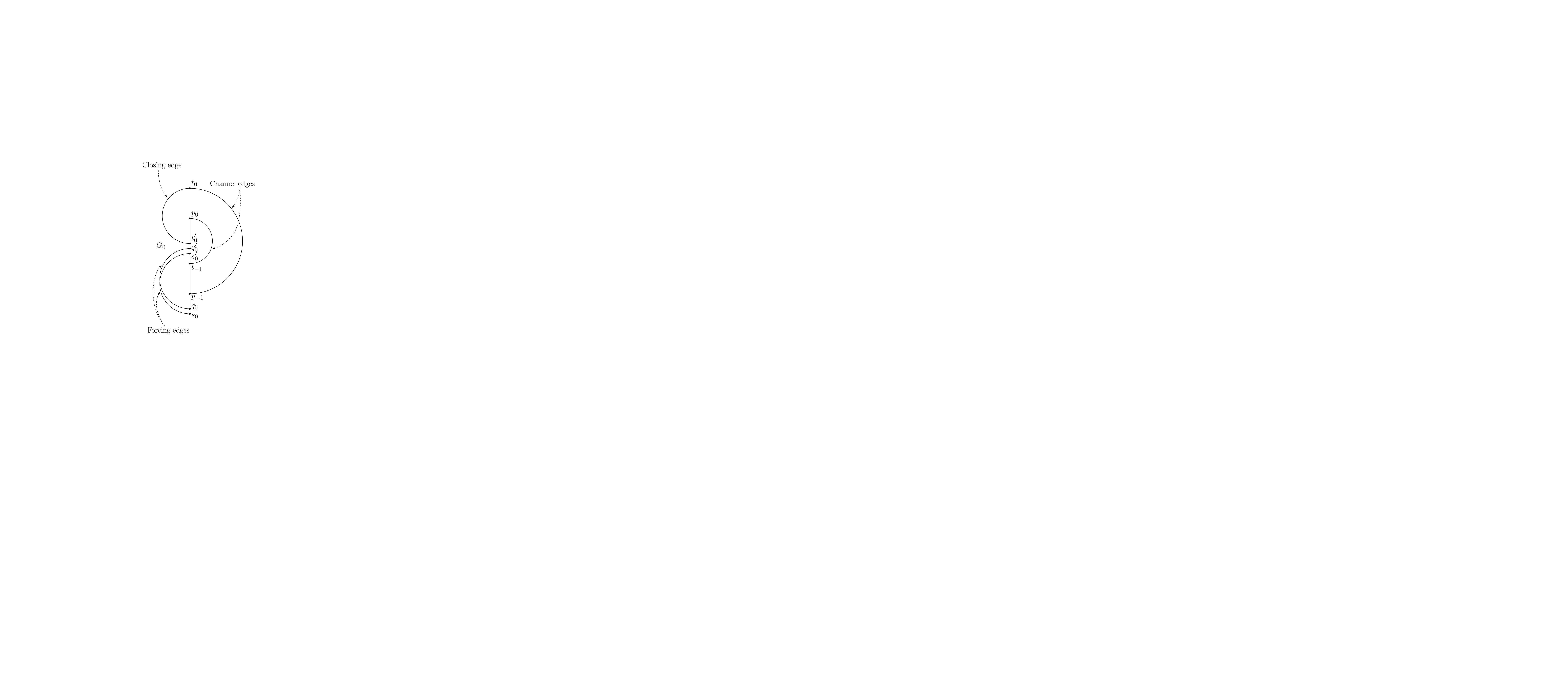}
		\subcaption{$G_0$}
		\label{fi:hardness-a}
	\end{subfigure}
	\hfil
	\begin{subfigure}{.4\textwidth}
		\centering
		\includegraphics[width=\columnwidth, page=2]{hardness}
		\subcaption{$G_k$}
		\label{fi:hardness-b}
	\end{subfigure}
	\caption{\small{Definition of \shell{}s. Edges are oriented from bottom to top. }}\label{fi:hardness}
\end{figure}

\begin{restatable}{lemma}{leshell}\label{le:shell} Every \shell{} $G_h$ for $h \geq 0$ admits a 3UBE. In any 3UBE $\gamma=\langle \pi, \sigma \rangle$ of $G_h$ the following conditions hold for every $i=0,1,\dots,h$:
	\begin{enumerate}
		\item[\textsf{S1}] all vertices of $G_i$ are between $s_i$ and $t_i$ in $\pi$;
		\item[\textsf{S2}] the channel edges of $G_i$ are in the same page;
		\item[\textsf{S3}] if $i > 0$, the channel edges of $G_i$ and those of $G_{i-1}$ are in different pages.
	\end{enumerate} 
\end{restatable}

\medskip
Note that Condition \textsf{S1} uniquely defines the vertex ordering of $G_h$ in every 3UBE. Namely, the path $s_h \leadsto p_0$ precedes each path $t_{i-1} \leadsto p_i$ (for $i=1,\dots,h$), and each path $t_{i-1} \leadsto p_i$ precedes the path $t_{i} \leadsto p_{i+1}$ (for $i=1,\dots,h-1$) (see \cref{fi:hardness-c} for an example with $h=2$).

\begin{figure}[!t]
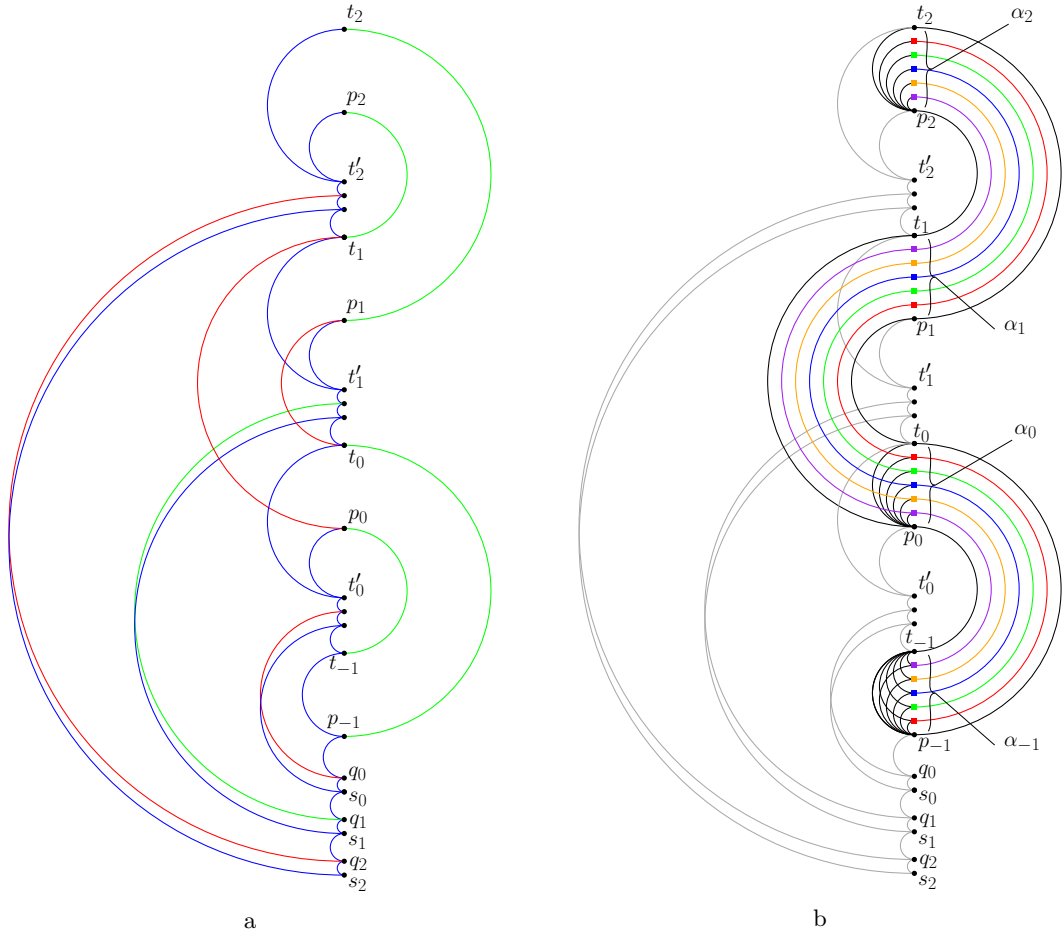

	\centering
	\begin{subfigure}{.4\textwidth}
		\centering
		\includegraphics[width=\columnwidth, page=3]{hardness}
		\subcaption{}
		\label{fi:hardness-c}
	\end{subfigure}
	\hfil
	\begin{subfigure}{.4\textwidth}
		\centering
		\includegraphics[width=\columnwidth, page=4]{hardness}
		\subcaption{}
		\label{fi:hardness-d}
	\end{subfigure}
	\caption{\small{(a) A 3UBE of the \shell{} $G_2$; the colors of the edges represent the pages. (b) Definition of $H_{h,s}$ for $h=2$ and $s=5$. In both figures edges are oriented from bottom to top.}}\label{fi:hardness-2}
\end{figure}

\subparagraph{Filled shell digraphs.} 
Let $G_h$ be a \shell{}. A \emph{\filled{}} $H_{h,s}$ (for $h \geq 0$ and $s \geq 1$) is obtained from $G_h$ by adding $h+2$ groups $\alpha_{-1}, \alpha_0, \dots, \alpha_h$ of $s$ vertices each; see \cref{fi:hardness-d} for an illustration.  The vertices of group $\alpha_i$ are denoted as $v_{i,1}, v_{i,2}, \dots v_{i,s}$. These vertices will be used to map the elements of the set $S$ of an instance of \BETW to an instance of \TUBE. For each vertex $v_{-1,j}$ of the set $\alpha_{-1}$ there is a directed edge $(p_{-1},v_{-1,j})$ and a directed edge $(v_{-1,j},t_{-1})$. For each vertex $v_{i,j}$ of the set $\alpha_i$ with $i \geq 0$ and $i$ even, there is a directed edge  $(p_i,v_{i,j})$. Finally,  for each vertex $v_{i,j}$ of the set $\alpha_i$ with $i \geq 0$, there is a directed edge  $(v_{i-1,j},v_{i,j})$. 

\begin{restatable}{lemma}{lefilled}\label{le:filled}
	Every \filled{} $H_{h,s}$ for $s>0$ and even $h \geq 0$ admits a 3UBE. In any 3UBE $\gamma=\langle \pi, \sigma \rangle$ of $H_{h,s}$ the following conditions hold for every $i=-1,0,1,\dots,h$:
	\begin{enumerate}
		\item[\textsf{F1}] the vertices of the group $\alpha_{i}$ are between $p_{i}$ and $t_{i}$ in $\pi$;
		\item[\textsf{F2}] if $i \geq 0$ the vertices of $\alpha_{i}$ are in reverse order with respect to those of  $\alpha_{i-1}$ in $\pi$;
		\item[\textsf{F3}] if $i \geq 0$ each edge $(v_{i-1,j},v_{i,j})$ is in the page of the channel edges of $G_i$ (for $j=1,\dots,s$).
	\end{enumerate} 
\end{restatable}

\medskip 
Observe that, by Condition \textsf{F2}, all groups $\alpha_i$ with even index have the same ordering in $\pi$ and all groups with odd index have the opposite order. As mentioned above the vertices in the groups $\alpha_i$ will correspond to the elements of the set $S$ of an instance of \BETW in the reduced instance of \TUBE. If the reduced instance admits a 3UBE, the order of the groups in $\pi$ will give the desired order for the instance of \BETW.

\subparagraph{$\mathbf{\Lambda}$-filled shell digraphs and hardness proof.}
Starting from a \filled{} $H_{h,s}$, a \emph{\gadgeted{}} $\widehat{H}_{h,s}$ is obtained  by replacing some edges with a gadget that has two possible configurations in any 3UBE of $\widehat{H}_{h,s}$. More precisely, we replace each edge $(t'_i,p_i)$ of $H_{h,s}$ for $i$ odd with the gadget shown in \cref{fi:hardness-e}. The gadget replacing $(t'_i,p_i)$  will be denoted as $\Lambda_i$. 
Notice that, this replacement preserves Conditions \textsf{F1}--\textsf{F3} of \cref{le:filled}.  

\begin{figure}[!t]
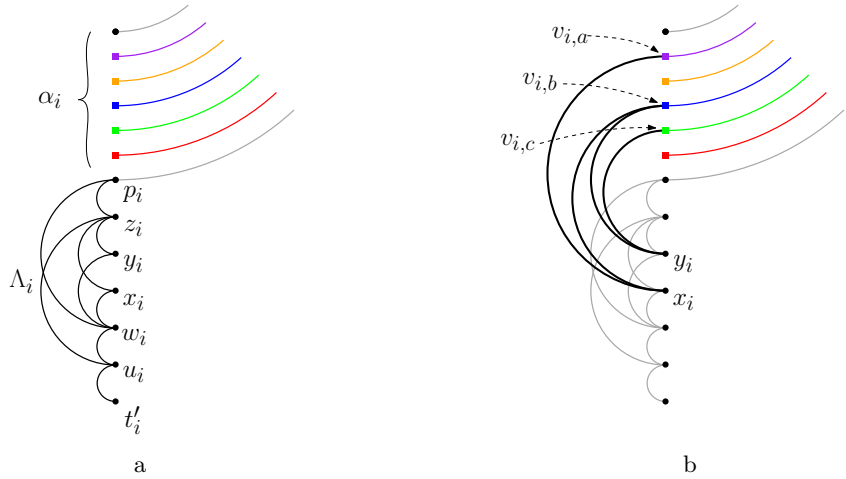

	\centering
	\begin{subfigure}{.35\textwidth}
		\centering
		\includegraphics[width=\columnwidth, page=5]{hardness}
		\subcaption{}
		\label{fi:hardness-e}
	\end{subfigure}
	\hfil
	\begin{subfigure}{.35\textwidth}
		\centering
		\includegraphics[width=\columnwidth, page=6]{hardness}
		\subcaption{}
		\label{fi:hardness-f}
	\end{subfigure}
	\caption{\small{(a) A gadget $\Lambda_i$ (black edges). (b) The triplet edges of $G_i$ (bold edges). }}\label{fi:hardness-3}
\end{figure} 

\begin{restatable}{lemma}{legadget}\label{le:gadgeted}
	Every \gadgeted{} $\widehat{H}_{h,s}$ for $s>0$ and even $h \geq 0$ admits a 3UBE. In any 3UBE $\gamma=\langle \pi, \sigma \rangle$ of $\widehat{H}_{h,s}$ the following conditions hold for every $i=1,3,\dots,h-1$:
	\begin{enumerate}
		\item[\textsf{G1}] the vertices of the gadget $\Lambda_{i}$ are between $t'_{i}$ and $p_{i}$ in $\pi$;
		\item[\textsf{G2}] the vertices $x_{i}$ and $y_{i}$ are between $w_{i}$ and $z_{i}$ in $\pi$ and there exists a 3UBE $\gamma'=\langle \pi', \sigma' \rangle$ of $\widehat{H}_{h,s}$ where the order of $x_i$ and $y_i$ is exchanged in $\pi'$.
	\end{enumerate} 
\end{restatable}

\begin{figure}[!t]
	\centering
	\includegraphics[width=\columnwidth, page=8]{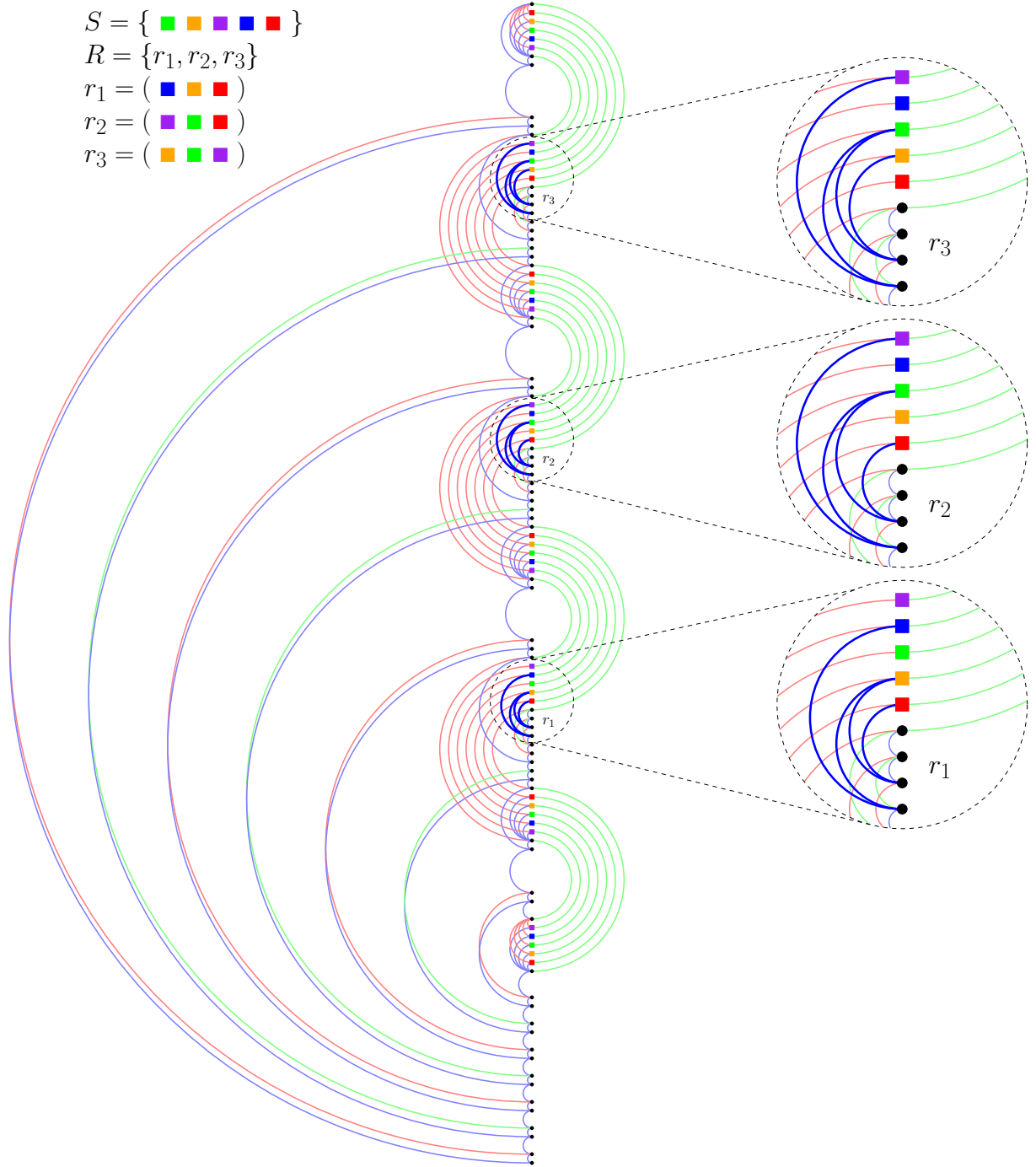}
	\caption{\small{A 3UBE of the $st$-graph $G_I$ reduced from a positive instance $I=\langle S,R\rangle$ of \BETW; the edge colors represent the corresponding pages. Edges are oriented from bottom to top.}}\label{fi:hardness-4}
\end{figure}  

\begin{restatable}{theorem}{thhardness}\label{th:hardness}
	\TUBE is NP-complete even for $st$-graphs. 
\end{restatable}
\begin{proof}[sketch]
	\TUBE is clearly in NP. To prove the hardness we describe a reduction from \BETW.
	From an instance $I=\langle S, R \rangle$ of \BETW we construct an instance $G_I$ of \TUBE that is an $st$-graph; we start from the \gadgeted{} $\widehat{H}_{h,s}$ with $h=2|R|$ and $s=|S|$. Let $v_1,v_2,\dots,v_s$ be the elements of $S$. They are represented in $\widehat{H}_{h,s}$ by the vertices $v_{i,1},v_{i,2},\dots,v_{i,s}$ of the groups $\alpha_i$, for $i=-1,0,1,\dots,h$. In the reduction 
	each group $\alpha_i$ with odd index is used to encode one triplet and, in a 3UBE of $G_I$, the order of the vertices in these groups (which is the same by Condition \textsf{F2}) corresponds to the desired order of the elements of $S$ for the instance $I$.  
	Number the triplets of $R$ from $1$ to $|R|$ and let  $(v_a,v_b,v_c)$ be the $j$-th triplet. We use the group $\alpha_{i}$ and the gadget $\Lambda_{i}$ with $i=2j-1$ to encode the triplet $(v_a,v_b,v_c)$. More precisely, we add to $\widehat{H}_{h,s}$ the edges $(x_{i},v_{i,a})$, $(x_{i},v_{i,b})$, $(y_{i},v_{i,b})$, and $(y_{i},v_{i,c})$ (see \cref{fi:hardness-f}). These edges are called \emph{triplet edges} and are denoted~as~$T_i$. In any 3UBE of $G_I$ the triplet edges are forced to be in the same page and this is possible if and only if the constraints defined by the triplets in $R$ are respected. The digraph obtained by the addition of the triplet edges is not an $st$-graph because the vertices of the last group $\alpha_{h}$ are all sinks. The desired instance $G_I$ of \TUBE is the $st$-graph obtained by adding the edges $(v_{h,j},t_h)$ (for $j=1,2,\dots,s$). \Cref{fi:hardness-4} shows a 3UBE of the $st$-graph $G_I$ reduced from a positive instance $I$ of \BETW.
\end{proof}

\medskip For $k>3$, the reduction from an instance $I$ of \BETW to an instance $G_I$ of \KUBE is similar.
In the shell digraph every pair of forcing edges is replaced by a bundle of $k-1$ edges that mutually conflict (see \cref{fi:hardness-i}).
The edges in each such bundle require $k-1$ pages and force all edges that conflict with them to use the $k$-th page.
Analogously, the two edges $(u_i,z_i)$ and $(w_i,p_i)$ of the gadget $\Lambda_i$ are replaced by a bundle of $k-1$ edges that mutually conflict (see \cref{fi:hardness-j}); this forces the triplet edges to be in the $k$-th page.

\begin{corollary}
	\KUBE is NP-complete for every $k \geq 3$, even for $st$-graphs.
\end{corollary}	

\begin{figure}[!t]
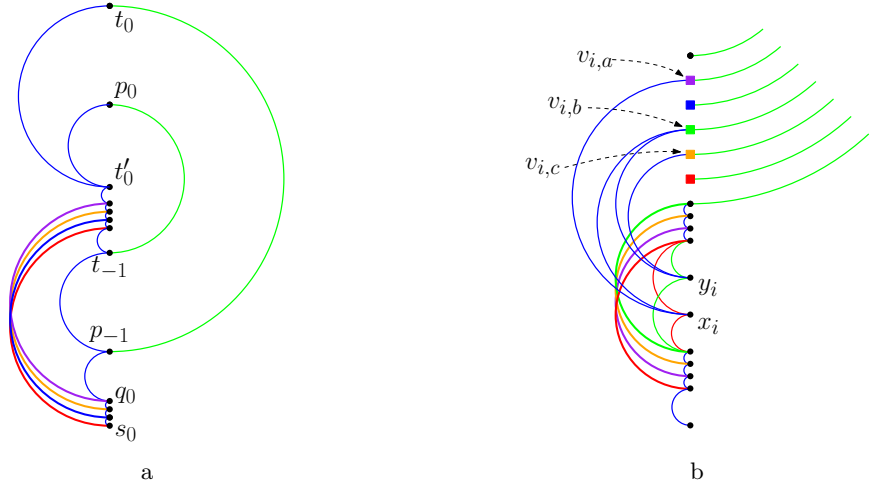

	\centering
	\begin{subfigure}{.35\textwidth}
		\centering
		\includegraphics[width=\columnwidth, page=9]{hardness}
		\subcaption{}
		\label{fi:hardness-i}
	\end{subfigure}
	\hfil
	\begin{subfigure}{.35\textwidth}
		\centering
		\includegraphics[width=\columnwidth, page=10]{hardness}
		\subcaption{}
		\label{fi:hardness-j}
	\end{subfigure}
	\caption{\small{Reduction for \KUBE (example with $k=5$). (a) Replacement of the forcing edges.  (b) Replacement of the gadget $\Lambda_i$. In both figures colors represent the pages. }}\label{fi:hardness-5}
\end{figure}

\section{Existential Results for 2UBE}\label{se:existential}

Let $f$ be an internal face of a plane $st$-graph, and let $p_l$ and $p_r$ be the left and the right path of $f$; $f$ is a \emph{generalized triangle} if either $p_l$ or $p_r$ is a single edge (i.e., a transitive edge), and it is a \emph{rhombus} if each of $p_l$ and $p_r$ consists of exactly two edges (see \cref{fi:generalized-triangle,fi:rhombus}).

Let $G$ be a plane $st$-graph. A \emph{forbidden configuration} of $G$ consists of a transitive edge $e=(u,v)$ shared by two internal faces $f$ and $g$ such that $s_f=s_g=u$ and $t_f=t_g=v$ (i.e., two generalized triangles sharing the transitive edge); see~\cref{fi:forbidden-conf}. The absence of forbidden configurations is a necessary condition for the existence of an embedding-preserving \UBE. 
If $G$ is triangulated, the absence of forbidden configurations is also a sufficient condition~\cite{MchedlidzeS11}.

\begin{figure}[!t]
	\centering
	\begin{subfigure}{.3\textwidth}
		\centering
		\includegraphics[width=0.6\columnwidth, page=1]{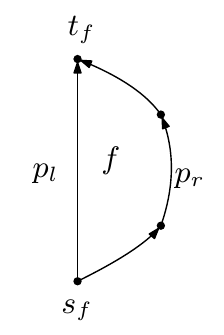}
		\subcaption{}
		\label{fi:generalized-triangle}  
	\end{subfigure}
	\begin{subfigure}{.3\textwidth}
		\centering
		\includegraphics[width=0.6\columnwidth, page=2]{definitions}
		\subcaption{}
		\label{fi:rhombus}
	\end{subfigure}
	\begin{subfigure}{.3\textwidth}
		\centering
		\includegraphics[width=0.6\columnwidth, page=3]{definitions}
		\subcaption{}
		\label{fi:forbidden-conf}
	\end{subfigure}
	\caption{\small{(a) A generalized triangle $G$. (b) A rhombus (c)  A forbidden configuration.}}\label{fi:definitions}
\end{figure}

\begin{restatable}{theorem}{thlongrightpath}\label{th:long-right-path}
  Any plane $st$-graph such that the left and the right path of every internal face contain at least two and three edges, respectively, admits an embedding-preserving \UBE.
\end{restatable}
\begin{proof}[sketch]
	We prove how to construct an embedding-preserving \HPC. The idea is to construct $\overline{G}$ by adding a face of $G$ per time from left to right, according to a topological ordering of the dual graph of $G$. When a face $f$ is added, its right path is attached to the right boundary of the current digraph. We maintain the invariant that at least one edge $e$ in the left path of $f$ belongs to the Hamiltonian path of the current digraph. The Hamiltonian path is extended by replacing $e$ with a path that traverses the vertices of the right path of $f$. To this aim, dummy edges are suitably inserted inside $f$. When all faces are added, the resulting graph is an \HPC $\overline{G}$ of $G$. The idea is illustrated~in~\cref{fi:existential}.\end{proof}

\begin{figure}[htb]
	\centering
	\begin{subfigure}{.45\textwidth}
		\centering
		\includegraphics[width=0.5\columnwidth, page=1]{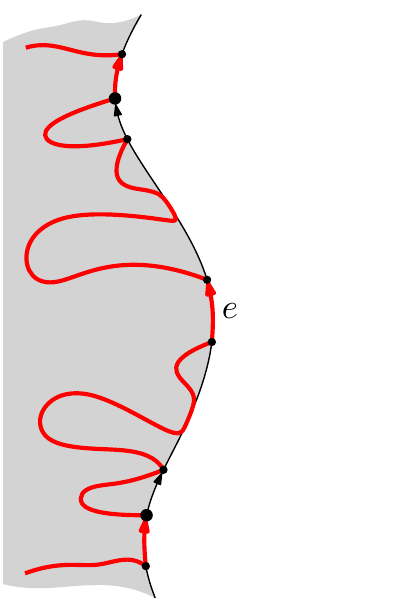}
		\subcaption{}
		\label{fi:existential-1}
	\end{subfigure}
	\begin{subfigure}{.45\textwidth}
		\centering
		\includegraphics[width=0.5\columnwidth, page=2]{existential}
		\subcaption{}
		\label{fi:existential-2}
	\end{subfigure}
	\caption{\small{Idea of the construction in the proof of~\cref{th:long-right-path}. Dummy edges are dashed.}}\label{fi:existential}
\end{figure}

\medskip The next theorem is proved with a construction similar to that of \cref{th:long-right-path}.

\begin{restatable}{theorem}{thrhombi}\label{th:rhombi}
	Let $G$ be a plane $st$-graph such that every internal face of $G$ is a rhombus. Then $G$ admits an embedding-preserving \UBE. 
\end{restatable}

\section{Testing 2UBE for Plane Graphs with Special Faces}\label{se:testing-special}

By~\cref{th:long-right-path}, if all internal faces of a plane $st$-graph $G$ are such that their left and right path contain at least two and at least three edges, respectively, $G$ admits an embedding-preserving \UBE. If these conditions do not hold, an embedding-preserving \UBE may not exist (see \cref{fi:no2ube}). We now describe an efficient testing algorithm for a plane $st$-graph $G=(V,E)$ whose internal faces are generalized triangles or rhombi (see \cref{fi:nohp_rhombi_and_triangles}). We construct a \emph{mixed} graph $G_M=(V,E \cup E_U)$, where $E_U$ is a set of undirected edges and $(u,v) \in E_U$ if $u$ and $v$ are the two vertices of a rhombus face $f$ distinct from $s_f$ and $t_f$ (red edges in \cref{fi:hp_rhombi_and_triangles}). For a rhombus face $f$, the graph obtained from $G$ by adding the directed edge $(u,v)$ inside $f$ is still a plane $st$-graph (see, e.g.~\cite{DBLP:conf/wg/BekosGDLMR18,Battista98}). Since there is only one edge of $E_U$ inside each rhombus face of $G$, this implies the following property.

\begin{property}
	\label{pr:mixed-acyclic}
	Every orientation of the edges in $E_U$ transforms $G_M$ into an acyclic digraph.
\end{property}

\begin{figure}[htb]
	\centering
	\begin{subfigure}{.3\textwidth}
		\centering
		\includegraphics[width=\columnwidth, page=1]{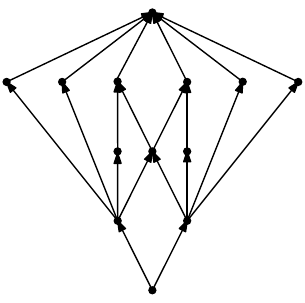}
		\subcaption{}
		\label{fi:no2ube}
	\end{subfigure}
	\hfil
	\begin{subfigure}{.3\textwidth}
		\centering
		\includegraphics[width=\columnwidth, page=3]{testing_shortfaces1}
		\subcaption{}
		\label{fi:nohp_rhombi_and_triangles}
	\end{subfigure}
	\hfil
	\begin{subfigure}{.3\textwidth}
		\centering
		\includegraphics[width=\columnwidth, page=4]{testing_shortfaces1}
		\subcaption{}
		\label{fi:hp_rhombi_and_triangles}
	\end{subfigure}
	\caption{\small{(a) A plane $st$-graph that does not admit a \UBE~\cite{NowakowskiP89}. (b) A plane $st$-graph $G$ whose faces are generalized triangles or rhombi. (b) The mixed graph $G_M=(V,E,E_U)$.}}
	\label{fi:rhombi_and_triangles}
\end{figure}

\begin{restatable}{theorem}{thtestembeddingpreserving}\label{th:test-embedding-preserving}
	Let $G$ be a plane $st$-graph such that every internal face of $G$ is either a generalized triangle or a rhombus. There is an $O(n)$-time algorithm that decides whether $G$ admits an embedding-preserving \UBE, and which computes it in the positive case. 
\end{restatable}
\begin{proof}	
	The edges of $E_U$ are the only edges that can be used to construct an embedding-preserving HP-completion of $G$. This, together with  \cref{th:2ube-hpc-preserving}, implies that $G$ admits a \UBE if and only if the undirected edges of $G_M$ can be oriented so that the resulting digraph $\overrightarrow{G_M}$ has a directed Hamiltonian path from $s$ to $t$. By \cref{pr:mixed-acyclic}, any orientation of the undirected edges of $G_M$ gives rise to an acyclic digraph. On the other hand an acyclic digraph is Hamiltonian if and only if it is unilateral (see, e.g.~\cite[Theorem 4]{abdgkmpst-grd-18}); we recall that a digraph is \emph{unilateral} if each pair of vertices is connected by a directed path (in at least one of the two directions)~\cite{unilateral10}. Testing whether the undirected edges of $G_M$ can be oriented so that the resulting digraph $\overrightarrow{G_M}$ is unilateral, and computing such an orientation if it exists, can be done in time $O(|V|+|E|+|E_U|) = O(n)$~\cite[Theorem 4]{unilateral10}. A Hamiltonian path of $\overrightarrow{G_M}$ is given by a topological ordering of its vertices.
\end{proof}

\section{Testing Algorithms for 2UBE Parameterized by the Branchwidth}\label{se:testing}

In this section, we show that 
the \TWOUBE problem is fixed-parameter tractable with respect to the branchwidth of the input $st$-graph both in the fixed and in the variable embedding setting. Since the treewidth $tw(G)$ and the branchwidth $bw(G)$ of a graph $G$ are within a constant factor from each other (i.e., $bw(G) - 1 \leq tw(G) \leq \lfloor \frac{3}{2}bw(G) \rfloor -1$~\cite{DBLP:journals/jct/RobertsonS91}), our FPT algorithm also extends to graphs of bounded treewidth. Previously, the complexity of this problem was settled only for graphs of treewidth at most~$2$ in the variable embedding~setting\footnote{To our knowledge, no efficient algorithm was known for treewidth~$2$ in the fixed~embedding~setting.}~\cite{GiacomoDLW06}.

We use the SPQR-tree data structure~\cite{DBLP:journals/siamcomp/BattistaT96} to efficiently handle the planar embeddings of the input digraphs and sphere-cut decompositions~\cite{DBLP:journals/combinatorica/SeymourT94} to develop a dynamic-programming approach on the skeletons of the rigid components.
For the definition of the \emph{SPQR-tree} $\mathcal T$ of a biconnected graph and the related concepts of \emph{skeleton} $\skel(\mu)$ and \emph{pertinent graph} $\pert{\mu}$ of a node $\mu$ of $\mathcal T$, \emph{types} of the nodes of $\mathcal T$ (namely, \emph{S-,P-,Q-, and R-nodes}), and \emph{virtual edges} of a skeleton, see \cref{se:app-prel}. 
To ease the description, we can assume that each S-node has exactly two children~\cite{DBLP:journals/siamdm/DidimoGL09} and that the skeleton of each node $\mu$ does not contain the virtual edge representing the parent of $\mu$. 
In particular, we will exploit the following property of $\mathcal T$ when $G$ is an $st$-graph containing the edge $e=(s,t)$ and $\mathcal T$ is rooted at the Q-node of $e$. 

\begin{property}[\cite{DBLP:journals/siamcomp/BattistaT96}]\label{obs:spqr-st-graphs}
	Let $\mu \in \mathcal T$ with poles $u$ and $v$. Without loss of generality, assume that the directed paths connecting $u$ and $v$ in $G$ are oriented from $u$ to $v$. Then, $\pert{\mu}$ is a $uv$-graph.
\end{property}

For the definition of \emph{branchwidth} and \emph{sphere-cut decomposition}, and for the related concepts of \emph{middle set} $\midset(e)$ and \emph{noose} $\mathcal O_e$ of an arc $e$ of the decomposition, and \emph{length} of a noose, see \cref{se:app-prel}. We denote a sphere-cut decomposition of a plane graph $G=(V,E)$ by the triple $\langle T,\xi, \Pi=\bigcup_{a \in E(T)} \pi_a \rangle$, where $T$ is a ternary tree whose leaves are in a one-to-one correspondence with the edges of $G$, which is defined by a bijection $\xi: \mathcal L(T) \leftrightarrow E(G)$ between the leaf set $\mathcal L(T)$ of $T$ and the edge set $E$, and where $\pi_a$ is a circular order of $\midset(a)$, for each arc $a$ of $T$. 
In particular, we will exploit the property that each of the two subgraphs that lie in the interior and in the exterior of a noose is connected and that the set of nooses forms a laminar set family, that is, any two nooses are either disjoint or nested.

Without loss of generality, we assume that the input $st$-graph $G$ contains the edge $(s,t)$, which guarantees that $G$ is biconnected. In fact, in any \UBE of $G$ vertices $s$ and $t$ have to be the first and the last vertex of the spine, respectively. Thus, either $(s,t)$ is an edge of $G$ or it can be added to any of the two pages of the spine of a \UBE of $G$ to obtain a \UBE $\langle \pi,\sigma \rangle$ of $G \cup (s,t)$. Clearly, the edge $(s,t)$ will be incident to the outer face of $\Gamma(\pi,\sigma)$. 

\subparagraph{Overview.} Our approach leverages on the classification of the embeddings of each triconnected component of the biconnected graph $G$. Intuitively, such classification is based on the visibility of the spine that the embedding ``leaves'' on its outer face. We show that the planar embeddings of a triconnected component that yield a \UBE of the component can be partitioned into a finite number of equivalence classes, called \emph{embedding types}. By visiting the SPQR-tree $\mathcal{T}$ of $G$ bottom-up, we describe how to compute all the \emph{realizable embedding types} of each triconnected component, that is, those embedding types that are allowed by some embedding of the component. To this aim we will exploit the realizable embedding types of its child components. If the root of $\mathcal{T}$, which represents the whole $st$-graph $G$, admits at least one planar embedding belonging to some embedding type, then $G$ admits a \UBE.  
The most challenging part of this approach is handling the triconnected components that correspond to the P-nodes, where the problem is reduced to a maximum flow problem on a capacitated flow network with edge demands, and to the R-nodes, where a sphere-cut decomposition of bounded width is used to efficiently compute the feasible embedding types.

\subparagraph{Embedding Types.} 
\begin{figure}[!]
	\centering
	\includegraphics[width=\textwidth,page=1]{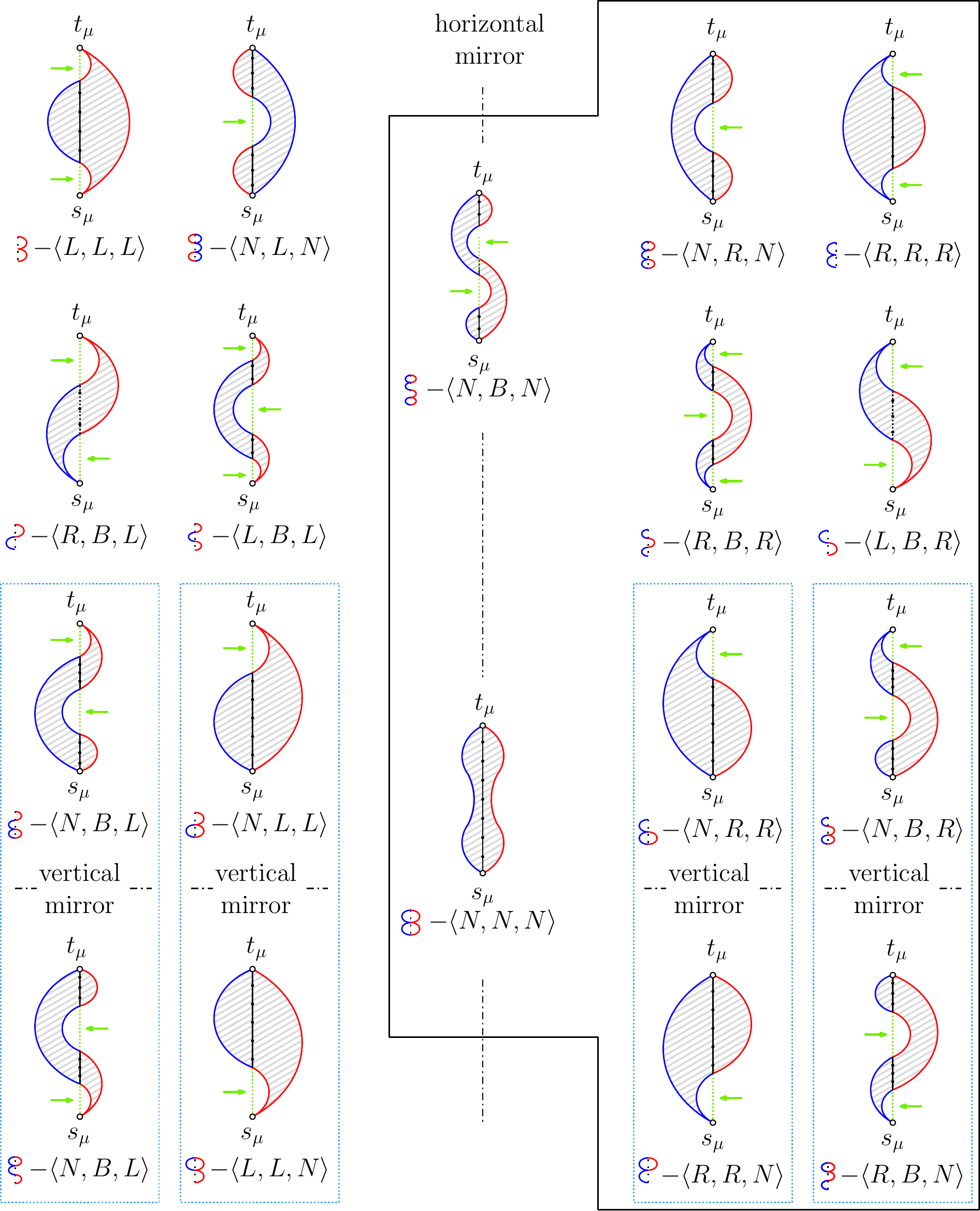}
	\caption{Illustrations of the possible embedding types of a node $\mu$ with poles $s_\mu$ and $t_\mu$; the portion of the spine that is visible from the left or from the right is green. Pairs of embedding types in the same dotted box are one the vertically-mirrored copy of the other. Embedding types on the top are the horizontally-mirrored copy of the ones on the bottom. Embedding types \protect\NBN and \protect\NNN are the horizontal and vertical mirrored copies of themselves. 
	}\label{fig:types}
\end{figure}
Given a \UBE $\langle \pi, \sigma \rangle$, the two pages will be called the \emph{left page} (the one to the left of the spine) and the \emph{right page} (the one to the right of the spine), respectively. We write $\sigma(e)=L$ (resp. $\sigma(e)=R$) if the edge $e$ is assigned to the left page (resp. right page). 
A point $p$ of the spine is \emph{visible from the left (right) page} if it is possible to shoot a horizontal ray originating from $p$ and directed leftward (rightward) without intersecting any edge in $\Gamma(\pi, \sigma)$.   
Let $\mu$ be a node of the SPQR-tree $\cal T$ of $G$ rooted at $(s,t)$. Recall that, by \cref{obs:spqr-st-graphs}, since $\cal T$ has been rooted at $(s,t)$, the pertinent graph $\pert{\mu}$ and the skeleton $\skel(\mu)$ of~$\mu$ are $s't'$-graphs, where $s'$ and $t'$ are the poles of~$\mu$~. We denote by $s_\mu$ (by $t_\mu$) the pole of $\mu$ that is the source (the sink) of $\pert{\mu}$ and of $\skel(\mu)$.
Let $\langle \pi_\mu, \sigma_\mu \rangle$ be a \UBE of $\pert{\mu}$ and let
$\mathcal{E}_\mu$ be the embedding of $\Gamma(\pi_\mu, \sigma_\mu)$. We say that $\mathcal{E}_\mu$ has \emph{embedding type} (or is \emph{of Type}) $\langle s\_{vis}, spine\_{vis},t\_{vis}\rangle$ with $s\_{vis},t\_{vis} \in \{L,R,N\}$ and $spine\_{vis} \in \{L,R,B,N\}$ where:
\begin{enumerate}
	\item $s\_{vis}$ is $L$ (resp., $R$) if in $\mathcal{E}_\mu$ there is a portion of the spine incident to $s$ and between $s$ and $t$ that is visible from the left page (resp., from the right page); otherwise, $s\_{vis}$ is $N$.
	\item $t\_{vis}$ is $L$ (resp., $R$) if in $\mathcal{E}_\mu$ there is a portion of the spine incident to $t$ and between $s$ and $t$ that is visible from the left page (resp., from the right page); otherwise, $t\_{vis}$ is $N$.
	\item $spine\_{vis}$ is $L$ (resp., $R$) if in $\mathcal{E}_\mu$ there is a portion of the spine between $s$ and $t$ that is visible from the left page (resp., from the right page); 
	$spine\_{vis}$ is $B$ if in $\mathcal{E}_\mu$ there is a portion of the spine between $s$ and $t$ that is visible from the left page, and a portion of the spine between $s$ and $t$ that is visible from the right page; otherwise, $spine\_{vis}$ is $N$.
\end{enumerate}

\noindent
We also say that a node $\mu$ and $\pert{\mu}$ \emph{admits Type $\langle x, y, z\rangle$} if $\pert{\mu}$ admits an embedding of Type $\langle x, y, z\rangle$.
We have the following lemma.

\begin{restatable}{lemma}{leconstanttypes}\label{lem:constant-types}
	Let $\mu$ be a node of $\cal T$, let $\langle \pi_\mu, \sigma_\mu \rangle$ be a \UBE of $\pert{\mu}$ and let $\mathcal{E}_\mu$ be a planar embedding of $\Gamma(\pi_\mu, \sigma_\mu)$. 
	Then $\mathcal{E}_\mu$ has exactly one embedding type, where the possibile embedding types are the~$18$ depicted in \cref{fig:types}.
\end{restatable}

Let $\langle \pi, \sigma \rangle$ be a \UBE of $G$, let $\mu$ a node of $\mathcal{T}$, and let $\langle \pi_\mu, \sigma_\mu \rangle$ be the restriction of $\langle \pi, 
\sigma \rangle$ to $\pert{\mu}$. 
Further, let $\langle \pi'_\mu, \sigma'_\mu \rangle \neq \langle \pi_\mu, \sigma_\mu \rangle$ be a \UBE of $\pert{\mu}$.

\begin{restatable}{lemma}{leraplacement}\label{le:replacement}  
	If  $\langle \pi'_\mu, \sigma'_\mu \rangle$ and $\langle \pi_\mu, \sigma_\mu \rangle$ have the same embedding type, then $G$ admits a \UBE whose restriction to $\pert{\mu}$ is $\langle \pi'_\mu, \sigma'_\mu \rangle$.
\end{restatable}
\begin{figure}[t!]
	\centering
	\begin{subfigure}[b]{.24\textwidth}
		\centering
		\includegraphics[page=1,scale=.7]{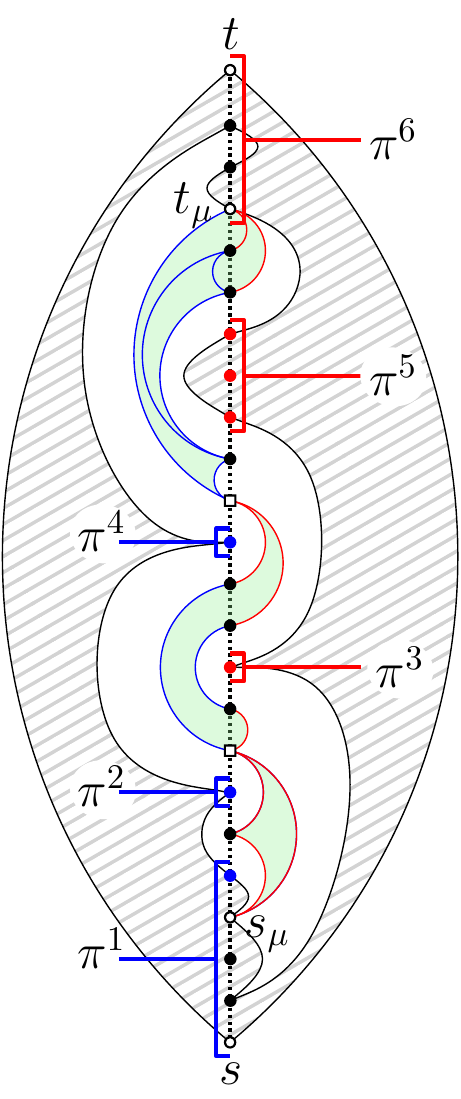}
		\subcaption{Drawing\\ $\Gamma(\pi,\sigma)$ of $G$}
		\label{fig:Canonical-Starting}
	\end{subfigure}\hfil
	\begin{subfigure}[b]{.24\textwidth}
		\centering
		\includegraphics[page=2,scale=.7]{replacement}
		\subcaption{Drawing\\ $\Gamma(\pi'_\mu,\sigma'_\mu)$ of $\pert{\mu}$}
		\label{fig:Canonical-small}
	\end{subfigure}\hfil
	\begin{subfigure}[b]{.24\textwidth}
		\centering
		\includegraphics[page=4,scale=.7]{replacement}
		\subcaption{Drawing $\Gamma^*$ of\\ $G_{\overline{\mu}} \cup \pert{\mu}$}
		\label{fig:Disconnected-embedding}
	\end{subfigure}\hfil
	\begin{subfigure}[b]{.24\textwidth}
		\centering
		\includegraphics[page=3,scale=.7]{replacement}
		\subcaption{Drawing\\ $\Gamma'$ of $G$}
		\label{fig:Final-embedding}
	\end{subfigure}
	\caption{Illustrations for \cref{le:replacement}. The \protect\UBEs of $\pert{\mu}$ are of Type \protect\LBN.}
\end{figure}
\begin{proof}[sketch]
	First, insert a possibly squeezed copy of $\Gamma(\pi'_\mu, \sigma'_\mu)$ (\cref{fig:Canonical-small}) inside $\Gamma(\pi, \sigma)$ (\cref{fig:Canonical-Starting}) in the interior of the face $f_\mu$ of the plane digraph $G_{\overline{\mu}}$ resulting from removing $\pert{\mu}$ (except its poles) from $\Gamma(\pi, \sigma)$. 
	Second, suitably move parts of the boundary of $f_\mu$ along portions of the spine incident to the inserted drawing of $\pert{\mu}$ (\cref{fig:Disconnected-embedding}). Then, continuously move the copies of the poles of $\mu$ inside $f_\mu$ towards their copies in $\Gamma(\pi, \sigma)$, without intersecting~any~edge, to obtain a drawing $\Gamma'$ of $G$~(\cref{fig:Final-embedding}).\end{proof}

Recall that, for each node $\mu$ of $\mathcal{T}$, $\pert{\mu}$ may have exponentially many embeddings, given by the permutations of the children of the P-nodes and by the flips of the R-nodes. \cref{le:replacement} is the reason why we only need to compute a single embedding for each embedding type realizable by $\pert{\mu}$, i.e., a constant number of embeddings instead of an exponential number.

We first describe an algorithm to decide if $G$ admits a \UBE and its running time. The same procedure can be easily refined to actually compute a \UBE of $G$, with no additional cost, by decorating each node $\mu \in \mathcal T$ with the embedding choices performed at $\mu$, for each of its $O(1)$ possible embedding types.

\subparagraph{Testing Algorithm.} The algorithm is based on computing, for each non-root node $\mu$ of $\mathcal{T}$, the set of embedding types realizable by $\pert{\mu}$, based on whether $\mu$ is an S-, P-, Q-, or an R-node. Since, by \cref{lem:constant-types,le:replacement}, $G$ admits a \UBE if and only if the pertinent graph of the unique child of the root Q-node admits an embedding of at least one of the $18$ possible embedding types, this approach allows us to solve the \TWOUBE problem for $G$.

Recall that the only possible embedding choices for $G$ happen at P- and R-nodes. While the treatment of Q- and S-nodes does not require any modification when considering the variable and the fixed embedding settings, for P- and R-nodes we will discuss how to compute the embedding types that are realizable by $\pert{\mu}$ in both such settings. In particular, in the fixed embedding scenario the above characterization needs to additionally satisfy the constraints imposed by the fixed embedding on the skeletons of the P- and R-nodes in $\mathcal{T}$.

Note that a leaf Q-node only admits embeddings of type \QLLL or \QRRR. Also, combining \UBEs of the two children of an S-node $\mu$ always yields a \UBE of $\pert{\mu}$, whose embedding type can be easily computed. 
In \cref{apx:s-node}, we prove the following.

\begin{restatable}{lemma}{lemsdecision}\label{lem:s-decision}
	Let $\mu$ be an S-node. The set of embedding types realizable by $\pert{\mu}$ can be computed in $O(1)$ time, both in the fixed and in the variable~embedding~setting.
\end{restatable}

\subparagraph{P-nodes.} Let $\mu$ be a P-node with poles $s_\mu$ and $t_\mu$. Recall that an embedding for a P-node is obtained by choosing a permutation for its children and an embedding type for each child. Our approach to compute the realizable types of $\pert{\mu}$ consists of considering one type at a time for $\mu$. For each embedding type, we check whether the children of $\mu$, together with their realizable embedding types, can be arranged in a finite number of families of permutations (which we prove to be a constant number) so to yield an embedding of the considered type. 
In order to ease the following description, consider that the arrangements of the children for obtaining some embedding types can be easily derived from the arrangements to obtain the (horizontally) symmetric ones by (i) reversing the left-to-right sequence of the children in the construction and (ii) by taking, for each child, the horizontally-mirrored embedding type; for instance, the arrangements to construct an embedding of Type \LLN can be obtained from the ones to construct an embedding of Type \RRN, and vice versa. Moreover, two embedding types, namely Type \NBN and Type \NNN, are (horizontally) self-symmetric. As a consequence, in order to consider all the embedding types that are realizable by $\pert{\mu}$ we describe how to obtain only $10$ ``relevant'' embedding types (enclosed by a solid polygon in \cref{fig:types}): \RRR, \NRN, \RBR, \LBR, \NRR, \NBR, \RRN, \RBN, \NBN, and \NNN.


Next, we give necessary and sufficient conditions under which the pertinent graph of a P-node admits an embedding of Type \NRR. Then, we show how to test these conditions efficiently by exploiting a suitably defined flow network.
The conditions for the remaining types, given in \cref{see:apex-p-node}, can be tested with the same~algorithmic~strategy.

\begin{figure}[t]
	\begin{subfigure}[b]{.31\textwidth}
		\centering
		\includegraphics[page=2,height=55mm]{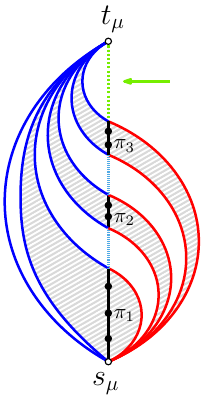}
		\subcaption{\NRR; Case~1}
		\label{fig:P-NRR-C1}
	\end{subfigure}
	\begin{subfigure}[b]{.31\textwidth}
		\centering
		\includegraphics[page=3,height=55mm]{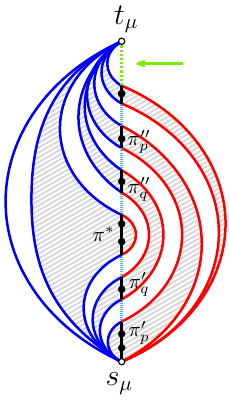}
		\subcaption{\NRR; Case~2}
		\label{fig:P-NRR-C2}
	\end{subfigure}
	\begin{subfigure}[b]{.31\textwidth}
		\centering
		\includegraphics[page=4,height=55mm]{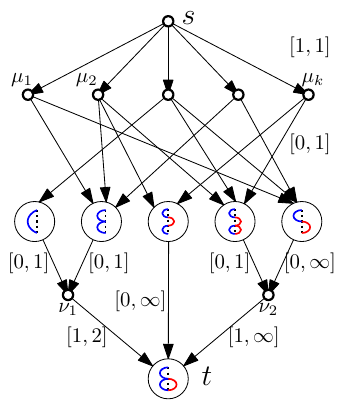}
		\subcaption{Network $\mathcal N$ for Case~2}
		\label{fig:P-NRR-C2-NetWork}
	\end{subfigure}
	\caption{Case 1 (a) and Case 2 (b) of \cref{le:p-NRR} for a P-nodes $\mu$ of Type \protect\NRR. The spine is colored either green, blue, or black. The green part is the portion of the spine that is visible from the right, the black parts correspond to the bottom-to-top sequences of the internal vertices of $\pert{\mu}$ inherited from the \UBEs of the children of $\mu$, the blue parts join sequences inherited from different children. (c) Capacitated flow network $\mathcal N$ with edge demands corresponding to (b).}
\end{figure}

\begin{restatable}[Type \protect\NRR]{lemma}{pnrr}\label{le:p-NRR}
	Let $\mu$ be a P-node. Type \NRR is admitted by $\mu$ in the variable embedding setting if and only if at least one of two cases occurs.
	\begin{inparaenum}[\bf ({Case}~1)]
		\item The children of $\mu$ can be partitioned into two parts: The first part consists either of a Type-\QRRR Q-node child, or of a Type-\NRR child, or both.
		The second part consists of any number, even zero, of Type-\LBR children.
		\item The children of $\mu$ can be partitioned into three parts: The first part consists either of a Type-\QRRR Q-node child, or of a non-Q-node Type-\RRR child, or both.
		The second part consists of any number, even zero, of Type-\RBR children. The third part consists of any positive number of Type-\NBR or Type-\LBR children, with at most one Type-\NBR~child.
	\end{inparaenum}
\end{restatable}

Regarding the time complexity of testing the existence of a Type-\NRR embedding of $\pert{\mu}$, we show that deciding if one of {\bf ({Case}~1)} or {\bf ({Case}~2)} of \cref{le:p-NRR} applies can be reduced to a network flow problem on a network $\mathcal{N}$ with edge demands. The network for {\bf ({Case}~2)} is depicted in \cref{fig:P-NRR-C2-NetWork}. The details of this construction are given in \cref{see:apex-p-node}, where the next lemma is proven.

\begin{restatable}{lemma}{lepdecisionvariable}\label{lem:p-decision-variable}
	Let $\mu$ be a P-node with $k$ children. The set of embedding types realizable by $\pert{\mu}$ can be computed in $O(k^2)$ time in the variable embedding setting.
\end{restatable}

The fixed embedding scenario for a P-node $\mu$ can be addressed by processing the children of $\mu$ in the left-to-right order defined by the given embedding~of~$G$. The details of such an approach are given in \cref{apx:p-case-fixed}, where the following is proven.

\begin{lemma}\label{lem:p-decision-fixed}
	Let $\mu$ be a P-node with $k$ children. The set of embedding types realizable by $\pert{\mu}$ can be computed in $O(k)$ time in the fixed embedding setting.
\end{lemma}

\cref{lem:s-decision,lem:p-decision-fixed} yield a counterpart, in the fixed embedding setting, of the linear-time algorithm by Di Giacomo et al.~\cite{GiacomoDLW06} to compute \UBEs of series-parallel graphs.

\begin{theorem}
	There exists an $O(n)$-time algorithm to decide whether an $n$-vertex series-parallel $st$-graph admits an embedding-preserving \UBE.
\end{theorem}

\begin{figure}[t!]
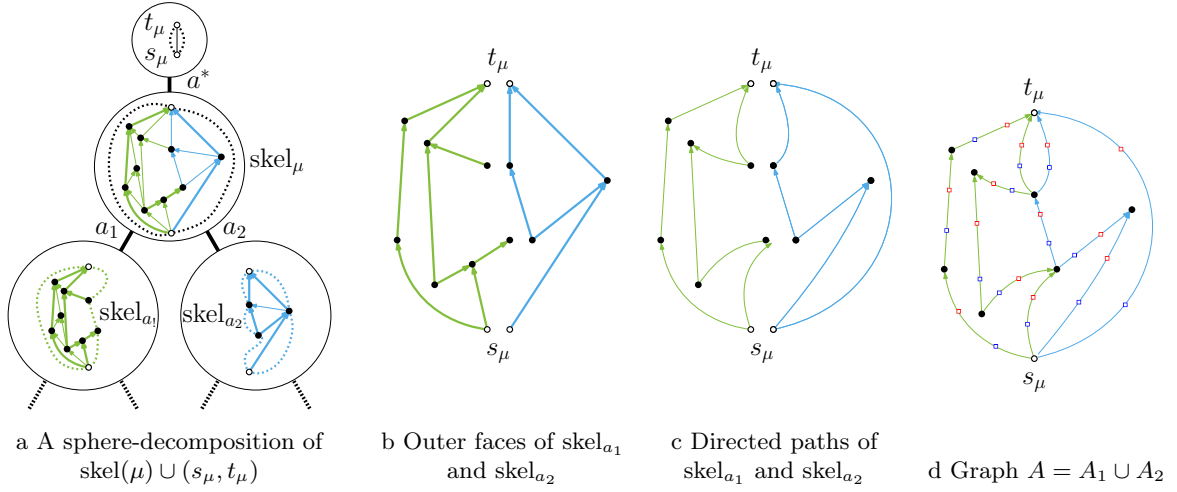

	\begin{subfigure}[b]{.3\textwidth}
		\centering
		\includegraphics[page=6,scale=.7]{R-node}
		\subcaption{A sphere-decomposition of $\skel(\mu) \cup (s_\mu,t_\mu)$}
		\label{fig:rigid-sphere-cut}
	\end{subfigure}\hfil
	\begin{subfigure}[b]{.2\textwidth}
		\centering
		\includegraphics[page=2,scale=.7]{R-node}
		\subcaption{Outer faces of $\skel_{a_1}$ and $\skel_{a_2}$}
		\label{fig:rigid-b}
	\end{subfigure}\hfil
	\begin{subfigure}[b]{.2\textwidth}
		\centering
		\includegraphics[page=3,scale=.7]{R-node}
		\subcaption{Directed paths of $\skel_{a_1}$ and $\skel_{a_2}$}
		\label{fig:rigid-c}
	\end{subfigure}\hfil
	\begin{subfigure}[b]{.2\textwidth}
		\centering
		\includegraphics[page=4,scale=.7]{R-node}
		\subcaption{Graph $A=A_1 \cup A_2$}
		\label{fig:rigid-d}
	\end{subfigure}
	\caption{Illustrations for the R-node case. (a) A partial sphere-cut decomposition of the graph $\skel(\mu) \cup (s_\mu,t_\mu)$ rooted at the edge $(s_\mu,t_\mu)$, where $\mu$ is an R-node. The nooses are dotted curves. 
  (b) Outer faces of graphs $\skel_{a_1}$ and $\skel_{a_2}$. (c) The graphs defined by the directed paths on the outer face of $\skel_{a_1}$ and $\skel_{a_2}$. (d) The auxiliary graph $A$ for $\skel_{a}$.
	}
	\label{fig:rigid}
\end{figure}

\subparagraph{R-nodes.} Let $\mu$ be an R-node and $(T,\xi,\Pi)$ be a sphere-cut decomposition of $\skel(\mu) \cup (s_\mu,t_\mu)$ of width~$\bw$, rooted at the node $\rho$ with $\xi(\rho)=(s_\mu,t_\mu)$; refer to \cref{fig:rigid-sphere-cut}.
Each arc $a$ of $T$ is associated with a subgraph $\skel_a$ of $\skel(\mu)$ and with a subgraph $\mypert_a$ of $\pert{\mu}$, both bounded by the noose of $a$. Let $a_1$ and $a_2$ be the two arcs leading to $a$ from the bottom of $T$.
Intuitively, our strategy to compute the embedding types of $\pert{\mu}$ is to visit $T$ bottom-up maintaining a succinct description of size $O(\bw)$ of the properties of the \UBEs of $\mypert_a$. 
To this aim, we construct cycles composed of directed edges that are in one-to-one correspondence with maximal directed paths along the outer face of $\skel_{a_1}$ and $\skel_{a_2}$ (\cref{fig:rigid-b,fig:rigid-c}), which we use to define an auxiliary graph $A$ whose \UBEs concisely represent the possible \UBEs of $\mypert_a$ obtained by combining the \UBEs of $\mypert_{a_1}$ and $\mypert_{a_2}$ (\cref{fig:rigid-d}).
When we reach the arc $a^*$ incident to $\rho$  with $\skel_{a^*}=\skel(\mu)$, we use the computed properties to determine the embedding types realizable by $\pert{\mu}$. We provide full details in \cref{apx:r-node}.

\begin{restatable}{lemma}{ledecision}\label{le:r-decision}
	Let $\mu$ be an R-node whose skeleton $\skel(\mu)$ has $k$ children and branchwidth $\bw$. The set of embedding types realizable by $\pert{\mu}$ can be computed in $O(2^{O(\bw\log{\bw})}\cdot k)$ time, both in the fixed and in the variable embedding setting, provided that a sphere-cut decomposition $\langle T_\mu, \xi_\mu, \Pi_\mu \rangle$ of width $\beta$ of $\skel^+(\mu)$ is given.
\end{restatable}

By \cref{lem:s-decision,lem:p-decision-fixed,lem:p-decision-variable,le:r-decision} and since $\mathcal T$ has $O(|G|)$ size~\cite{DBLP:journals/siamcomp/BattistaT96,DBLP:journals/siamdm/DidimoGL09}, we get the following.

\begin{theorem}\label{th:test-branchwidth}
	There exists an $O(2^{O(\bw \log {\bw})}\cdot n + n^2 + g(n))$-time algorithm to decide if an $n$-vertex planar (plane) $st$-graph of branchwidth $\bw$  admits a (embedding-preserving) \UBE, where $g(n)$ is the computation time of a sphere-cut decomposition of an $n$-vertex plane graph.
\end{theorem}

Observe that $g(n)$ is $O(n^3)$ by the result in~\cite{DBLP:journals/combinatorica/SeymourT94}. Thus, we get the following.

\begin{corollary}\label{co:test-result}
	There exists an $O(2^{O(\bw \log {\bw})}\cdot n + n^3)$-time algorithm to decide whether an $n$-vertex planar (plane) $st$-graph of branchwidth $\bw$  admits a (embedding-preserving) \UBE.
\end{corollary}

Since the branchwidth of a planar graph $G$ is at most $2.122\sqrt{n}$~\cite{DBLP:journals/jgt/FominT06}, \cref{co:test-result} immediately implies that the \TWOUBE problem can be solved in sub-exponential time.

\begin{corollary}\label{co:test-general}
	There exists an $O(2^{O(\sqrt{n} \log {\sqrt{n}})} + n^3)$-time algorithm to decide whether an $n$-vertex planar (plane) $st$-graph admits a (embedding-preserving) \UBE.
\end{corollary}

\section{Conclusion and Open Problems}
Our results provide significant advances on the complexity of the {\KUBE} problem. We showed NP-hardness for $k \geq 3$; we gave FPT- and polynomial-time algorithms for relevant families of planar $st$-graphs when $k=2$. 
We point out that our FPT-algorithm can be refined to run in $O(n^2)$ time for $st$-graphs of treewidth at most~$3$, by constructing in linear time a sphere-cut decomposition of their rigid components.
We conclude with some open problems.

\begin{itemize}
	\item The main open question is about the complexity of the \TWOUBE problem, which has been conjectured to be NP-complete in the general case~\cite{HeathP99}. 	
	\item The digraphs in our NP-completeness proof are not upward planar. Since there are upward planar digraphs that do not admit a $3$UBE~\cite{Hung89}, it would be interesting to study whether the problem remains NP-complete for three pages and upward planar digraphs.
	
	\item Finally, it is natural to investigate other families of planar digraphs for which a \UBE always exists or polynomial-time testing algorithms can be devised.
\end{itemize}

\paragraph{Acknowledgments.} This research began at the Bertinoro Workshop on Graph Drawing 2018.

\bibliographystyle{alpha}
\bibliography{../biblio}

\clearpage
\appendix

\input{appendix_preliminaries}

\input{appendix_hardness}

\input{appendix_existential}

\input{appendix_testing}

\end{document}

%% file: appendix_preliminaries.tex
\section{Additional Material for \cref{se:preliminaries}}\label{se:app-prel}

\subparagraph{Connectivity and Planarity.} A graph $G$ is \emph{$1$-connected}, or \emph{simply-connected}, if there is a path between any two vertices. $G$ is \emph{$k$-connected}, for $k \geq 2$, if the removal of $k-1$ vertices leaves the graph $1$-connected. A $2$-connected ($3$-connected) graph is also called \emph{biconnected} (\emph{triconnected}).

A \emph{planar drawing} of $G$ is a geometric representation in the plane such that: $(i)$ each vertex $v \in V(G)$ is drawn as a distinct point $p_v$; $(ii)$ each edge $e=(u,v) \in E(G)$ is drawn as a simple curve connecting $p_u$ and $p_v$; $(iii)$ no two edges intersect in $\Gamma$ except at their common end-vertices (if they are adjacent). A graph is \emph{planar} if it admits a planar drawing. A planar drawing $\Gamma$ of $G$ divides the plane into topologically connected regions, called \emph{faces}. The \emph{external face} of $\Gamma$ is the region of unbounded size; the other faces are \emph{internal}. A \emph{planar embedding} of $G$ is an equivalence class of planar drawings that define the same set of (internal and external) faces, and it can be described by the clockwise sequence of vertices and edges on the boundary of each face plus the choice of the external face. Graph $G$ together a given planar embedding is an \emph{embedded planar graph}, or simply a \emph{plane graph}: If $\Gamma$ is a planar drawing of $G$ whose set of faces is that described by the planar embedding of $G$, we say that $\Gamma$ \emph{preserves} this embedding, or also that $\Gamma$ is an \emph{embedding-preserving drawing} of $G$. 

\subparagraph{Sphere-cut decomposition.}
A \emph{branch decomposition} $\langle T,\xi \rangle$ of a graph $G=(V,E)$ consists of an unrooted ternary tree $T$ (each node of $T$ has degree one or three) and of a bijection $\xi : \mathcal L(T) \leftrightarrow E(G)$ from the leaf set $\mathcal L(T)$  of $T$ to the edge set $E(G)$ of $G$. 
For each arc $a$ of $T$, let $T_1$ and $T_2$ be the two connected components of $T-a$, and, for $i=1,2$, let $G_i$ be the subgraph of $G$ that consists of the edges corresponding to the leaves of $T_i$, i.e., the edge set $\{\xi(\mu):\mu \in \mathcal L(T) \cap V(T_i)\}$. The \emph{middle set} $\midset(a) \subseteq V (G)$ is the intersection of the vertex sets of $G_1$ and $G_2$, i.e., $\midset(a) := V(G_1) \cap V(G_2)$.
The \emph{width} $\bw(\langle T,\xi \rangle)$ of $\langle T,\xi \rangle$ is the maximum size of the middle sets over all arcs of $T$, i.e., $\bw(\langle T,\xi \rangle) := max\{|\midset(a)|: a \in T\}$. An \emph{optimal branch decomposition} of $G$ is a branch decomposition with minimum width; this width is called the \emph{branchwidth} $\bw(G)$ of~$G$. An optimal branch decomposition of a given planar graph with $n$ vertices can be constructed in $O(n^3)$ time~\cite{DBLP:journals/combinatorica/SeymourT94}. 

Let $\Sigma$ be a sphere. 
A \emph{$\Sigma$-plane} graph $G$ is a planar graph $G$ embedded (i.e., topologically drawn) on $\Sigma$.
A \emph{noose} of a $\Sigma$-plane graph $G$ is a closed simple curve on $\Sigma$ that  
\begin{inparaenum}[(i)]
	\item \label{Cond:noodeI} intersects $G$ only at vertices and
	\item \label{Cond:noodeII} traverses each face at most once.
\end{inparaenum}
The \emph{length} of a noose $O$ is the number of vertices it intersects. Every noose $O$ bounds two closed discs $\Delta_1$, $\Delta_2$ in $\Sigma$, i.e., $\Delta_1 \cap \Delta_2= O$ and $\Delta_1 \cup \Delta_2=\Sigma$.
For a $\Sigma$-plane graph $G$, a \emph{sphere-cut decomposition} $\langle T,\xi, \Pi=\bigcup_{a \in E(T)} \pi_a \rangle$ of $G$ is a branch decomposition $\langle T,\xi \rangle$ of $G$ together with a set $\Pi$ of circular orders $\pi_a$ of $\midset(a)$, for each arc $a$ of $T$, such that there exists a noose $O_a$ whose closed discs $\Delta_1$ and $\Delta_2$ enclose the drawing of $G_1$ and of $G_2$, respectively, for each arc $a$ of $T$. 
Observe that, $O_a$ intersect $G$ exactly at $\midset(a)$ and its length is $|\midset(a)|$. Also, 
\cref{Cond:noodeII} of the definition of noose implies that graphs $G_1$ and $G_2$ are both connected and that the set of nooses forms a laminar set family, that is, any two nooses are either disjoint or nested.
A clockwise traversal of $O_a$ in the drawing of $G$ defines the cyclic ordering $\pi_a$ of $\midset(a)$. We always assume that the vertices of every middle set $\midset(a) = V (G_1) \cap V (G_2)$ are enumerated according to $\pi_a$. We will exploit the following main result of Dorn et al~\cite{DBLP:journals/algorithmica/DornPBF10}.

\begin{theorem}[\cite{DBLP:journals/algorithmica/DornPBF10},Theorem~1]
	Let $G$ be a connected $n$-vertex $\Sigma$-plane graph having branchwidth $\bw$ and no vertex of degree one. There exists a sphere-cut decomposition of $G$ having width $\bw$ which can be constructed in $O(n^3)$ time.
\end{theorem}

\subparagraph{SPQR-trees of Planar $\mathbf{st}$-Graphs.} Let $G$ be a biconnected graph. An \emph{SPQR-tree} $\mathcal T$ of $G$ is a tree-like data structure that represents the decomposition of $G$ into its triconnected components and can be computed in linear time~\cite{DBLP:journals/siamcomp/BattistaT96,DBLP:conf/gd/GutwengerM00,DBLP:journals/siamcomp/HopcroftT73}. See \cref{fi:spqr-tree} for an illustration. Each node $\mu$ of $\mathcal T$ corresponds to a triconnected component of $G$ with two special vertices, called \emph{poles}; the triconnected component corresponding to a node $\mu$ is described by a multigraph $\skel(\mu)$ called the \emph{skeleton} of $\mu$. 
Let $\skel^+(\mu)=\skel(\mu) \cup (u,v)$, where $u$ and $v$ are the poles of $\mu$.ì
A node $\mu$ of $\mathcal T$ is of one the following types:
\begin{inparaenum}[(i)]
	\item \emph{R-node}, if $\skel^+(\mu)$ is triconnected; 
	\item \emph{S-node}, if $\skel^+(\mu)$ is a cycle of length at least three; 
	\item \emph{P-node}, if $\skel^+(\mu)$ is a bundle of at least three parallel edges; and 
	\item \emph{Q-nodes}, if it is a leaf of $\mathcal T$; in this case the node represents a single edge of the graph and $\skel^+(\mu)$ consists of two parallel edges.
\end{inparaenum}
A \emph{virtual edge} in $\skel(\mu)$ corresponds to a tree node $\nu$ adjacent to $\mu$ in $\mathcal T$. 
The edge of $G$ corresponding to the root $\rho$ of $\mathcal T$ is the  \emph{reference edge} of $G$, and $\mathcal T$ is the SPQR-tree of $G$ \emph{with respect to $e$}. For every node $\mu \neq \rho$ of $\mathcal T$, the subtree $T_\mu$ rooted at $\mu$ induces a subgraph $\pert{\mu}$ of $G$ called the \emph{pertinent graph} of $\mu$, which is described by $T_\mu$ in the decomposition: The edges of $\pert{\mu}$ correspond to the Q-nodes (leaves) of $T_\mu$. 

If $G$ is planar, the SPQR-tree $\mathcal T$ of $G$ rooted at a Q-node $\rho$ representing edge $e$ implicitly describes all planar embeddings of $G$ in which $e$ is incident to the outer face. All such embeddings are obtained by combining the different planar embeddings of the skeletons of P- and R-nodes: For a P-node $\mu$, the different embeddings of $\skel(\mu)$ are the different permutations of its edges. If $\mu$ is an R-node, $\skel(\mu)$ has two possible planar embeddings, obtained by flipping $\skel(\mu)$ at its poles. 
Let $\mu$ be a node of $\mathcal T$, let $\mathcal E$ be an embedding of $G$, let $\mathcal E_\mu$ be the embedding $\pert{\mu}$ in $\mathcal E$, and let $f^{O}_\mu$ be the outer face of $\mathcal E_\mu$.
The path along $f^{O}_\mu$ between $s_\mu$ and $t_\mu$ that leaves $f^{O}_\mu$ to its left (resp. to its right) when traversing the boundary of $f^{O}_\mu$ from $s_\mu$ to $t_\mu$ 
to its the \emph{left outer path} (resp. the \emph{right outer path}) of $\mathcal E_\mu$.
The \emph{left} (resp., \emph{right}) \emph{outer face} of $\mathcal E_\mu$ is the face of $\mathcal E$ that is incident to the left (resp., to the right) outer path of $\mathcal E_\mu$. In \cref{fi:spqr-tree-a}, the left outer face and the right outer face of the S-node whose poles are the nodes labeled $8$ and $13$ are yellow and green, respectively.

\begin{figure}[!t]
	\centering
	\hspace{-6mm}
	\begin{subfigure}{.32\textwidth}
		\centering
		\includegraphics[width=\columnwidth, page=2]{spqr-tree}
		\subcaption{}
		\label{fi:spqr-tree-a}
	\end{subfigure}
	\begin{subfigure}{.64\textwidth}
		\centering
		\includegraphics[width=\columnwidth, page=1]{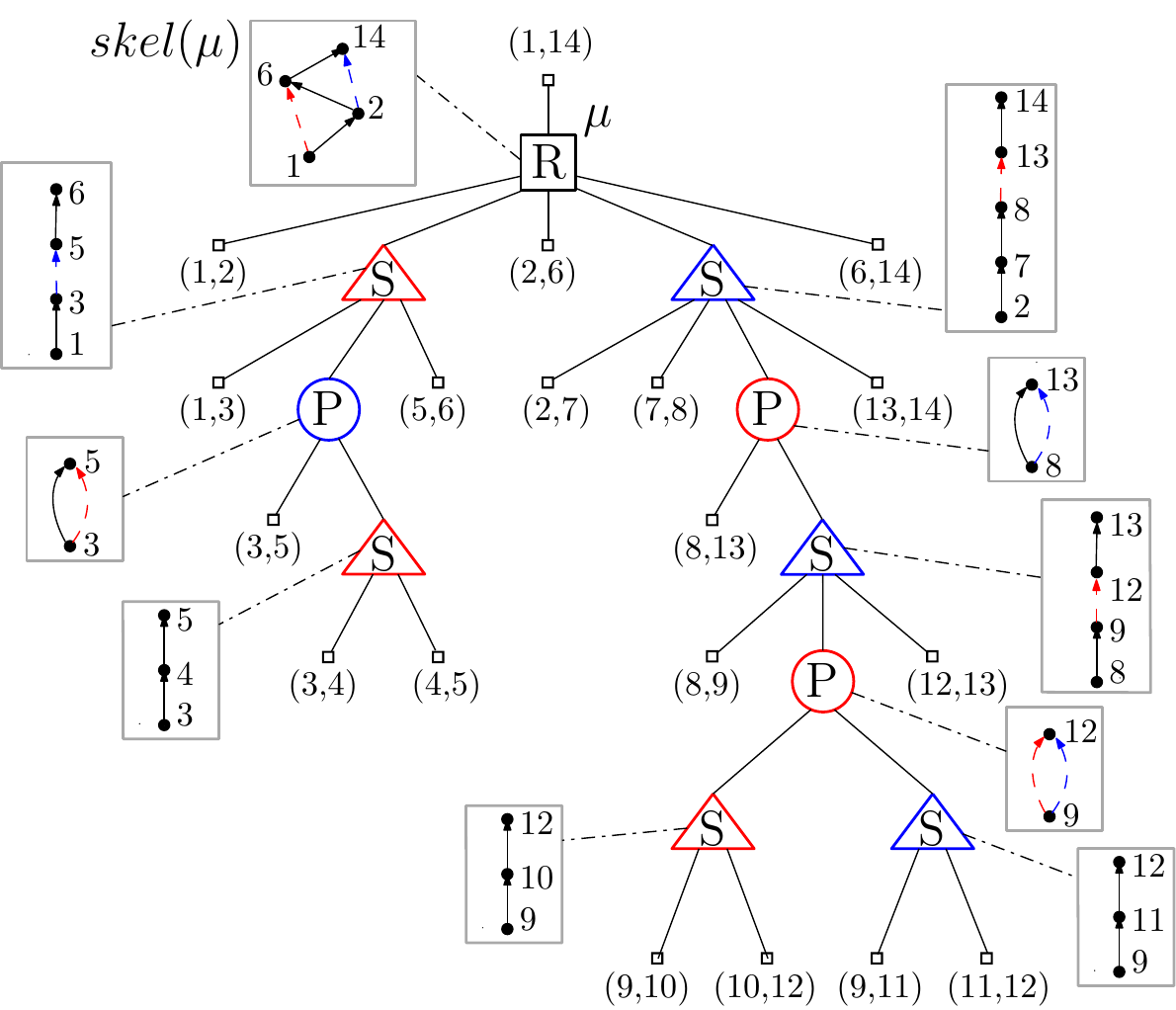}
		\subcaption{}
		\label{fi:spqr-tree-b}
	\end{subfigure}
	\caption{\small{(a) A biconnected planar $st$-graph $G$. (b) SPQR-tree of $G$ rooted at the edge $(s,t)$.}}\label{fi:spqr-tree}
\end{figure}

%% file: appendix_hardness.tex
\section{Additional Material for \cref{se:complexity}}\label{se:app-hardness}

\leshell*
\begin{proof}
	The proof is by induction on $h$.
	
	\textbf{Base case $h=0$.}
	We describe how to define a 3UBE $\gamma=\langle \pi, \sigma \rangle$ of $G_0$. The eight vertices of the directed path $s_0 \leadsto p_0$ must appear in $\pi$ in the same order they appear along the path. Consider now $t_0$. Because of the closing edge $(t'_0,t_0)$, we have $\pi(t'_0) < \pi(t_0)$. If we put $t_0$ between $t'_0$ and $p_0$, the channel edges $(p_{-1},t_0)$ and $(t_{-1},p_0)$ and the forcing edges $(s_{0},s'_0)$ and $(q_{0},q'_0)$ would mutually conflict. But then a 3UBE would not exist by \cref{pr:conflicting}. Thus, the only possibility is that $t_0$ is the last vertex in $\pi$. This uniquely defines the order $\pi$ and implies condition \textsf{S1}. As for the page assignment $\sigma$, the two forcing edges must be in different pages because they conflict. Since each of the two channel edges conflict with both forcing edges, the channel edges cannot be assigned to the pages used for the forcing edges. Thus, they must be in the same page, which is possible because the two channel edges do not conflict (this proves condition \textsf{S2}). Finally, the closing edge conflicts with the channel edge $(t_{-1},p_0)$ and thus it cannot be in the same page as the channel edges; since however it does not conflict with any other edge it can be assigned to one of the pages used for the forcing edges. This concludes the proof that a 3UBE of $G_0$ exists and that it must satisfy conditions \textsf{S1} and \textsf{S2}. Condition \textsf{S3} does not apply in this case.
	
	\textbf{Inductive case $h>0$.} By induction, $G_{h-1}$ admits a 3UBE $\gamma'=\langle \pi', \sigma' \rangle$ that satisfies \textsf{S1}--\textsf{S3}. We extend $\gamma'$ to a 3UBE $\gamma=\langle \pi, \sigma \rangle$ of $G_h$ as follows. Since $\gamma'$ satisfies \textsf{S1}, $s_{h-1}$ is the first vertex in $\pi'$ and $t_{h-1}$ is the last one. The vertices of path $s_h \leadsto s_{h-1}$ must appear in $\pi$ in the same order they appear along the path. Analogously, the vertices of $t_{h-1} \leadsto p_h$ must appear in $\pi$ in the order they have along the path. Because of the closing edge $(t'_h,t_h)$, we have $\pi(t'_h) < \pi(t_h)$. Therefore, $s_h$ must be the first vertex along $\pi$. Consider now $t_h$. If we put $t_h$ between $t'_h$ and $p_h$, the channel edges $(p_{h-1},t_h)$ and $(t_{h-1},p_h)$ and the forcing edges $(s_{h},s'_h)$ and $(q_{h},q'_h)$ would mutually conflict. But then a 3UBE would not exist by \cref{pr:conflicting}. Thus, $t_h$ must be the last vertex in $\pi$. This uniquely defines the order $\pi$ and implies condition \textsf{S1} for $G_h$. As for the page assignment $\sigma$, observe that the only exclusive edge of $G_h$ that conflicts with some edge of $G_{h-1}$ is the edge $(p_{h-1},t_h)$, which only conflicts with the channel edge $(p_{h-2},t_{h-1})$ of $G_{h-1}$. This implies that $(p_{h-1},t_h)$ must be in a page different from the one of $(p_{h-2},t_{h-1})$. The two forcing edges of $G_h$ must be in a page different from the channel edge $(p_{h-1},t_h)$ and since they conflict, they must be in different pages. The channel edge $(t_{h-1},p_h)$ conflicts with the forcing edges but not with the other channel edge $(p_{h-1},t_h)$. Thus, the channel edges must be in the same page (which proves condition \textsf{S2}). The fact that the page of $(p_{h-1},t_h)$ must be different from that of $(p_{h-2},t_{h-1})$, implies condition \textsf{S3}. Finally, the closing edge conflicts with the channel edge $(t_{h-1},p_h)$ and thus it cannot be in the same page as the channel edges; since however it does not conflict with any other edge, it can be assigned to one of the pages used for the forcing edges. This concludes the proof that a 3UBE of $G_h$ exists and that it satisfies conditions \textsf{S1}, \textsf{S2}, and \textsf{S3}.  
\end{proof}

\lefilled*
\begin{proof}
	The proof is by induction on $h$. 
	
	\textbf{Base case $h=0$.} We describe how to define a 3UBE $\gamma=\langle \pi, \sigma \rangle$ of $H_{0,s}$. The subgraph $G_0$ of $H_{0,s}$ admits a 3UBE that satisfies conditions \textsf{S1}--\textsf{S3} by \cref{le:shell}. 
	Let $v_{-1,j}$ (with $1 \leq j \leq s$) be a vertex of $\alpha_{-1}$. The edges $(p_{-1},v_{-1,j})$ and $(v_{-1,j},t_{-1})$ imply $\pi(p_{-1}) < \pi(v_{-1,j})$ and $\pi(v_{-1,j}) < \pi(t_{-1})$, which proves condition \textsf{F1} for $\alpha_{-1}$. Consider now the group $\alpha_0$; the edge $(p_{0},v_{0,j})$ implies $\pi(p_{0}) < \pi(v_{0,j})$. On the other hand, if we put $v_{0,j}$ after  $t_0$, the edge $(v_{-1,j},v_{0,j})$, the channel edge $(p_{-1},t_0)$, and the two forcing edges $(s_0,s'_0)$ and $(q_0,q'_0)$ would mutually conflict. But then a 3UBE would not exist by \cref{pr:conflicting}. Thus, each vertex of group $\alpha_0$ must be between $p_0$ and $t_0$ in $\pi$, which implies condition \textsf{F1} for the group $\alpha_{0}$.  
	As for the page assignment $\sigma$, each $(v_{-1,j},v_{0,j})$ conflicts with each forcing edge of $G_0$, and thus it must be in the page of the channel edges of $G_0$. This implies condition \textsf{F3}. The edges  $(v_{-1,j},v_{0,j})$ can be assigned to the same page only if the vertices of $\alpha_{-1}$ appear in reverse order with respect to those of $\alpha_0$ in $\pi$. Thus, condition \textsf{F2} holds and a 3UBE of $H_{0,s}$ can be defined by choosing an arbitrary order for the vertices of $\alpha_{-1}$ and the reverse order for the vertices of $\alpha_0$. 
	
	\textbf{Inductive case $h>0$.}  Consider the subgraph $H'_{h,s}$ of $H_{h,s}$ consisting of $H_{h-2,s}$ plus the exclusive vertices and edges of $G_h$. By induction and by~\cref{le:shell}, $H'_{h,s}$ admits a 3UBE $\gamma'=\langle \pi', \sigma' \rangle$ that satisfies conditions \textsf{F1}--\textsf{F3} and conditions \textsf{S1}--\textsf{S3}. We extend $\gamma'$ to a 3UBE $\gamma=\langle \pi, \sigma \rangle$ of $H_{h,s}$ as follows. By condition \textsf{F1} of $\gamma'$, each vertex $v_{h-2,j}$ is before $t_{h-2}$ in $\pi'$; on the other hand, because of the edges $(p_h,v_{h,j})$, each vertex $v_{h,j}$ must follow $p_{h}$ in $\pi$. This implies that each $v_{h-1,j}$ is between $p_{h-1}$ and $t_{h-1}$ in $\pi$. Indeed, if $v_{h-1,j}$ was before $p_{h-1}$ in $\pi$, the edge $(v_{h-1,j},v_{h,j})$, the channel edges $(p_{h-1},t_h)$ and the two forcing edges of $G_h$ would mutually conflict and therefore a 3UBE would not exist by~\cref{pr:conflicting}. On the other hand, if $v_{h-1,j}$ was after $t_{h-1}$, the edge $(v_{h-2,j},v_{h-1,j})$, the channel edges $(p_{h-2},t_{h-1})$ and the two forcing edges of $G_{h-1}$ would mutually conflict and again a 3UBE would not exist by \cref{pr:conflicting}. Thus, each vertex of group $\alpha_{h-1}$ must be between $p_{h-1}$ and $t_{h-1}$, which proves condition \textsf{F1} for the group $\alpha_{h-1}$. Consider now a vertex $v_{h,j}$. If it was after $t_h$ in $\pi$, then the edge $(v_{h-1,j},v_{h,j})$, the channel edge $(p_{h-1},t_{h})$ and the two forcing edges of $G_h$ would mutually conflict -- again a 3UBE would not exist by \cref{pr:conflicting}. Hence, each vertex of $\alpha_h$ is between $p_h$ and $t_h$, which proves condition \textsf{F1} also for $\alpha_h$.  
	
	As for the page assignment $\sigma$, each $(v_{h-2,j},v_{h-1,j})$ conflicts with each forcing edge of $G_{h-1}$ and hence it must be in the page of the channel edges of $G_{h-1}$. The same argument applies to the edges $(v_{h-1,j},v_{h,j})$ with respect to the forcing edges of $G_h$. Thus the edges $(v_{h-1,j},v_{h,j})$ must be in the page of the channel edges of $G_h$, which proves condition \textsf{F3}.  
	
	The edges  $(v_{h-2,j},v_{h-1,j})$ can be assigned to the same page only if the vertices of $\alpha_{h-2}$ appear in reverse order with respect to those of $\alpha_{h-1}$ in $\pi$. Analogously, the edges  $(v_{h-1,j},v_{h,j})$ can be assigned to the same page only if the vertices of $\alpha_{h-1}$ appear in reverse order with respect to those of $\alpha_{h}$ in $\pi$. Thus, condition \textsf{F2} holds and a 3UBE of $H_{h,s}$ can be defined by ordering the vertices of $\alpha_{h-1}$ in reverse order with respect to those of $\alpha_{h-2}$ and the vertices of $\alpha_h$ with the same order as those of $\alpha_{h-2}$.  
\end{proof}

\legadget*
\begin{proof}
	By \cref{le:filled} $H_{h,s}$ admits a 3UBE $\gamma'=\langle \pi',\sigma'\rangle$ where $t'_{i}$ and $p_{i}$ are consecutive in $\pi'$ for each $i=0,1,\dots,h$. Notice that the gadget $\Lambda_{i}$ admits a 3UBE (actually a 2UBE) $\gamma_{i}$. If we replace the edge $(t'_{i},p_{i})$ with $\gamma_{i}$, we do not create any conflict between the edges of $\Lambda_i$ and the other edges of $H_{h,s}$. This proves that $\widehat{H}_{h,s}$ has a 3UBE. 
	About condition \textsf{G1}, observe that since any vertex of the gadget $\Lambda_{i}$ belongs to a directed path from $t'_{i}$ to $p_{i}$, then the vertices of $\Lambda_{i}$ must be between $t'_{i}$ and $p_{i}$. Analogously, $x_{i}$ and $y_{i}$ both appear in a directed path from $w_{i}$ to $z_{i}$ and therefore they must be between $w_i$ and $z_i$. Suppose that $\pi(x_i) < \pi(y_i)$ (the other case is symmetric). If we exchange the order of $x_i$ and $y_i$ in $\pi'$ we introduce a conflict between $(w_i,x_i)$ and $(y_i,z_i)$, which do not conflict with any other edges. If they are in the same page in $\gamma$ it is sufficient to change the page of one of them in $\gamma'$. 
\end{proof}

\thhardness*
\begin{proof} 
	\TUBE is clearly in NP. To prove the hardness we describe a reduction from \BETW.
	From an instance $I=\langle S, R \rangle$ of \BETW we construct an instance $G_I$ of \TUBE that is an $st$-graph; we start from the \gadgeted{} $\widehat{H}_{h,s}$ with $h=2|R|$ and $s=|S|$. Let $v_1,v_2,\dots,v_s$ be the elements of $S$. They are represented in $\widehat{H}_{h,s}$ by the vertices $v_{i,1},v_{i,2},\dots,v_{i,s}$ of the groups $\alpha_i$, for $i=-1,0,1,\dots,h$. In the reduction 
	each group $\alpha_i$ with odd index is used to encode one triplet and, in a 3UBE of $G_I$, the order of the vertices in these groups (which is the same by condition \textsf{F2}) corresponds to the desired order of the elements of $S$ for the instance $I$.  
	Number the triplets of $R$ from $1$ to $|R|$ and let  $(v_a,v_b,v_c)$ be the $j$-th triplet. We use the group $\alpha_{i}$ and the gadget $\Lambda_{i}$ with $i=2j-1$ to encode the triplet $(v_a,v_b,v_c)$. More precisely, we add to $\widehat{H}_{h,s}$ the edges $(x_{i},v_{i,a})$, $(x_{i},v_{i,b})$, $(y_{i},v_{i,b})$, and $(y_{i},v_{i,c})$ (see \cref{fi:hardness-f}). These edges are called \emph{triplet edges} and are denoted as $T_i$. In any 3UBE of $G_I$ the triplet edges are forced to be in the same page and this is possible if and only if the constraints defined by the triplets in $R$ are respected. The digraph obtained by the addition of the triplet edges is not an $st$-graph because the vertices of the last group $\alpha_{h}$ are all sinks. The desired instance $G_I$ of \TUBE is the $st$-graph obtained by adding the edges $(v_{h,j},t_h)$ (for $j=1,2,\dots,s$). \Cref{fi:hardness-4} shows a 3UBE of the $st$-graph $G_I$ reduced from a positive instance $I$ of \BETW.  
	
	We now show that $I$ is a positive instance of \BETW if and only if the $st$-graph $G_I$ constructed as described above admits a 3UBE.
	Suppose first that $I$ is a positive instance of \BETW, i.e., there exists an ordering $\tau$ of $S$ that satisfies all triplets in $R$. The subgraph $\widehat{H}_{h,s}$ of $G_I$ admits a 3UBE that satisfies conditions \textsf{S1}--\textsf{S3}, \textsf{F1}--\textsf{F3}, and \textsf{G1}--\textsf{G2} by \cref{le:shell,,le:filled,,le:gadgeted}. Observe that the order of the vertices of the groups $\alpha_i$ can be arbitrarily chosen (provided that all groups with even index have the same order and the groups with odd index have the reverse order). Thus we can choose the order of the groups with odd index to be equal to $\tau$. Let $\gamma=\langle \pi, \sigma \rangle$ be the resulting 3UBE of $\widehat{H}_{h,s}$. We now show that if we add the triplet edges to $\gamma$, these edges do not conflict. Let $(v_a,v_b,v_c)$ be the triplet encoded by the triplet edges $T_i$ and suppose that $\tau(v_a) < \tau(v_b) < \tau(v_c)$ (the other case is symmetric). Since the vertices of the groups with odd index are ordered in $\pi$ as in $\tau$, we have $\pi(v_{i,a}) < \pi(v_{i,b}) < \pi(v_{i,c})$. If $\pi(x_i) < \pi(y_i)$ then the edges $T_i$ do not conflict. If otherwise $\pi(y_i) < \pi(x_i)$, by condition \textsf{G2} we can exchange the order of $x_i$ and $y_i$, thus guaranteeing again that the triplet edges $T_i$ do not conflict. On the other hand, the triplet edges $T_i$ conflict with the edges $E^{\Lambda}_i$ of $\Lambda_i$, with the channel edges $E^{ch}_i$ of $G_i$, and with all the edges $E^{\alpha}_i$ connecting group $\alpha_{i-1}$ to group $\alpha_{i}$. All the edges in $E^{\Lambda}_i \cup E^{ch}_i \cup E^{\alpha}_i$ can be assigned to only two pages. Indeed, the edges $E^{\Lambda}_i$ require two pages, while one page is enough for the edges of $E^{ch}_i \cup E^{\alpha}_i$. Also, since the edges of $E^{\Lambda}_i$ do not conflict with those in $E^{ch}_i \cup E^{\alpha}_i$, two pages suffice for all of them. Hence, the triplet edges $T_i$ can all be assigned to the third page. Since this is true for all the triplet edges, $G_I$ admits a 3UBE. 
	
	Suppose now that $G_I$ admits a 3UBE $\gamma=\langle \pi,\sigma \rangle$. By \cref{le:shell,,le:filled,,le:gadgeted} $\gamma$ satisfies conditions \textsf{S1}--\textsf{S3}, \textsf{F1}--\textsf{F3}, and \textsf{G1}--\textsf{G2} .
	By condition \textsf{F2} the order of the vertices of the groups $\alpha_i$ with odd index is the same for all groups. We claim that all triplets in $R$ are satisfied if this order is used as the order $\tau$ for the elements of $S$. 
	Let $(v_a,v_b,v_c)$ be the triplet encoded by the triplet edges $T_i$. By condition \textsf{G2}, the vertices $x_i$ and $y_i$ of the gadget $\Lambda_i$ are between $w_i$ and $z_i$ in $\pi$. Thus the triplet edges $T_i$  conflict with the edges $(u_i,z_i)$ and $(w_i,p_i)$. These two edges must be in two different pages because they conflict. It follows that the triplet edges $T_i$ must all be in the same page, i.e., the third one. Since the three edges of $T_i$ are in the same page we have either $\pi(x_i) < \pi(y_i)$ and $\pi(v_{i,a}) < \pi(v_{i,b}) < \pi(v_{i,c})$ or $\pi(y_i) < \pi(x_i)$ and $\pi(v_{i,c}) < \pi(v_{i,b}) < \pi(v_{i,a})$  (any other order would cause a crossing between the edges of $T_i$). In both cases vertex $v_{i,b}$ is between $v_{i,a}$ and $v_{i,c}$, i.e., $v_b$ is between $v_a$ and $v_c$ in $\tau$. Since this is true for all triplets, $I$ is a positive instance of \BETW.
\end{proof}

%% file: appendix_existential.tex
\section{Additional Material for \cref{se:existential}}\label{se:app-existential}

\thlongrightpath*
\begin{proof}
	We prove how to construct an \HPC. The idea is to construct $\overline{G}$ by adding a face of $G$ per time from left to right. Namely, the faces of $G$ are added according to a topological ordering of the dual graph of $G$.  
	When a face $f$ is added, its right path is attached to the right boundary of the current digraph. We maintain the invariant that at least one edge $e$ in the left path of $f$ belongs to the Hamiltonian path of the current digraph. The Hamiltonian path is extended by replacing $e$ with a path that traverses the vertices of the right path of $f$. To this aim, dummy edges are suitably inserted inside $f$. When all faces are added, the resulting graph is an \HPC $\overline{G}$ of $G$.
	
	More precisely, let $G^*$ be the dual graph of $G$. Let $s^*=f_0,f_1,\dots,f_N,f_{N+1}=t^*$ be a topological sorting of $G^*$. Denote by $G_0$ the left boundary of $G$ and by $G_i$, for $i=1,2,\dots,N$, the subgraph of $G$ consisting of the faces $f_1,f_2,\dots,f_i$. $G_i$ can be obtained by adding the right path  $p^i_r$ of face $f_i$ to $G_{i-1}$. We construct a sequence $\overline{G}_0, \overline{G}_1,\dots, \overline{G}_N$ of $st$-graphs such that $\overline{G}_i$ is an \HPC of $G_i$. Clearly, $\overline{G}_N$ will be an \HPC of $G$. While constructing the sequence, we maintain the following invariant: given any two consecutive edges along the right boundary of $\overline{G}_i$, at least one of them belongs to the Hamiltonian path $P_{\overline{G}_i}$ of $\overline{G}_i$. $\overline{G}_0$ coincides with $G_0$ and all its edges are in $P_{\overline{G}_0}$, so the invariant holds. Suppose then that $\overline{G}_{i-1}$, with $i>1$, satisfies the invariant. To construct $\overline{G}_i$ we must add the right path $p^i_r$ of $f_i$ plus possibly some dummy edges inside $f_i$. Let $s_f=v_0,v_1,v_2,\dots,v_{k-1},t_f=v_k$ be the right path $p^i_r$ of $f_i$ and let $s_f=u_0,u_1,u_2,\dots,u_{h-1},t_f=u_h$ be the left path $p^i_l$ of $f_i$. By hypothesis $p^i_r$ has at least three edges, and therefore $k \geq 3$; moreover, since $G$ has no transitive edge, $h\geq 2$. Notice that $p^i_l$ is a subpath of the right boundary of $\overline{G}_{i-1}$ and that the right boundary of $\overline{G}_i$ is obtained from the right boundary of $\overline{G}_{i-1}$  by replacing $p^i_l$ with $p^i_r$. Let $(u_{-1},s_f)$ be the edge along the right boundary of $\overline{G}_{i-1}$ entering $s_f$ and let $(t_f,u_{h+1})$ be the edge along the right boundary of $\overline{G}_{i-1}$ exiting $t_f$. We have different cases depending on whether $(u_{-1},s_f)$ and $(t_f,u_{h+1})$ belong to $P_{\overline{G}_{i-1}}$ or not.
	
	\begin{figure}[htb]
		\centering
		\begin{subfigure}{.45\textwidth}
			\centering
			\includegraphics[width=0.5\columnwidth, page=1]{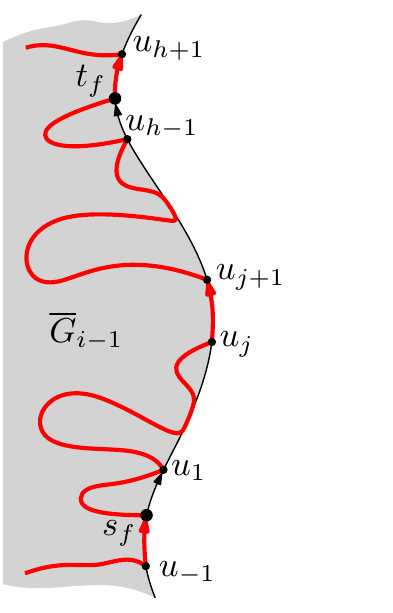}
			\subcaption{}
			\label{fi:long-right-path-1}
		\end{subfigure}
		\begin{subfigure}{.45\textwidth}
			\centering
			\includegraphics[width=0.5\columnwidth, page=2]{long-right-path}
			\subcaption{}
			\label{fi:long-right-path-2}
		\end{subfigure}
		\caption{\small{Illustration for the proof of~\cref{th:long-right-path}: Case 1. }}\label{fi:long-right-path-a}
	\end{figure}

	\textbf{Case 1: both $(u_{-1},s_f)$ and $(t_f,u_{h+1})$ belong to $P_{\overline{G}_{i-1}}$.} See \cref{fi:long-right-path-a} for an illustration. By the invariant there is an edge $(u_j,u_{j+1})$ with $0 \leq j \leq h-1$ between $s_f$ and $t_f$ that belongs to $P_{\overline{G}_{i-1}}$. We add the two dummy edges $(u_j,v_1)$ and $(v_{k-1},u_{j+1})$, thus ``extending'' $P_{\overline{G}_{i-1}}$ to a Hamiltonian path $P_{\overline{G}_{i}}$ of $\overline{G}_i$; namely, the edge $(u_j,u_{j+1})$ is bypassed by the path $u_j,v_1,v_2,\dots,v_{k-1},u_{j+1}$. The only edges of $p^i_r$ that do not belong to $P_{\overline{G}_{i}}$ are $(s_f,v_1)$ and $(v_{k-1},t_f)$. Since $(u_{-1},s_f)$ and $(t_f,u_{h+1})$ belong to $P_{\overline{G}_{i-1}}$ they also belong to $P_{\overline{G}_{i}}$ and thus the invariant is preserved. 
	
	\begin{figure}[htb]
		\centering
		\begin{subfigure}{.45\textwidth}
			\centering
			\includegraphics[width=0.5\columnwidth, page=3]{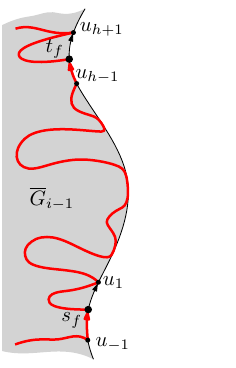}
			\subcaption{}
			\label{fi:long-right-path-3}
		\end{subfigure}
		\centering
		\begin{subfigure}{.45\textwidth}
			\centering
			\includegraphics[width=0.5\columnwidth, page=4]{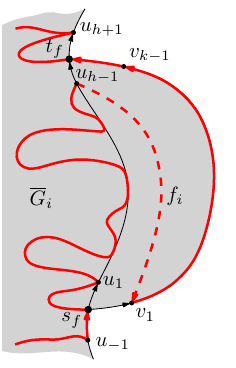}
			\subcaption{}
			\label{fi:long-right-path-4}
		\end{subfigure}
		\caption{\small{Illustration for the proof of~\cref{th:long-right-path}: Case 2. }}\label{fi:long-right-path-b}
	\end{figure}
	
	\textbf{Case 2: $(u_{-1},s_f)$ belongs to $P_{\overline{G}_{i-1}}$, while $(t_f,u_{h+1})$ does not.} See \cref{fi:long-right-path-b} for an illustration. By the invariant the edge $(u_{h-1},t_f)$ belongs to $P_{\overline{G}_{i-1}}$. We add the dummy edge $(u_{h-1},v_1)$. This``extends'' $P_{\overline{G}_{i-1}}$ to $P_{\overline{G}_{i}}$ of $\overline{G}_i$ bypassing the edge $(u_{h-1},t_f)$ with the path $u_{h-1},v_1,v_2,\dots,v_{k-1},t_f$. The only edge of $p^i_r$ that does not belong to $P_{\overline{G}_{i}}$ is $(s_f,v_1)$. Since $(u_{-1},s_f)$ belong to $P_{\overline{G}_{i-1}}$ it also belongs to $P_{\overline{G}_{i}}$ and thus the invariant is preserved.

	\textbf{Case 3: $(u_{-1},s_f)$ does not belongs to $P_{\overline{G}_{i-1}}$, while $(t_f,u_{h+1})$ does.} This case is symmetric to the previous one.
	
	\begin{figure}[htb]
		\centering
		\begin{subfigure}{.45\textwidth}
			\centering
			\includegraphics[width=0.5\columnwidth, page=5]{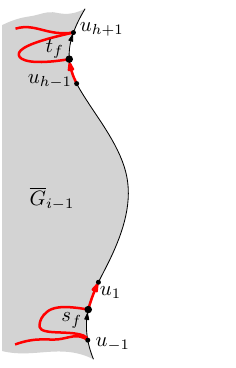}
			\subcaption{}
			\label{fi:long-right-path-5}
		\end{subfigure}
		\centering
		\begin{subfigure}{.45\textwidth}
			\centering
			\includegraphics[width=0.5\columnwidth, page=6]{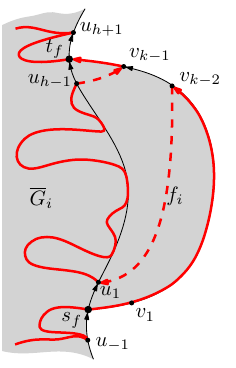}
			\subcaption{}
			\label{fi:long-right-path-6}
		\end{subfigure}
		\caption{\small{Illustration for the proof of~\cref{th:long-right-path}: Case 4.}}\label{fi:long-right-path-c}
	\end{figure} 
	
	\textbf{Case 4: neither $(u_{-1},s_f)$ nor $(t_f,u_{h+1})$ belong to $P_{\overline{G}_{i-1}}$.} See \cref{fi:long-right-path-c} for an illustration. By the invariant the edges $(s_f,u_{1})$ $(u_{h-1},t_f)$ belong to $P_{\overline{G}_{i-1}}$. We add the two dummy edges $(v_{k-2},u_1)$ and $(u_{h-1},v_{k-1})$. In this case we ``extend'' $P_{\overline{G}_{i-1}}$ to $P_{\overline{G}_{i}}$ bypassing $(s_f,u_{1})$ with the path $s_f,v_1,\dots,v_{k-2}$ and bypassing the edge $(u_{h-1},t_f)$ with the path $u_{h-1}, v_{k-1}, t_f$. The only edge of $p^i_r$ that does not belong to $P_{\overline{G}_{i}}$ is $(v_{k-2},v_{k-1})$; moreover, the two edges $(s_f,v_1)$ and $(v_{k-1},t_f)$ belong to $P_{\overline{G}_{i}}$. Thus the invariant is preserved.
\end{proof}

\thrhombi*
\begin{proof}
	The statement can be proved using the same technique as in the proof of~\cref{th:long-right-path}. If all faces are rhombi, when we construct $\overline{G}_i$ from $\overline{G}_{i-1}$, we have that $p^i_l$ is a path $s_f,u_1,t_f$ and $p^i_r$ is a path $s_f,v_1,t_f$. One between $(s_f,u_1)$ and $(u_1,t_f)$ belongs to $P_{\overline{G}_{i-1}}$. If $(s_f,u_1)$ belongs to $P_{\overline{G}_{i-1}}$ we add the dummy edge $(v_1,u_1)$. In this case we bypass the edge $(s_f,u_1)$ with the path $s_f,v_1,u_1$. If $(u_1,t_f)$ belongs to $P_{\overline{G}_{i-1}}$ we add the dummy edge $(u_1,v_1)$. In this case we bypass the edge $(u_1,t_f)$ with the path $u_1,v_1,t_f$. In both cases it is easy to see that the invariant is maintained.  
\end{proof}

%% file: appendix_testing.tex
\section{Additional Material for \cref{se:testing}}\label{se:app-testing}

\leconstanttypes*

\begin{proof}
The first statement follows from the definition of embedding type. To see that the number of embedding types allowed by $\pert{\mu}$ is at most $18$ it suffices to consider the following facts.
First, the number of different embedding types allowed by $\pert{\mu}$ is at most $36$. Further, some combinations are ``impossible'', in the sense that not all combinations of values for $s\_{vis}$, $spine\_{vis}$, and $t\_{vis}$ appear in a \UBE. In particular, we have that following values for $spine\_{vis}$ are forbidden: 
\begin{inparaenum}[(a)]
\item\label{pr:impossibile-l} $spine\_{vis}$ cannot be $L$, if either $s\_{vis}=R$ or $t\_{vis}=R$;
\item\label{pr:impossibile-r} $spine\_{vis}$ cannot be $R$, if either $s\_{vis}=L$ or $t\_{vis}=L$;
\item\label{pr:impossibile-n} $spine\_{vis}$ cannot be $N$, if either $s\_{vis}$ or $t\_{vis}$ is different from $N$.
\end{inparaenum}
\cref{pr:impossibile-l} rules out $5$ combinations; 
\cref{pr:impossibile-r} and \cref{pr:impossibile-n} rule out $5$ and $8$ more combinations, respectively.
This leaves us with $18$ embedding types; see \cref{fig:types}.
\end{proof}

\leraplacement*

\begin{proof}
We show how to construct a \UBE $\langle \pi', \sigma' \rangle$ of $G$ whose restriction to $\pert{\mu}$~is~$\langle \pi'_\mu, \sigma'_\mu \rangle$. 
For the ease of description we actually show how to construct an upward planar drawing $\Gamma'$ of $G$ in which each vertex $v$ lies along the spine in the same bottom-to-top order determined by $\pi'$ and each edge is drawn on the page assigned by $\sigma'$. Clearly, such a drawing implies the existence of $\langle \pi', \sigma' \rangle$.

Consider the canonical drawing $\Gamma(\pi,\sigma)$ of $G$; refer to \cref{fig:Canonical-Starting}.
Consider the drawing $\Gamma_{\overline{\mu}}$ of $G_{\overline{\mu}}$ obtained by restricting $\Gamma(\pi,\sigma)$ to the edges of $G$ not in $\pert{\mu}$ and to their endpoints. 
Denote by $\pi_{\overline{\mu}}$ the bottom-to-top order of the vertices of $G_{\overline{\mu}}$ in $\Gamma_{\overline{\mu}}$. 
Let $\pi_1, \pi_2, \dots, \pi_k$ be the maximal subsequences of $\pi_{\overline{\mu}}$ between $s_\mu$ and $t_\mu$ and composed of consecutive vertices in $\pi_{\overline{\mu}}$ that are also consecutive in $\pi$ (refer to \cref{fig:Canonical-Starting}). Observe that sequences $\pi_i$, $i=1, \dots, k$, may be formed by a single vertex or by multiple vertices. Also, the first sequence $\pi_1$ includes $s_\mu$ and the last sequence $\pi_k$ includes $t_\mu$. Further, unless $\pert{\mu}$ is an edge, for which the statement is trivial, we have that $k \geq 3$.   

Drawing $\Gamma_{\overline{\mu}}$ contains a face $f_\mu$ that is incident to $s_\mu$, $t_\mu$, and all the starting and ending vertices of the sequences $\pi_i$, with $i = 2, \dots, k-1$. In particular, some starting and ending vertices of the sequences $\pi_i$ are encountered when traversing $f_\mu$ clockwise from $s_\mu$ to $t_\mu$ (\emph{left vertices of $f_\mu$}) and some of them are encountered when traversing $f_\mu$ counter-clockwise from $s_\mu$ to $t_\mu$ (\emph{right vertices of $f_\mu$}).

We show how to insert into face $f_\mu$ the drawing $\Gamma(\pi'_\mu,\sigma'_\mu)$ of $\pert{\mu}$, producing the promised drawing $\Gamma'$ of the \UBE $\langle \pi', \sigma' \rangle$ of~$G$.

First, consider the drawing $\Gamma(\pi'_\mu,\sigma'_\mu)$ of $\pert{\mu}$ (see \cref{fig:Canonical-small}) and insert it, possibly after squeezing it, into $f_\mu$ in such a way that its spine lays entirely on the line of the spine of $\Gamma_{\overline{\mu}}$ and in such a way that the vertices of $\mu$ do not fall in between the vertices of any maximal sequence $\pi_i$, $i=1, \dots, k$. This is always possible since $k \geq 3$ and, therefore, there exists a portion of the spine of $\Gamma_{\overline{\mu}}$ which is in the interior of $f_\mu$. 
Observe that this implies that $s_\mu$ and $t_\mu$ have now a double representation, since the drawing of the source $s'_\mu$ and sink $t'_\mu$ of $\Gamma(\pi'_\mu,\sigma'_\mu)$ do not coincide with the drawing of $s_\mu$ and $t_\mu$ in $\Gamma_{\overline{\mu}}$. Denote by $\Gamma^*$ the resulting drawing.

Suppose $\langle \pi_\mu, \sigma_\mu \rangle$ (and, hence, also $\langle \pi'_\mu, \sigma'_\mu \rangle$) has Type $\langle x, y, z \rangle$. Observe that if $y \in\{N, R\}$ then there are no left vertices of $f_\mu$. Otherwise, if $y \in \{L,B\}$, then it is possibile to identify two vertices $v_{\ell,b}$ and $v_{\ell,t}$ of $\Gamma(\pi'_\mu,\sigma'_\mu)$ such that a portion of the spine is visible from the left between $v_{\ell,b}$ and $v_{\ell,t}$.  
Move between $v_{\ell,b}$ and $v_{\ell,t}$ on the spine in $\Gamma^*$ all the vertices of the sequences $\pi_i$ whose starting and ending vertices are left vertices of $f_\mu$, preserving their relative order. Analogously, observe that if $y \in\{N, L\}$, then there are no right vertices of $f_\mu$. Otherwise, if $y\in\{R, B\}$, then it is possibile to identify two vertices $v_{r,b}$ and $v_{r,t}$ of $\Gamma(\pi'_\mu,\sigma'_\mu)$ such that a portion of the spine is visible from the left between $v_{r,b}$ and $v_{r,t}$. 
Move between $v_{r,b}$ and $v_{r,t}$ on the spine in $\Gamma^*$ all the vertices of the sequences $\pi_i$ whose starting and ending vertices are right vertices of $f_\mu$, preserving their relative order. Refer to \cref{fig:Disconnected-embedding}.

Observe that in $\Gamma^*$ the edges of $G$ lie in the pages prescribed by $\sigma'$ by construction. Also, in $\Gamma^*$ the bottom-to-top order of the vertices is the same as $\pi'$ with the exception of the two duplicates vertices $s'_\mu$ and $t'_\mu$. Thus, by identifying $s_\mu$ with $s'_\mu$ and $t_\mu$ with $t'_\mu$, we obtain the promised upward drawing $\Gamma'$ of $G$ in which each vertex $v$ lies along the spine in the same bottom-to-top order determined by $\pi'$ and each edge is drawn as a $y$-monotone curve on the page assigned by $\sigma'$. To complete the proof we argue about the planarity of $\Gamma'$. Clearly, identifying $s_\mu$ with $s'_\mu$ ($t_\mu$ with $t'_\mu$) does not introduce crossings when $s_\mu$ with $s'_\mu$ ($t_\mu$ with $t'_\mu$) are consecutive along the spine in $\Gamma^*$ (see $t'_\mu$ and $t_\mu$ in \cref{fig:Disconnected-embedding}). In the case in which $s_\mu$ and $s'_\mu$ ($t_\mu$ and $t'_\mu$) are not consecutive along the spine in $\Gamma^*$, necessarily $x \in \{L,R\}$ ($z \in \{L,R\}$). Suppose $x = L$ ($z=L$), the case when $x = R$ ($z=R$) being analogous. There cannot exist in $\Gamma_{\overline{\mu}}$ an edge $e=(a,b)$ such that $\pi(a) < \pi(s_\mu) < \pi(b)$ and such that $\sigma(e)=R$. Therefore, we can continuously move $s'_\mu$, together with its incident edges, toward $s_\mu$, remaining inside the region bounded by $f_\mu$ without intersecting any edge of $G_{\overline{\mu}}$ (see \cref{fig:Final-embedding}). This concludes the proof. 
\end{proof}

\subsection{S-nodes}\label{apx:s-node}

\lemsdecision*

\begin{proof}
 Let $\mu$ be an S-node with poles $s_{\mu}$ and $t_{\mu}$. Let $\mu'$ and $\mu''$ be the two children of $\mu$ with poles $s_{\mu'}$, $t_{\mu'}$ and $s_{\mu''}$, $t_{\mu''}$, respectively, where $s_{\mu'}=s_{\mu}$, $t_{\mu'}=s_{\mu''}$, and $t_{\mu''}=t_{\mu}$. Clearly, combining each pair of \UBEs of the two children of $\mu$ always yields a \UBE of $\pert{\mu}$. 
Let $\langle x',y',x' \rangle$ and $\langle x'',y'',z'' \rangle$ be the embedding types of \UBEs of $\pert{\mu'}$ and $\pert{\mu''}$, respectively. The embedding type $\langle x, y, z \rangle$ of the \UBE of $\pert{\mu}$ resulting from combining such \UBEs can be computed as follows. We have $x = x'$ and $z = z'$. As for $y$, we have that: \begin{inparaenum}[(i)] \item $y = N$ iff $y'=y''=N$; \item $y=L$ if either $y'=L$ and $y'' \in \{L,N\}$ or $y''=L$ and $y' \in \{L,N\}$; \item $y=R$ if either $y'=R$ and $y'' \in \{R,N\}$ or $y''=R$ and $y' \in \{R,N\}$; \item $y=B$ if at least one of $y'$ and $y''$ is $B$ or one of them is $L$ and the other is $R$.\end{inparaenum}~Since the number of embedding types realizable by the two children of $\mu$ is constant, the statement follows.
\end{proof}

\subsection{P-nodes: Fixed Embedding}\label{apx:p-case-fixed}

Let $\mu$ be a P-node with poles $s_\mu$ and $t_\mu$ and let $\mu_1,\mu_2,\dots,\mu_k$ be the children of $\mu$ in the left-to-right order defined by the given embedding~of~$G$.

Let $\skel^i(\mu)$ and $\mypert^i(\mu)$ be the subgraphs of $\skel(\mu)$ and of $\pert{\mu}$, respectively, determined by children $\mu_1,\mu_2,\dots,\mu_i$ of $\mu$. 
Clearly, the embedding types that are realizable by $\mypert^1(\mu)$ are those that are realizable by $\pert{\mu_1}$.
For $i=2,\dots,k$, we can compute the set $X_i$ of embedding types that are realizable by $\mypert^i(\mu)$ as follows.  We consider each embedding type $\langle x,y,z \rangle \in X_{i-1}$ and each embedding type $\langle a,b,c\rangle$ that is realizable by $\pert{\mu_i}$, and we either determine that no embedding type can be obtained by composing the pair $\langle x,y,z \rangle$ and $\langle a,b,c\rangle$ or compute an embedding type $\langle p,q,w\rangle$ for $\mypert^i(\mu)$, which we add to $X_i$, as follows.
\begin{description}
\item[Rejection:] If $x \in \{N,L\}$ and $a \in \{N,R\}$, then we reject the pair.
Similarly, if $z \in \{N,L\}$ and $c \in \{N,R\}$, then we reject the pair.
If $y=L$, then we reject the instance, as $\pert{\mu_i}$ contains at least an internal vertex that has to be placed on the spine to the left of $\mypert^{i-1}(\mu_i)$ and between $s_\mu$ and $t_\mu$.
\item[Embedding Types for $\mypert^i{\mu}$:]
If $x=a$, then $p=x$, and if $x = N$ or $a=N$, then $p=N$. 
Similarly, if $z=c$, then $w=z$, and if $z = N$ or $c=N$, then $w=N$. 
If $y=L$ and $b \in \{R,B\}$, then $q=B$; also, if $y=L$ and $b = L$, then $q=L$. 
Similarly, if $y=R$ and $b \in \{L,B\}$, then $q=B$; also, if $y=R$ and $b = R$, then $q=R$.
Finally, if $y=B$ and $b  \in \{R,B\}$, then $q=B$; if $y=B$ and $b=L$, then $q=L$. 
\end{description}

Since the set of rules above defines the set $X_k$ of embedding types realizable by $\pert{\mu}$ and the above computations can easily be performed in time linear in the number of children of $\mu$, we obtain \cref{lem:p-decision-fixed}.

\subsection{P-nodes: Variable Embedding}\label{see:apex-p-node}

In this section, we give necessary and sufficient conditions under which a P-node $\mu$ admits one of the relevant embedding types (enclosed by a solid polygon in \cref{fig:types}) in the variable embedding setting.
We start with the conditions for Type \NRR and discuss the time complexity of testing such conditions.

\pnrr*

\begin{proof}
The necessity can be argued by considering that all children must be of Type $\langle \cdot ,\cdot, R \rangle$ and at most one child can be of Type $\langle N ,\cdot, \cdot \rangle$.

For the sufficiency, we consider the two cases separately. 
In {\bf ({Case}~1)} we order the children of $\mu$ so that from left to right we have the possible Q-node child, if any, followed by the unique Type-\NRR child, followed by the Type-\LBR children in any order; see \cref{fig:P-NRR-C1}.
In {\bf ({Case}~2)} we order the children of $\mu$ so that from left to right we have the possible Type-\QRRR Q-node child, if any, followed by the possible non-Q-node Type-\RRR child, followed by the Type-\RBR children in any order, followed by the possible Type-\NBR, if any, followed by the Type-\LBR children in any order, if any; see \cref{fig:P-NRR-C2}. 

We show that the ordering $\mu_1,\mu_2,\dots,\mu_k$ of the children of $\mu$ described in the two cases above determines a \UBE of $\pert{\mu}$. 

Denote with $\langle \pi_i,\sigma_i \rangle$ the \UBE of $\pert{\mu_i}$, where $\pi_i$ is the bottom-to-top ordering of the internal vertices of $\pert{\mu_i}$ in the \UBE and where $\sigma_i:E(\pert{\mu_i}) \rightarrow \{R,L\}$ is the assignment of the edges of $\pert{\mu_i}$ to the two pages of the \UBE. 
In {\bf ({Case}~1)} we construct a \UBE $\langle \pi,\sigma \rangle$ of $\pert{\mu}$ as follows: see \cref{fig:P-NRR-C1} for an example. The bottom-to-top ordering $\pi$ of the vertices of $\pert{\mu}$ is $s_\mu, \pi_1, \pi_2,\dots,\pi_k, t_\mu$. The assignment  $\sigma$ of the edges of $\pert{\mu}$ to the two pages is the one determined by the $\sigma_i$`s, that is, $\sigma(e)=\sigma_i(e)$ if $e \in \pert{\mu_i}$. Further, if edge $(s_\mu, t_\mu)$ exists, then $\sigma \big((s_\mu, t_\mu)\big)=L$. Clearly, $\langle \pi,\sigma \rangle$ is of a \UBE of $\pert{\mu}$ of Type \NRR.
In {\bf ({Case}~2)} we construct a \UBE $\langle \pi,\sigma \rangle$ of $\pert{\mu}$ as follows: see \cref{fig:P-NRR-C2} for an example.
Observe that the children of $\mu$ of Type~\RBR or Type~\NBR determine a consecutive subsequence $\mu_p, \dots, \mu_q$ of the left-to-right sequence of the children of $\mu$, with $2 \leq p \leq q \leq k$. For each child $\mu_i$, $i=p,\dots, q$, by the definition of Type~\RBR and Type~\NBR, there exist two consecutive vertices $v'$ and $v''$ in $\pi_i$ such that portion of the spine between them is visible from the left. Therefore, we can split $\pi_i$ into two subsequences $\pi'_i$ and $\pi''_i$ where $v' \in \pi'_i$ and $v'' \in \pi''_i$.  
We construct $\langle \pi,\sigma \rangle$ as follows. Let $\mu^*$ be the possible Type-\RRR child of $\mu$.
The bottom-to-top ordering $\sigma$ of the vertices of $\pert{\mu}$ in the \UBE is $s_\mu, \pi'_p, \pi'_{p+1}, \dots, \pi'_q, \pi^*, $ $\pi''_q, \pi''_{q-1}, \dots, \pi''_p, \pi_{q+1}, \dots \pi_k, t_\mu$, where $\pi^* = \emptyset$ if $\mu^*$ does not exist. The assignment $\sigma$ of the edges of $\pert{\mu}$ to the pages of the \UBE  is the one determined by the $\sigma_i$`s. Further, if edge $(s_\mu, t_\mu)$ exists, then $\sigma \big((s_\mu, t_\mu)\big)=L$. Clearly, $\langle \pi,\sigma \rangle$ is of a \UBE of $\pert{\mu}$ of Type \NRR.
\end{proof}

Regarding the computational complexity of deciding if one of {\bf ({Case}~1)} or {\bf ({Case}~2)} of \cref{le:p-NRR} applies, we show that such a task can be reduced to a network flow problem on a network $\mathcal{N}$ with edge demands.

We describe the construction of $\mathcal{N}$ for {\bf ({Case}~2)} (refer to \cref{fig:P-NRR-C2,fig:P-NRR-C2-NetWork}); the construction for {\bf ({Case}~1)} being similar.  
Network $\mathcal{N}$ is a capacitated flow network where each arc $e$ has a label $[l,u]$, where $l$ is a lower bound and $u$ is an upper bound on the flow traversing $e$ in a feasible flow. In particular, $\mathcal{N}$ contains a source $s$; a sink $t$; a node $\mu_i$ for each child $\mu_i$ of $\mu$; a node $t \in \{\QRRR, \RRR, \RBR, \NBR, \LBR\}$, representing each embedding type used in constructing the sequence of {\bf ({Case}~2)}; and two special nodes $\nu_1$ and $\nu_2$. Network $\mathcal{N}$ contains the following arcs: $(s,\mu_i)$ with label $[1,1]$, for $i = 1, \dots, k$; an arc $(\mu_i,t)$ with label $[0,1]$, for each type $t$ for which $\pert{\mu_i}$ admits a Type-$t$ \UBE; an arc $(\QRRR,\nu_1)$ with label $[0,1]$; an arc $(\RRR,\nu_1)$ with label $[0,1]$; an arc $(\nu_1,t)$ with label $[1,2]$; an arc $(\RBR,t)$ with label $[0,\infty]$; an arc $(\RBR,\nu_2)$ with label $[0,1]$; an arc $(\LBR,\nu_2)$ with label $[0,\infty]$; and an arc $(\nu_2,t)$ with label $[1,\infty]$. Observe that $\mathcal{N}$ has size linear in the number $k$ of children of $\mu$ due to the fact that the outdegree of $\mu_i$ nodes is bounded by the number of embedding types, $s$ has outdegree $k$, and the remaining nodes have outdegree at most~$1$.

Clearly, a sequence corresponding to {\bf ({Case}~2)} exists if and only if the network $\mathcal{N}$ admits a feasible flow, which has value $k$. Testing the existence of a feasible flow can be reduced in linear time to a maxflow problem in a suitable capacitated network $\mathcal{N}'$~\cite{DBLP:books/daglib/0015106}. We obtain a solution by applying the standard max-flow algorithm by Ford–Fulkerson, which runs in $O(|E(\mathcal{N})|\times f)$, where $f$ is the value of the maximum flow. 

\medskip
Next, we give the conditions for the remaining relevant embedding types \RRR, \NRN, \RBR, \LBR, \NBR, \RRN, \RBN, \NBN, and \NNN, whose testing can be carried out with the same algorithmic strategy as for Type \NRR.
Altogether, we obtain \cref{lem:p-decision-variable}.

\begin{lemma}[Type \protect\LBR]\label{le:p-LBR}
Let $\mu$ be a P-node. Type \LBR is admitted by $\mu$ in the variable embedding setting
if and only if all of its children admit Type \LBR. 
\end{lemma}

\begin{proof}
For the necessity, consider any \UBE  $\mathcal E$ of $\pert{\mu}$ of Type \LBR. Let $\mu^*$ be any child of $\mu$ and let $\mathcal{E}^*$ be the \UBE of $\pert{\mu^*}$ obtained by restricting $\mathcal E$ to $\pert{\mu^*}$. It is immediate that
there exist a portion of the spine incident to $s_\mu$ and between $s_\mu$ and $t_\mu$ that is visible from the left page and a portion of the spine incident to $t_\mu$ and between $s_\mu$ and $t_\mu$ that is visible from the right page. Thus, $\mu$ admits Type \LBR.
For the sufficiency, we order the children of $\mu$ arbitrarily; see \cref{fig:P-LBR} for an example. Clearly, the resulting embedding of $\pert{\mu}$ is of Type \LBR.
\begin{figure}[h]
\centering
\begin{subfigure}[b]{.24\textwidth}
  \centering
  \includegraphics[page=1]{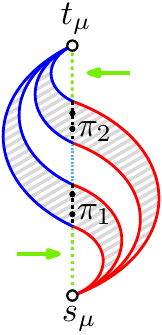}
  \subcaption{\LBR}
  \label{fig:P-LBR}
\end{subfigure}
\begin{subfigure}[b]{.24\textwidth}
  \centering
  \includegraphics[page=2]{P-node}
  \subcaption{\RBR}
  \label{fig:P-RBR}
\end{subfigure}
\begin{subfigure}[b]{.24\textwidth}
  \centering
  \includegraphics[page=3]{P-node}
  \subcaption{\RRR; Case~1}
  \label{fig:P-RRR-C1}
\end{subfigure}
\begin{subfigure}[b]{.24\textwidth}
  \centering
  \includegraphics[page=4]{P-node}
  \subcaption{\RRR; Case~2}
  \label{fig:P-RRR-C2}
\end{subfigure}
\caption{Constructions for P-nodes of Type \protect\LBR, \protect\RBR, and \protect\RRR.}
\end{figure}
\end{proof}

\begin{lemma}[Type \protect\RBR]\label{le:p-RBR}
Let $\mu$ be a P-node. Type \RBR is admitted by $\mu$ in the variable embedding setting if and only if all of its children admit Type \RBR. 
\end{lemma}

\begin{proof}
Both the necessity and sufficiency proofs follow the same line as the proofs for \cref{le:p-LBR}; see \cref{fig:P-RBR} for an example.
\end{proof}

\begin{lemma}[Type \protect\RRR]\label{le:p-RRR}
Let $\mu$ be a P-node. Type \RRR is admitted by $\mu$ in the variable embedding setting if and only if
one of these two cases occurs:
\begin{inparaenum}[\bf ({Case}~1)]
\item There exists a Q-node $\omega$ corresponding to edge $(s_\mu , t_\mu)$ between the poles of $\mu$, all the other children of $\mu$ except for at most one node $\xi$ admit Type \RBR, and $\xi$, if it exists, admits Type \RRR.
\item Edge $(s_\mu , t_\mu)$ between the poles of $\mu$ does not exist, there exists a child $\xi$ of $\mu$ admitting Type \RRR, and all the other children of $\mu$, if any, admit Type \RBR.
\end{inparaenum}
\end{lemma}

\begin{proof}
It is easy to see that these conditions are necessary: in fact, for any child $\mu_i$ of $\mu$, the restriction of a Type-\RRR embedding of $\pert{\mu}$ to $\pert{\mu_i}$ is clearly of Type $\langle R, X, R\rangle$, with $X \in \{R,B\}$. Also, an embedding of Type \RRR for $\pert{\mu}$ is incompatible with more than one embedding of Type \RRR for its children unless one of them is a single edge. If no child has embedding of Type \RRR, then $\mu$ would have Type \RBR.
For the sufficiency, we consider the two cases separately. 
Let $\mu_1, \dots, \mu_k$ be the non-Q-node children of $\mu$.
In {\bf ({Case}~1)} we order the children of $\mu$ so that from left to right we have node $\omega$, followed by $\xi$, if any, followed by the remaining children of $\mu$; see \cref{fig:P-RRR-C1}.  
In {\bf ({Case}~2)} we order the children of $\mu$ so that from left to right we have node $\xi$ followed by the remaining children of $\mu$; see \cref{fig:P-RRR-C2}. 
\end{proof}

\begin{lemma}[Type \protect\NBR]\label{le:p-NBR}
Let $\mu$ be a P-node. Type \NBR is admitted by $\mu$ in the variable embedding setting if and only if
at least one of these two cases occurs:
\begin{inparaenum}[\bf ({Case}~1)]
\item All children admit either Type \LBR or Type \RBR, and there exists at least a Type-\LBR child and a Type-\RBR child.
\item There exists a Type-\NBR child and all other children admit either Type \LBR or Type \RBR. 
\end{inparaenum}
\end{lemma}

\begin{proof}
It is easy to see that these conditions are necessary. In fact, for any child $\mu_i$ of $\mu$, the restriction of a Type-\NBR embedding of $\pert{\mu}$ to $\pert{\mu_i}$ is clearly of Type $\langle X, Y, R\rangle$, with $X \in \{R,L,N\}$ and $Y = B$ since an embedding of $\mu$ with $pocket_{side}=B$ implies an embedding of all its children with $pocket_{side}=B$. 

For the sufficiency, in both cases we order the children from left to right in such a way that Type-\RBR children precede all the Type-\LBR children; see \cref{fig:P-NBR-C1,fig:P-NBR-C2}. In {\bf ({Case}~2)} the \NBR child is placed between the last Type-\RBR child e the first Type-\LBR child; see \cref{fig:P-NBR-C2}.  
\begin{figure}[h]
\centering
\begin{subfigure}[b]{.24\textwidth}
  \centering
  \includegraphics[page=5]{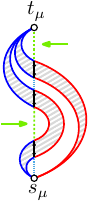}
  \subcaption{\NBR; Case~1}
  \label{fig:P-NBR-C1}
\end{subfigure}
\begin{subfigure}[b]{.24\textwidth}
  \centering
  \includegraphics[page=6]{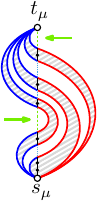}
  \subcaption{\NBR; Case~2}
  \label{fig:P-NBR-C2}
\end{subfigure}
\caption{Constructions for P-nodes of Type \protect\NBR.}
\end{figure}
\end{proof}

\begin{lemma}[Type \protect\RBN]\label{le:p-RBN}
Let $\mu$ be a P-node. Type \RBN is admitted by $\mu$ in the variable embedding setting if and only if
at least one of two cases occurs, which can be obtained from {\bf ({Case}~1)} and {\bf ({Case}~2)} discussed for Type \NBR (\cref{le:p-NBR}) by considering the types of its children after a vertical flip, i.e., we replace each Type-$\langle X,Y,Z \rangle$ child with a Type-$\langle Z,Y,X \rangle$ child.
\end{lemma}

\begin{proof}
The necessity and sufficiency proofs are symmetric to those presented in \cref{le:p-NBR}.
\end{proof}

\begin{lemma}[Type \protect\NBN]\label{le:p-NBN}
Let $\mu$ be a P-node. Type \NBN is admitted by $\mu$ in the variable embedding setting if and only if
at least one of six cases occur. Since Type \NBN is self-symmetric, we have that, for each case, we can obtain a symmetric one by reversing the left-to-right sequence of the children in the construction and by taking, for each child, the horizontally-mirrored embedding Type. Thus, for each pair of symmetric cases, we only describe one.
\begin{inparaenum}[\bf ({Case}~1)]
\item There exists a child that admits Type \NBN and the other children admit either Type \RBR or Type \LBL.
\item It is possible to partition the children of $\mu$ into three parts as follows. The first part consists of any positive number of Type-\RBR or Type-\NBR children, with at most one Type-\NBR child. 
The second part consists of any number, even zero, of Type-\LBR children. The third part consists of any positive number of Type-\LBN or Type-\LBL children, with at most one Type-\LBN child.
\item This case is obtained from {\bf Case~2} by considering the types of its children after a vertical flip.
\end{inparaenum}
\end{lemma}

\begin{proof}
For the sufficiency, we consider the three cases separately. 
In {\bf ({Case}~1)} we order the children of $\mu$ so that from left to right we have the Type-\RBR children, if any, in any order, followed by the unique Type-\NBN child, followed by the Type-\LBL children, if any, in any order; see \cref{fig:P-NBN-C1}.
In {\bf ({Case}~2)} we order the children of $\mu$ so that from left to right we have the possible Type-\RBR children, if any, in any order, followed by the possible Type-\NBR child, if any, followed by the possible Type-\LBR, if any, in any order, followed by the possible Type-\LBN child, if any, followed by the possible Type-\LBL children, if any, in any order (see \cref{fig:P-NBN-C2}). 
In {\bf ({Case}~3)} we proceed as in {\bf ({Case}~2)} by considering the types after a vertical flip.
Clearly, these orderings allow for a \UBE of $\pert{\mu}$.

The necessity can be argued by considering that at most one child can be of Type $\langle N ,\cdot, \cdot \rangle$, at most one child can be of Type $\langle \cdot ,\cdot, N \rangle$, and any child must be of Type $\langle \cdot ,B, \cdot \rangle$.
\begin{figure}[t!]
\centering
\begin{subfigure}[b]{.24\textwidth}
  \centering
  \includegraphics[page=9]{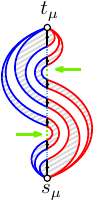}
  \subcaption{\NBN; Case~1}
  \label{fig:P-NBN-C1}
\end{subfigure}\hfil
\begin{subfigure}[b]{.24\textwidth}
  \centering
  \includegraphics[page=10]{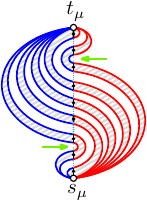}
  \subcaption{\NBN; Case~2}
  \label{fig:P-NBN-C2}
\end{subfigure}
\caption{Constructions for P-nodes of Type \protect\NBN.}
\end{figure}
\end{proof}

\begin{lemma}[Type \protect\RRN]\label{le:p-RRN}
Let $\mu$ be a P-node. Type \RRN is admitted by $\mu$ in the variable embedding setting if and only if
at least one of two cases occurs, which can be obtained from {\bf ({Case}~1)} and {\bf ({Case}~2)} discussed for Type \NRR (\cref{le:p-NRR}) by considering the types of its children after a vertical flip.
\end{lemma}

\begin{proof}
The necessity and sufficiency proofs are symmetric to those presented in \cref{le:p-NRR}.
\end{proof}

\begin{lemma}[Type \protect\NRN]\label{le:p-NRN}
Let $\mu$ be a P-node. Type \NRN is admitted by $\mu$ in the variable embedding setting if and only if
at least one of these six cases occurs:
\begin{inparaenum}[\bf ({Case}~1)]
\item There exists a child that admits Type \NRN and the other non-Q-node children admit Type \LBL.
\item It is possible to partition the children of $\mu$ into three parts as follows. The first part consists of at least one child of Type \RRR of which at most one is not a Q-node.
The second part consists of any number, even zero, of Type-\RBR children. The third part consists of any positive number of Type-\NBN or Type-\LBL children, with at most one Type-\NBN child.
\item It is possible to partition the children of $\mu$ into three parts as follows. The first part consists of at least one and at most two children of distinct types, namely, a Type-\QRRR Q-node and a Type-\RRN child.
The second part consists of any number, even zero, of Type-\RBL children. The third part consists of any positive number of Type-\NBL or Type-\LBL children, with at most one Type-\NBL child.
\item This case is obtained from the previous one by considering the types of the children of $\mu$ after a vertical flip.
\item It is possible to partition the children of $\mu$ into three parts as follows. The first part consists of at least one child of Type \RRR of which at most one is not a Q-node.
The second part consists of any number, even zero, of children of Type \RBN or Type \RBL, with at most  one Type-\RBN child.
The third part consists of any positive number of Type-\NBL or Type-\LBL children, with at most one Type-\NBL child.
\item This case is obtained from the previous one by considering the types of the children of $\mu$ after a vertical flip.
\end{inparaenum}
\end{lemma}

\begin{proof}
For the sufficiency, we consider the six cases separately. 
In {\bf ({Case}~1)} we order the children of $\mu$ so that from left to right we have the Type-\RRR  Q-node child, if any, followed by the unique Type-\NRN child, followed by the Type-\LBL children, if any, in any order; see \cref{fig:P-NRN-C1}.
In {\bf ({Case}~2)} we order the children of $\mu$ so that from left to right we have the Type-\RRR  Q-node child, if any, followed by the Type-\RRR non-Q-node child, if any, followed by the Type \RBR children, if any, in any order, followed by the unique Type-\NBN child, if any, followed by the Type-\LBL children, if any, in any order; see \cref{fig:P-NRN-C2}.
In {\bf ({Case}~3)} we order the children of $\mu$ so that from left to right we have the Type-\RRR  Q-node child, if any, followed by the Type-\RRN child, if any, followed by the Type \RBL children, if any, in any order, followed by the unique Type-\NBL child, if any, followed by the Type-\LBL children, if any, in any order; see \cref{fig:P-NRN-C3}.
In {\bf ({Case}~4)} we proceed as in {\bf ({Case}~3)} by considering the types after a vertical flip.
In {\bf ({Case}~5)} we order the children of $\mu$ so that from left to right we have the Type-\QRRR Q-node child, if any, followed by the Type-\RRR non-Q-node child, if any, followed by the Type-\RBN child, if any, followed by the Type \RBL children, if any, in any order, followed by the unique Type-\NBL child, if any, followed by the Type-\LBL children, if any, in any order; see \cref{fig:P-NRN-C5}.
In {\bf ({Case}~6)} we proceed as in {\bf ({Case}~5)} by considering the types after a vertical flip.

The necessity can be argued by considering that at most one child can be of Type $\langle N ,\cdot, \cdot \rangle$, at most one child can be of Type $\langle \cdot ,\cdot, N \rangle$, and that any child must be of Type $\langle \cdot ,R, \cdot \rangle$ or of Type $\langle \cdot ,B, \cdot \rangle$. However, not all the children of $\mu$ may be of Type $\langle \cdot ,B, \cdot \rangle$ as otherwise the resulting embeddings would be of Type $\langle \cdot ,B, \cdot \rangle$.
\begin{figure}[h]
\centering
\begin{subfigure}[b]{.24\textwidth}
  \centering
  \includegraphics[page=11]{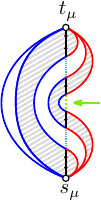}
  \subcaption{\NRN; Case~1}
  \label{fig:P-NRN-C1}
\end{subfigure}
\begin{subfigure}[b]{.24\textwidth}
  \centering
  \includegraphics[page=12]{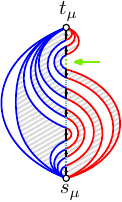}
  \subcaption{\NRN; Case~2}
  \label{fig:P-NRN-C2}
\end{subfigure}
\begin{subfigure}[b]{.24\textwidth}
  \centering
  \includegraphics[page=13]{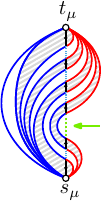}
  \subcaption{\NRN; Case~3}
  \label{fig:P-NRN-C3}
\end{subfigure}
\begin{subfigure}[b]{.24\textwidth}
  \centering
  \includegraphics[page=14]{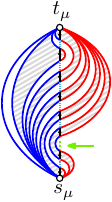}
  \subcaption{\NRN; Case~5}
  \label{fig:P-NRN-C5}
\end{subfigure}
\caption{Constructions for P-nodes of Type \protect\NRN.}
\end{figure}
\end{proof}

\begin{lemma}[Type \protect\NNN]\label{le:p-NNN}
Let $\mu$ be a P-node. Type \NNN is admitted by $\mu$ in the variable embedding setting if and only if
at least one of $26$ cases occurs. Since Type \NNN is self-symmetric, for each pair of symmetric cases, we only describe one. Further, we assume that Case~$i$ applies only if Case~$1$, Case~$2$, \dots, Case~$i-1$ do not apply.
\begin{inparaenum}[\bf ({Case}~1)]
\item There exist exactly two children, one of which is a Type-\QRRR Q-node and the other is a Type-\NNN node.
\item There exists exactly one Type-\NRN child, at least one and at most two Type-\RRR children of which at most one is a non-Q-node child, and all the remaining children are of Type \LBL.
\item It is possible to partition the children of $\mu$ into five parts as follows. The first part consists of at least one and at most two Type-\RRR children of which at most one is a non-Q-node child. The second part consists of any number of Type-\RBR children. The third part consists of exactly one Type-\NBN child. The fourth part consists of any number of Type-\LBL children. The fifth part consists of at least one and at most two Type-\LLL children of which at most one is a non-Q-node child.
\item There exists exactly one Type-\RRN child, any number of Type-\RBL and Type-\LBL children, at most one Type-\NBL child, and at least one and at most two Type-\LLL children of which at most one is a non-Q-node child. 
\item This case is obtained from the previous one by considering the types of the children of $\mu$ after a vertical flip.
\item There exists exactly one Type-\RRN child, any number of Type-\RBL children, at most one Type-\NLL child, and at least one child that is either a Type-\NLL child or it is a Q-node Type-\LLL child.
\item This case is obtained from the previous one by considering the types of the children of $\mu$ after a vertical flip.
\item There exists exactly one Type-\RBN child, at most one Type-\NBL child, any number of Type-\RBR, Type-\RBL, and Type-\LBL children. Further, there exist at least one and at most two non-Q-node children of distinct types between Type \RRR and Type \LLL, where if only one of such children exists then $\mu$ has also a Q-node child.
\item This case is obtained from the previous one by considering the types of the children of $\mu$ after a vertical flip.
\item There exists exactly one Type-\RBN child, any number of Type-\RBR and Type-\RBL children. Further, there exist at least one and at most two non-Q-node children of distinct types between Type \RRR and Type \NLL, where if only one of such children exists then $\mu$ has also a Q-node child.
\item This case is obtained from the previous one by considering the types of the children of $\mu$ after a vertical flip.
\item There exists exactly one non-Q-node Type-\RRR child, any number of Type-\RBR and Type-\LBR children, and at least one and at most two Type-\LLL children of which at most one is a non-Q-node child.
\item This case is obtained from the previous one by considering the types of the children of $\mu$ after a vertical flip.
\end{inparaenum}
\end{lemma}

\begin{proof}
\begin{figure}[t]
\centering
\begin{subfigure}[b]{.24\textwidth}
  \centering
  \includegraphics[page=15]{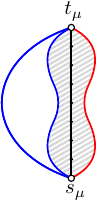}
  \subcaption{\NNN; Case~1}
  \label{fig:P-NNN-C1}
\end{subfigure}
\begin{subfigure}[b]{.24\textwidth}
  \centering
  \includegraphics[page=16]{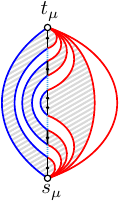}
  \subcaption{\NNN; Case~2}
  \label{fig:P-NNN-C2}
\end{subfigure}
\begin{subfigure}[b]{.24\textwidth}
  \centering
  \includegraphics[page=17]{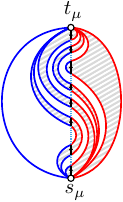}
  \subcaption{\NNN; Case~3}
  \label{fig:P-NNN-C3-a}
\end{subfigure}
\begin{subfigure}[b]{.24\textwidth}
  \centering
  \includegraphics[page=18]{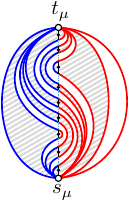}
  \subcaption{\NNN; Case~3}
  \label{fig:P-NNN-C3-b}
\end{subfigure}
\caption{Constructions for P-nodes of Type \protect\NNN. Subfigures (c) and (d) show the constructions for Case~3 when there exists no Type-\protect\RRR non-Q-node child and when there exists such a child, respectively.}
\end{figure}
We consider the three cases separately. 
In {\bf ({Case}~1)} we order the children of $\mu$ so that from left to right we have the Type-\QRRR Q-node child followed by the Type-\NNN child; see \cref{fig:P-NNN-C1}.
In {\bf ({Case}~2)} we order the children of $\mu$ so that from left to right we have the Type-\NRN child, followed by the Type-\LBL children, if any, in any order, followed by the Type-\LLL non-Q-node child, if any, followed by the Type-\QLLL Q-node child, if any; see \cref{fig:P-NNN-C2}.
In {\bf ({Case}~3)} we proceed as follows. If there exists no Type-\RRR non-Q-node child, we
order the children of $\mu$ so that from left to right we have the Type-\QRRR Q-node child (which exists by the conditions of the case), followed by the Type-\RBR children, if any, in any order, 
followed by the Type-\NBN child, followed by the Type-\LBL children, if any, in any order, followed by the Type-\LLL non-Q-node child (which exists by the conditions of the case); see \cref{fig:P-NNN-C3-a}. Otherwise, we
order the children of $\mu$ so that from left to right we have the Type-\RRR non-Q-node child, followed by the Type-\RBR children, if any, in any order,  followed by the Type-\NBN child, followed by the Type-\LBL children, if any, in any order, followed by the Type-\LLL non-Q-node child, if any, followed by the Type-\QLLL Q-node child, if any; see \cref{fig:P-NNN-C3-b}.
In {\bf ({Case}~4)} we order the children of $\mu$ so that from left to right we have the Type-\RRN child, followed by the Type-\RBL children, if any, in any order, followed by the Type-\NBL child, if any, followed by the Type-\LBL children, if any, in any order, followed by the Type-\LLL non-Q-node child, if any, followed by the Type-\QLLL Q-node child, if any; see \cref{fig:P-NNN-C4}.
\begin{figure}[t]
\centering
\begin{subfigure}[b]{.24\textwidth}
  \centering
  \includegraphics[page=19]{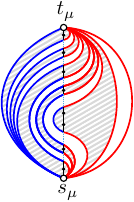}
  \subcaption{\NNN; Case~4}
  \label{fig:P-NNN-C4}
\end{subfigure}
\begin{subfigure}[b]{.24\textwidth}
  \centering
  \includegraphics[page=20]{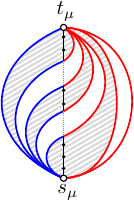}
  \subcaption{\NNN; Case~6}
  \label{fig:P-NNN-C6}
\end{subfigure}
\begin{subfigure}[b]{.24\textwidth}
  \centering
  \includegraphics[page=21]{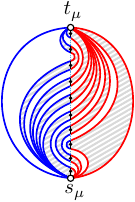}
  \subcaption{\NNN; Case~8}
  \label{fig:P-NNN-C8-a}
\end{subfigure}
\begin{subfigure}[b]{.24\textwidth}
  \centering
  \includegraphics[page=22]{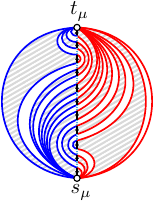}
  \subcaption{\NNN; Case~8}
  \label{fig:P-NNN-C8-b}
\end{subfigure}
\caption{Constructions for P-nodes of Type \protect\NNN. Subfigures (c) and (d) show the constructions for Case~8 when there exists no Type-\protect\RRR non-Q-node child and when there exists such a child, respectively.}
\end{figure}
In {\bf ({Case}~5)} we proceed as in {\bf ({Case}~4)} by considering the types after a vertical flip.
In {\bf ({Case}~6)} we order the children of $\mu$ so that from left to right we have the Type-\RRN child, followed by any number of Type-\RBL children, followed by the Type-\NLL child, if any, followed by the Type-\QLLL Q-node child, if any; see \cref{fig:P-NNN-C6}.
In {\bf ({Case}~7)} we proceed as in {\bf ({Case}~6)} by considering the types after a vertical flip.
In {\bf ({Case}~8)} we proceed as follows. If there exists no Type-\RRR non-Q-node child, we order the children of $\mu$ so that from left to right we have the Type-\QRRR Q-node child (which exists by the conditions of the case), followed by the Type-\RBR children, if any, in any order, followed by the Type-\RBN child, followed by the Type-\RBL children, if any, in any order, followed by the Type-\NBL child, if any, followed by the Type-\LBL children, if any, in any order, followed by the Type-\LLL non-Q-node child (which exists by the conditions of the case); see \cref{fig:P-NNN-C8-a}. Otherwise, we order the children of $\mu$ so that from left to right we have the Type-\RRR non-Q-node child, followed by the Type-\RBR children, if any, in any order, followed by the Type-\RBN child, followed by the Type-\RBL children, if any, in any order, followed by the Type-\NBL child, if any, followed by the Type-\LBL children, if any, in any order, followed by the Type-\LLL non-Q-node child, if any, followed by the Type-\QLLL Q-node child, if any; see \cref{fig:P-NNN-C8-b}.
In {\bf ({Case}~9)} we proceed as in {\bf ({Case}~8)} by considering the types after a vertical flip.
\begin{figure}[b]
\centering
\begin{subfigure}[b]{.24\textwidth}
  \centering
  \includegraphics[page=23]{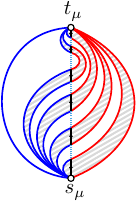}
  \subcaption{\NNN; Case~10}
  \label{fig:P-NNN-C10-a}
\end{subfigure}
\begin{subfigure}[b]{.24\textwidth}
  \centering
  \includegraphics[page=24]{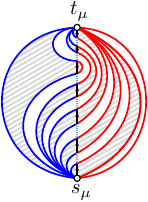}
  \subcaption{\NNN; Case~10}
  \label{fig:P-NNN-C10-b}
\end{subfigure}
\begin{subfigure}[b]{.24\textwidth}
  \centering
  \includegraphics[page=25]{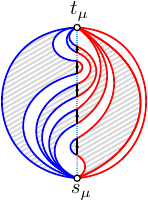}
  \subcaption{\NNN; Case~12}
  \label{fig:P-NNN-C12}
\end{subfigure}
\caption{Constructions for P-nodes of Type \protect\NNN. Subfigures (a) and (b) show the constructions for Case~10 when there exists no Type-\protect\RRR non-Q-node child and when there exists such a child, respectively.}
\end{figure}
In {\bf ({Case}~10)} we proceed as follows. If there exists no Type-\RRR non-Q-node child, we order the children of $\mu$ so that from left to right we have the Type-\QRRR Q-node child (which exists by the conditions of the case), followed by the Type-\RBR children, if any, in any order, followed by the Type-\RBN child, followed by the Type-\RBL children, if any, in any order, followed by the Type-\NLL child (which exists by the conditions of the case); see \cref{fig:P-NNN-C10-a}. Otherwise, we order the children of $\mu$ so that from left to right we have the Type-\RRR non-Q-node child, followed by the Type-\RBR children, if any, in any order, followed by the Type-\RBN child, followed by the Type-\RBL children, if any, in any order, followed by the Type-\NLL child, if any, followed by the Type-\QLLL Q-node child, if any; see \cref{fig:P-NNN-C10-b}.
In {\bf ({Case}~11)} we proceed as in {\bf ({Case}~10)} by considering the types after a vertical flip.
In {\bf ({Case}~12)} we order the children of $\mu$ so that from left to right we have the non-Q-node Type-\RRR child, followed by the Type-\RBR children, if any, in any order, followed by the Type-\LBR children, if any, in any order, followed by the Type-\LLL non-Q-node child, if any, followed by the Type-\QLLL Q-node child, if any; see \cref{fig:P-NNN-C12}.
In {\bf ({Case}~13)} we proceed as in {\bf ({Case}~12)} by considering the types after a vertical flip.
\end{proof}

\subsection{R-nodes}\label{apx:r-node}

Let $\mu$ be an R-node with poles $s_{\mu}$ and $t_{\mu}$.
We consider the graph $\skel^+(\mu) = \skel(\mu) \cup (s_\mu,t_\mu)$. Note that $\skel^+(\mu)$ is a triconnected planar graph. Let $\mathcal{E}_\mu$ be an embedding of $\skel^+(\mu)$ with $(s_\mu,t_\mu)$ on the outer face,
let $\bw$ be the branchwidth of $G$, and let $\langle T, \xi, \Pi \rangle$ be a sphere-cut decomposition of $\skel^+(\mu)$ with embedding $\mathcal{E}_\mu$ of width smaller than or equal to\footnote{The skeletons of the nodes of $\mathcal T$ are minors of $G$ and thus their branchwidth is bounded by the one of $G$.} $\bw$. We root $T$ to the leaf $\rho$ of $T$ such that $\xi(\rho) = (s_\mu,t_\mu)$. Each arc $a=(p_a,c_a)$ of $T$, connecting the parent node $p_a$ with the child node $c_a$ in $T$, is associated with the subgraph $\skel_a$ of $\skel^+(\mu)$ lying in the interior of the noose $O_a$. We denote by $\mypert_a$ the subgraph of $G$ obtained by replacing each virtual edges $e$ of $\skel_a$ with the pertinent graph of the node corresponding to $e$ in the SPQR-tree $\mathcal{T}$ of $G$; refer to \cref{fig:rigid}.

Intuitively, our strategy to compute the embedding types of $\pert{\mu}$ is to visit $T$ bottom-up maintaining a succinct description of size $O(\bw)$ of the properties of the noose $O_a$ of the current $\skel_a$ in a \UBE of $\mypert_a$. When we reach the arc $a^*$ that connects $\rho$ to the node of $T$ whose graph $\skel_{a^*}$ coincides with $\skel(\mu)$, we use the computed properties to determine which embedding types are realizable by $\pert{\mu}$; refer to \cref{fig:rigid-sphere-cut}.

The fixed and the variable embedding setting are treated analogously. In the fixed embedding setting, we only consider the embedding of $\skel^+(\mu)$ in which the embedding of $\skel(\mu)$ is inherited by the (fixed) embedding of $G$. In the variable embedding setting, instead, we consider each of the two embeddings of $\skel^+(\mu)$ obtained by flipping the embedding of $\skel(\mu)$ at its poles.

Let $a$ be an arc of $T$ and let $p_0, p_1, \dots, p_k$ be the sequence of maximal (upward or downward) directed paths traversed in a clockwise visit of the outer face of $\skel_a$. The next lemma show that $k \in O(\beta)$.

\begin{lemma}\label{lem:noose-decomposition} 
For each arc $a \in T$, by clockwise visiting the outer face of $\skel_a$, we traverse at most $O(\bw)$ maximal (upward or downward) directed paths.
\end{lemma}
\begin{proof}
Let $p_{u,v}$ be the path along the outer face of $\skel_a$ connecting two vertices of $\midset(a)$ that are clockwise consecutive in the circular ordering $\pi_a \in \Pi$ of $\midset(a)$. Since, by \cref{obs:spqr-st-graphs}, $\pert{\mu}$ is an $s_\mu t_\mu$-graph, and thus its faces have a single source and a single sink, we have that $p_{u,v}$ is composed of at most three maximal (upward or downward) directed paths. The statement then follows from the fact that $|\midset(a)| \leq \bw$. 
\end{proof}

Let $u$ and $v$ be two clockwise consecutive vertices along the outer face of  $\skel_a$ that are connected by an (upward or downward) directed path $p_i$. Given a \UBE $\langle \pi_a, \sigma_a\rangle$ of $\mypert_a$, we associate with $p_i$ an \emph{outer-visibility triple $t_i = \langle u\_{vis}, spine\_{vis}, v\_{vis} \rangle$ of $p_i$ in $\langle \pi_a, \sigma_a\rangle$}, for $i = 1, \dots, k$, which encodes the information about the visibility of the spine between $u$ and $v$ from the outer face of $\Gamma(\pi_a, \sigma_a)$. In particular, we have:
\begin{itemize}
\item $\mathbf{u\_{vis}=true}$ if and only there is a portion of the spine incident to $u$ and between $u$ and $v$ that is visible from the outer face of $\Gamma(\pi_a, \sigma_a)$.
\item $\mathbf{v\_{vis}=true}$ if and only if there is a portion of the spine incident to $v$ and between $u$ and $v$ that is visible from the outer face of $\Gamma(\pi_a, \sigma_a)$.
\item $\mathbf{spine\_{vis}=true}$ if and only if there is a portion of the spine between $u$ and $v$ that is visible from the outer face of $\Gamma(\pi_a, \sigma_a)$.
\end{itemize}
Observe that when $u\_{vis} = true$ or $v\_{vis} = true$, then necessarily $spine\_{vis} = true$. Hence, there are five possible types for the outer-visibility triples $\langle u\_{vis}, spine\_{vis}, v\_{vis} \rangle$.

We denote by $X_a$ the set of all the endpoints of the paths $p_1,\dots,p_k$.

Let $\langle \pi_\mu, \sigma_\mu\rangle$ be a \UBE of $\pert{\mu}$, let $a$ be an arc of $T$, and let $\langle \pi_a, \sigma_a\rangle$ be the restriction of 
$\langle \pi_\mu, \sigma_\mu\rangle$ to $\mypert_a$. Further, let $\langle \pi'_a, \sigma'_a\rangle \neq \langle \pi_a, \sigma_a\rangle$ be a \UBE of $\mypert_a$. The next lemma can be proved with the same strategy as in the proof of \cref{le:replacement}.

\begin{lemma}\label{lem:replacement-pert-a}
Graph $\pert{\mu}$ admits a \UBE $\langle \pi'_\mu, \sigma'_\mu\rangle$ whose restriction to $\mypert_a$ is $\langle \pi'_a, \sigma'_a\rangle$ if the following two conditions hold: 
\begin{inparaenum}
\item the bottom-to-top order of the vertices in $X_a$ is the same in $\pi_a$ as in $\pi'_a$, and 
\item for each $i=1,\dots,k$, the outer-visibility triple of $p_i$ in $\Gamma(\pi'_a, \sigma'_a)$ is the same as the outer-visibility triple of $p_i$ in $\Gamma(\pi_a, \sigma_a)$.
\end{inparaenum}
\end{lemma}

When visiting $T$ bottom-up, we compute and store in each arc $a$ of $T$ those pairs $(\pi_a, \langle t_1, t_2, \dots, t_k \rangle)$, called \emph{outer-shape pairs}, where $\pi_a$ is a bottom-to-top order of the vertices in $X_a$ and $t_i$ is an outer-visibility triple for path $p_i$, for $i=1,\dots,k$, such that there exists a \UBE $\langle \pi, \sigma \rangle$ of $\mypert_a$ satisfying the following properties: 
\begin{inparaenum}
\item
ordering $\pi_a$ is the restriction of $\pi$ to $X_a$ and 
\item $t_i$ is the outer-visibility triple determined by $\langle \pi, \sigma \rangle$ when clockwise traversing the outer face of $\Gamma(\pi, \sigma)$ between the endpoints of $p_i$.
\end{inparaenum}
By \cref{lem:replacement-pert-a}, this information is enough to succinctly describe all the relevant properties of the~\UBEs~of~$\mypert_a$.
%
Moreover, since $\skel_{a^*}=\skel(\mu)$ is bounded by two directed paths $p_1$ and $p_2$ whose end-vertices are~$s_\mu$~and~$t_\mu$, when we reach the arc $a^*$, all the embedding types that are realizable by $\pert{\mu}$ can be computed by inspecting the outer-visibility triples $t_1$ and $t_2$, over all outer-shape pairs $(\pi_a, \langle t_1, t_2, \dots, t_k \rangle)$ stored in $a^*$.

\smallskip
We show how to compute all outer-shape pairs $(\pi_a, \langle t_1, t_2, \dots, t_k \rangle)$, for~each~arc~$a \in T$.

Suppose $a$ leads to a leaf $\ell$ of $T$ such that $\xi(\ell) = (u,v)$. Assume that $\mypert_a$ is oriented from $u$ to $v$ and that the clockwise boundary of $\skel_a$ is composed by the upward path $p_1 = u,v$ and the downward path $p_2 = v,u$, the other cases being analogous. 
Then, all the outer-shape pairs $(\pi_a, \langle t_1, t_2 \rangle)$  for $a$ are such that $\pi_a = u,v$, since $\pert{\mu_\ell}$ is an $uv$-graph. We compute the pairs of outer-visibility triples $t_1=\langle \alpha, \beta, \gamma \rangle$ and $t_2=\langle \alpha', \beta', \gamma' \rangle$ for $p_1$ and $p_2$, respectively, as follows.
For each embedding type $\langle s\_{vis},spine\_{vis},t\_{vis} \rangle$ of the node $\mu_\ell$ of $\mathcal T$ represented by edge $\ell$ we define $t_1$ and $t_2$ as follows. If $s\_{vis} = L$, we set $\alpha = true$ and $\alpha' = false$. If $s\_{vis} = R$ we set $\alpha = false$ and $\alpha' = true$. If $s\_{vis} = N$ we set both $\alpha = \alpha' = false$. Analogously, we set $\beta$ and $\beta'$ according to the value of $t\_{vis}$. We set $\gamma = true$ if $spine\_{vis} = B$ or $spine\_{vis} = L$. Otherwise, we set $\gamma = false$. Finally, we set $\gamma' = true$ if $spine\_{vis} = B$ or $spine\_{vis} = R$. Otherwise, we set $\gamma' = false$.   

Suppose now that $a$ leads to a non-leaf node $\lambda$ of $T$. Then, $\lambda$ has two children $\lambda_1$ and $\lambda_2$, reached by two arcs $a_1$ and $a_2$ of $T$, for which we have already computed all the outer-shape pairs.
We are going to use the following observation, which follows from the planarity of $\skel(\mu)$ and from the fact that the nooses form a laminar set. 

\begin{observation}\label{obs:properties-skel-rigid}
Let $\langle \pi_\mu, \sigma_\mu \rangle$ be a \UBE of $\skel(\mu)$ and let $\langle \pi_{a_1}, \sigma_{a_1} \rangle$ and $\langle \pi_{a_2}, \sigma_{a_2} \rangle$ be the \UBEs of $\skel_{a_1}$ and $\skel_{a_2}$, respectively, determined by $\langle \pi_\mu, \sigma_\mu \rangle$. The following two properties hold:
\begin{inparaenum}
  \item the edges of $\skel_{a_1}$ lie in the outer face of the drawing obtained by restricting $\Gamma(\pi_\mu,\sigma_\mu)$ to $\skel_{a_2}$, and vice versa, and
  \item the edges of $\skel_{a_1}$ and $\skel_{a_2}$ do not interleave around a vertex in $\Gamma(\pi_\mu,\sigma_\mu)$. 
\end{inparaenum}
\end{observation}

For $i=1,2$, let $A_i=(V_i,E_i)$ be the \emph{auxiliary (multi)-graph} of $\mypert_{a_i}$ defined as follows. Initialize $V_i=X_{a_i}$ and add to $E_i$ a directed edge $(u,v)$ for each path $p$ directed from $u$ to $v$ traversed when clockwise visiting the outer face of $\skel_{a_i}$; see \cref{fig:rigid-c}. Then, replace each directed edge $(u,v)$ with a directed path $(u,x',x'',v)$, unless $(u,v)$ is an edge in $G$. Clearly, graph $A_i$ is connected.
Note that, there exists a one-to-one correspondence between the directed paths along the outer face of $\skel_{a_i}$ (see \cref{fig:rigid-b}) and the directed paths along the outer face of $A_i$ (see \cref{fig:rigid-c}).
We have that each \UBE of $A_i$ defines a set $\langle t_1,\dots,t_{k_i}\rangle$, where $t_i$ is an outer-visibility triple describing the visibility of each directed path along the outer face of $A_i$ in the \UBE. We define a new graph $A$ as the union of $A_1$ and $A_2$; see \cref{fig:rigid-d}. 
The purpose of graph $A$ is that of representing the visibility of the spine on the outer face of $\mypert_a$ between pairs of vertices in $X_a$ in a \UBE of $\mypert_a$. In particular, by assigning to the faces of the book embedding the three edges that replace each directed path $p_i$ along the outer face of $\mypert_a$, we are able to model the outer-visibility triples of paths $p_1, \dots, p_k$.
We compute all outer-shape pairs $\langle \pi_a, \langle t_1,\dots,t_k \rangle \rangle$ for $a$ using \cref{algo:pairs}.

\begin{algorithm}[h!]
\caption{Procedure to compute the outer-shape pairs of an arc $a \in T$.}
\paragraph{\textbf{Step~1}} Construct, by brute force, the set $\mathcal U$ of all the \UBEs $\langle \pi_A, \sigma_A \rangle$ of $A$.

\paragraph{\textbf{Step~2}} Let $\mathcal U^*$ be the subset of the \UBEs $\langle \pi_A, \sigma_A \rangle \in \mathcal U$ such that the drawing of 
$A_1$ lies in the outer face of the drawing of $A_2$ in $\Gamma(\pi_A, \sigma_A)$, and vice versa, and the edges of these graphs do not interleave around a vertex. 

\paragraph{\textbf{Step~3}} For each \UBE $\langle \pi_{A}, \sigma_{A}\rangle \in \mathcal U^*$ of $A$, perform the following operations:
\Los{}{
\paragraph{\textbf{Step~3.a}} 
Let  $\langle \pi_{A_i}, \sigma_{A_i} \rangle$ be the \UBE of $A_i$ obtained by restricting $\langle \pi_A, \sigma_A \rangle$ to $A_i$.
Clockwise visit the outer face of $\Gamma(\pi_{A_i},\sigma_{A_i})$ of $A_i$ to compute the sequence $\langle t^i_1,\dots,t^i_{k_i} \rangle$ of the outer-visibility triples of the maximal directed paths between the vertices~in~$X_{a_i}$.

\paragraph{\textbf{Step~3.b}} Verify whether the set of outer-shape pairs stored in $a_1$ contains
$\langle \pi_{a_1}, \langle t^1_1,\dots,t^1_{k_1} \rangle \rangle$ 
and whether the set of outer-shape pairs stored in $a_2$ contains $\langle \pi_{a_2}, \langle t^2_1,\dots,t^2_{k_2} \rangle \rangle$.

\paragraph{\textbf{Step~3.c}} If the pairs $\langle \pi_{a_1}, \langle t^1_1,\dots,t^1_{k_1} \rangle \rangle$  and $\langle \pi_{a_2}, \langle t^2_1,\dots,t^2_{k_2} \rangle \rangle$ pass the above test, then compute a pair $\langle \pi_a, \langle t_1,\dots,t_k \rangle \rangle$ for $a$ as follows:
\begin{enumerate}
\item[~~~\textbf{Step~3.c.1}] Set $\pi_a$ as the restriction of $\pi_A$ to the vertices in $X_a \subseteq V(A)$;
\item[~~~\textbf{Step~3.c.2}] 
Clockwise visit the outer face of $\Gamma(\pi_{A},\sigma_{A})$ to compute the sequence $\langle t^i_1,\dots,t^i_{k_i} \rangle$ of the outer-visibility triples of the maximal directed paths between \\the vertices~in~$X_{a}$.
\end{enumerate}
}
\label{algo:pairs}
\end{algorithm}

We have the following main lemma.

\ledecision*

\begin{proof}
We exploit \cref{algo:pairs} on the arcs of $T_\mu$ to compute their outer-shape pairs. As already observed, this allows us to compute the embedding types realizable by $\pert{\mu}$ from the outer-shape pairs of the arc $a^*$ incident to the leaf $\nu$ of $T_\mu$ such that $\xi_\mu(\nu)=(s_\mu,t_\mu)$. 
The correctness of \cref{algo:pairs} descends from \cref{lem:constant-types} and \cref{obs:properties-skel-rigid}.

We argue about the running time of \cref{algo:pairs} for a given arc $a$ of $T_\mu$. Note that, the auxiliary graph $A$ corresponding to $a$ contains $O(\beta)$ vertices, by \cref{lem:noose-decomposition}, and $O(\beta)$ edges, since it is planar. 
First, the running time of \textbf{Step~1} is bounded by the number of possible \UBEs~of~$A$. By enumerating all the $|V(A)|! \in O(\beta!) \in 2^{O(\beta\log{\beta})}$ linear orders of $V(A)$ and all the $2^{|E(A)|} \in 2^{O(\beta)}$ edge assignments for $E(A)$, we may construct $2^{O(\beta\log{\beta})}$ \UBEs of $A$.
\textbf{Step~2} is also bounded by the numbers of \UBEs of $A$, as testing whether the conditions of \cref{obs:properties-skel-rigid} are satisfied by each such a \UBE can be performed in time linear in the size of $A$, i.e., $O(\beta)$ time. 
It is clear that, for each of the \UBEs of $A$ considered at \textbf{Step~3}, the remaining steps of \cref{algo:pairs} can also be performed in $O(\beta)$ time.
Thus, \cref{algo:pairs} runs in $2^{O(\beta\log{\beta})}$ time.

The overall running time for computing the embedding types of an R-node $\mu$ follows from the fact that tree $T_\mu$ contains $O(k)$ nodes and arcs, since it is a ternary tree whose leaves are the $O(k)$ virtual edges of $\skel(\mu)$, and from the time spent to compute the outer-shape pairs for each arc of $T_\mu$, i.e., the running time of \cref{algo:pairs}.
\end{proof}